\documentclass[journal]{IEEEtran}

\usepackage[T1]{fontenc}
\usepackage{tikz}
\usetikzlibrary{arrows}
\usetikzlibrary{calc}
\usetikzlibrary{snakes}
\usepackage{cite}
\usepackage{graphicx}
\usepackage{array}
\usepackage{mdwmath}
\usepackage{mdwtab}
\usepackage{eqparbox}
\usepackage{fixltx2e}
\usepackage{url}
\usepackage{pgf}

\usepackage{amsmath,amsfonts,amssymb,amsthm}
\usepackage{xspace}
\usepackage{xcolor}
\usepackage{algorithmic}
\usepackage{bm}
\usepackage{booktabs}
\usepackage{nicefrac}
\usepackage{textcomp}
\usepackage{stmaryrd}
\usepackage{mathrsfs}
\usepackage{paralist}
\usepackage{comment}

\usepackage[font=footnotesize,labelsep=space,labelfont=bf]{caption}
\IEEEoverridecommandlockouts
\usepackage{url}

\usepackage[pdftex,bookmarks=true,pdfstartview=FitH,colorlinks,linkcolor=blue,filecolor=blue,citecolor=blue,urlcolor=blue,pagebackref=true]{hyperref}
\makeatletter

\newcommand{\Rmnum}[1]{\expandafter\@slowromancap\romannumeral #1@}
\makeatother

\newtheorem{thm}{Theorem}
\newtheorem{defn}{Definition}
\newtheorem{lem}{Lemma}

\newtheorem{remark}{Remark}
\newcommand{\namedref}[2]{\hyperref[#2]{#1~\ref*{#2}}}

\newcommand{\WZ}{\text{WZ}}
\newcommand{\PS}{\text{PS}}
\newcommand{\AS}{\text{AS}}
\newcommand{\image}{\text{image}}

\newcommand{\Sectionref}[1]{\namedref{Section}{sec:#1}}
\newcommand{\Subsectionref}[1]{\namedref{Section}{subsec:#1}}

\newcommand{\Appendixref}[1]{\namedref{Appendix}{app:#1}}
\newcommand{\Theoremref}[1]{\namedref{Theorem}{thm:#1}}

\newcommand{\Definitionref}[1]{\namedref{Definition}{defn:#1}}
\newcommand{\Lemmaref}[1]{\namedref{Lemma}{lem:#1}}
\newcommand{\Remarkref}[1]{\namedref{Remark}{remark:#1}}

\newcommand{\Figureref}[1]{\namedref{Figure}{fig:#1}}

\newcommand{\Footnoteref}[1]{\namedref{Footnote}{foot:#1}}
\newcommand{\Pageref}[1]{\hyperref[#1]{page~\pageref*{#1}}}

\definecolor{darkred}{rgb}{0.5, 0, 0} 
\definecolor{darkblue}{rgb}{0,0,0.5} 
\hypersetup{
    colorlinks=true,
    linkcolor=darkred,
    citecolor=darkblue,
    urlcolor=darkblue   
}

\renewcommand{\ij}{\ensuremath{\vec{ij}}\xspace}
\newcommand{\ji}{\ensuremath{\vec{ji}}\xspace}

\newcommand{\RI}{\ensuremath{RI}\xspace}

\newcommand{\X}{\ensuremath{\mathcal{X}}\xspace}
\newcommand{\Y}{\ensuremath{\mathcal{Y}}\xspace}
\newcommand{\Z}{\ensuremath{\mathcal{Z}}\xspace}
\newcommand{\R}{\ensuremath{\mathcal{R}}\xspace}

\renewcommand{\paragraph}[1]{\smallskip\noindent{\bf #1}~}

\begin{document}
\thispagestyle{empty}

\title{Communication and Randomness Lower Bounds for Secure Computation}

\author{Deepesh Data, {\it Student Member, IEEE}, \\ Vinod M. Prabhakaran, {\it Member, IEEE}, and Manoj M. Prabhakaran, {\it Member, IEEE}
\thanks{D. Data's research was fundted in part by a Microsoft Research India Ph.D. Fellowship. 
V. M. Prabhakaran's research was funded in part by a Ramanujan fellowship from the Department of Science and Technology, Government of India and in part by by Information Technology Research Academy (ITRA), Government of India under ITRA-Mobile grant ITRA/15(64)/Mobile/USEAADWN/01. 
M. M. Prabhakaran's research was funded in part by the NSF under grant 12-28856. 
This work was presented in part at the 34th International Cryptology Conference (CRYPTO), 2014 and the IEEE International Symposium on Information Theory (ISIT), 2015.}
\thanks{D. Data and V. M. Prabhakaran are with the School of Technology \& Computer Science, 
Tata Institute of Fundamental Research, Mumbai  400005, India.
Email: \{deepeshd,vinodmp\}@tifr.res.in. M. M. Prabhakaran is with the Department of Computer Science, University of Illinois, Urbana-Champaign, Urbana, IL 61801. Email: mmp@illinois.edu.}
}

\maketitle

\begin{abstract}
In secure multiparty computation (MPC), mutually distrusting users collaborate
to compute a function of their private data without revealing any additional
information about their data to the other users. While it is known that
information theoretically secure MPC is possible among $n$ users having access
to private randomness and are pairwise connected by secure, noiseless, and
bidirectional links against the collusion of less than $n/2$ users (in the {\em
honest-but-curious} model; the threshold is $n/3$ in the {\em malicious
model}), relatively little is known about the communication and randomness
complexity of secure computation, i.e., the amount of communication and
randomness required to compute securely.

In this work, we employ information theoretic techniques to obtain lower bounds
on communication and randomness complexity of secure MPC. We restrict ourselves
to a concrete interactive setting involving three users under which all
functions are securely computable against corruption of individual users in the
honest-but-curious model. We derive lower bounds for both the perfect security
case (i.e., zero-error and no leakage of information) and asymptotic security
(where the probability of error and information leakage vanish as block-length
goes to $\infty$).

Our techniques include the use of a data processing inequality for {\em
residual information} (i.e., the gap between mutual information and
G\'acs-K\"orner common information), a new information inequality for
3-user protocols, and the idea of {\em distribution switching} by which
lower bounds computed under certain worst-case scenarios can be shown to
apply for the general case.

Our lower bounds are shown to be tight for various functions of interest. In
particular, we show concrete functions which have ``communication-ideal''
protocols, i.e., which achieve the minimum communication simultaneously on
all links in the network. Also, we obtain the first {\em explicit} example
of a function that incurs a higher communication cost than the input length,
in the secure computation model of Feige, Kilian, and Naor
\cite{FeigeKiNa94}, who had shown that such functions exist. We also show
that our communication bounds imply tight lower bounds on the amount of
randomness required by MPC protocols for many interesting functions.

\end{abstract}

\section{Introduction}\label{sec:intro}
Secure multiparty computation (MPC) allows mutually distrusting users to collaborate in computational tasks. In particular, it allows such parties to compute a function of their private data without revealing any additional information about their data to the other users. This can be trivially achieved in the presence of a trusted central server as all the users can send their private data to the server, which  performs the computation and sends back the function value to all the users. 
The goal of secure MPC -- often referred to simply as MPC -- is to emulate this trusted server by a {\em protocol} in which the users communicate among themselves, and learn nothing but what they would have learned by interacting with the trusted server.
MPC has several important potential applications including privacy-preserving data mining, secure auction, secure machine learning, secure benchmarking (see, e.g., \cite{CramerDaNi15}). 

MPC was pioneered by the seminal works of Shamir, Rivest, and
Adleman~\cite{ShamirRiAd81}, Rabin~\cite{Rabin81}, Blum~\cite{Blum81},
Yao~\cite{Yao82mpc,Yao86}, Goldreich, Micali, and
Wigderson~\cite{GoldreichMiWi87}, and others.  All of these early results
were based on computational limitations of adversaries and some
cryptographic assumptions, such as the existence of one-way functions,
hardness of factoring large integers, etc.  In a remarkable result, Ben-Or,
Goldwasser, and Wigderson \cite{BenorGoWi88}, and independently Chaum,
Cr\'epeau, and Damg\r{a}rd~\cite{ChaumCrDa88}, showed that information
theoretically MPC is possible among $n$ users having access to private
randomness and are pairwise connected by private, noiseless, and
bidirectional links against the collusion of at most
$\lfloor\frac{n-1}{2}\rfloor$ users in the {\em honest-but-curious} model
and against the collusion of at most $\lfloor\frac{n-1}{3}\rfloor$ users in
the {\em malicious model}.  In the honest-but-curious model, users follow
the protocol honestly, but they retain all the messages exchanged during the
entire execution of the protocol, and in the end, the colluders pool their
information (data, private randomness, and messages) together and try to
find additional information about other users' data. This is also referred
to as the semi-honest model or a model with passive corruption.  In the
malicious model, dishonest users may arbitrarily deviate from the protocol.
These thresholds are known to be tight, in the sense that there are
functions which cannot be securely computed if these thresholds are
exceeded.  There is another line of work on information theoretically secure
computation which relies on stochastic resources, such as noisy channel or
distributed sources; and there, secure computation is possible even if these
thresholds are exceeded~\cite{CrepeauKi88}. In this work we focus on the
model of \cite{BenorGoWi88} where such resources are unavailable.  

Communication is a critical resource in any distributed computation.
Communication complexity of multi-party computation without any security
requirements has been widely studied since \cite{Yao79} (see
\cite{KushilevitzNi97book}), and more recently has seen the use of
information-theoretic tools as well, in~\cite{ChakrabartiShWiYa01} and
subsequent works. Independently, in the information theory literature,
communication requirements of interactive function computation have been
studied (e.g.~\cite{OrlitskyRoche,MaIs13}). In secure distributed computation, in
addition to communication, private randomness is also a crucial resource; it
is known that secure computation of nontrivial functions is not possible
deterministically~\cite{KushilevitzRo98}. However, relatively less is known
about the {\em lower bounds} on the amount of {\em communication} and {\em
randomness} required by secure computation protocols, with a few notable
exceptions
\cite{Kushilevitz92},\cite{FranklinYu92},\cite{ChorKu93},
\cite{FeigeKiNa94},\cite{KushilevitzMa97},\cite{KushilevitzRo98},
\cite{BlundoSaPeVa99},\cite{GalRo05},
which provide lower bounds for very specific functions (mostly for
modular-addition).  
For 2-user secure computation, Kushilevitz~\cite{Kushilevitz92} combinatorially characterized the communication complexity of securely computable functions with security against passive corruption of a single user.  For $n$-user secure protocols, Franklin and Yung~\cite{FranklinYu92} proved an $\Omega(n^2)$ lower bound on the {\em number of messages} exchanged for {\sc xor} function if the security is required against the corruption of $t=\Omega(n)$ users, which matches the upper bound (up to constants); if a stronger corruption model (fail-stop corruption) is assumed, then \cite{FranklinYu92} showed matching {\em amortized} upper and lower bounds for modular-addition for $t=\Omega(n)$, which implies that parallelization does not help in this stronger model.  Further, Chor and Kushilevitz~\cite{ChorKu93} gave tight lower and upper bounds (exact constants) on the number of messages exchanged for modular-addition against corruption of $t\leq n-2$ users.  For a restricted class of 3-user secure protocols, Feige, Kilian, and Naor~\cite{FeigeKiNa94}, along with positive results, obtained a modest communication lower bound (more than the input length) for random Boolean functions.  For secure computation of {\sc xor}, G\'al and Ros\'en~\cite{GalRo05} proved an $\Omega(\log n)$ lower bound on the amount of randomness required, which matches the upper bound $O(t^2\log(n/t))$ of Kushilevitz and Mansour~\cite{KushilevitzMa97} for any fixed $t$.  Kushilevitz and Ros\'en~\cite{KushilevitzRo98} studied the tradeoff between randomness and number of rounds required  in $n$-user secure computation of Boolean functions.  Using information theoretic tools, Blundo et al.~\cite{BlundoSaPeVa99} proved matching (exact constants) lower and upper bounds on the randomness complexity for modular-addition against the corruption of any $t=n-2$ users.

Obtaining strong lower bounds for communication and randomness in
information-theoretically secure MPC protocols is considered difficult, as
it has implications to other long-standing open problems in theoretical
computer science. In particular, Ishai and Kushilevitz~\cite{IshaiKu04}
showed that establishing strong MPC communication lower bounds (even with
restrictions on the number of rounds) would imply breakthrough lower bound
results for well-studied problems like private information retrieval and
locally decodable codes. 
Further, due to the protocols of \cite{BenorGoWi88,ChaumCrDa88}, lower
bounds for MPC communication (with a constant number of players) that are
super-linear in the input size would imply super-linear lower bounds for
circuit complexity -- a notoriously hard problem.  The protocol of
Damg\r{a}rd and Ishai~\cite{DamgardIs06} extended this to a non-constant
number of players.  Kushilevitz et al.~\cite{KushilevitzOsRo96} showed that
this relation to circuit size holds even for MPC protocols that use only a
constant number of random bits, if security is required only against
semi-honest corruption of a single player.
One of the goals of this work is to develop tools to tackle the difficult problem of
establishing lower bounds for communication and randomness in MPC, even if
they do not have immediate implications to circuit complexity,
private information retrieval, or locally decodable codes.

In this work, we also consider a relaxed notion of security for MPC --
namely {\em asymptotic security}. In the standard cryptographic definitions,
security is required for every input. Also, often the security is required
to be ``perfect'' in that the computation is always correct and there is
no information leaked about the inputs beyond the output.\footnote{While in this paper we mostly focus on perfect security, the more
general notion of {\em statistical security} -- which allows the error in
security to be ``negligible'' as a function of a security parameter (i.e.,
given any polynomial in the security parameter, eventually the error becomes
less than its reciprocal) -- is similar in that it also requires security
for every input and there is no distribution over inputs.}

In contrast, in asymptotic security, users are given many independent copies
of the inputs, and we allow the error in security (probability of error in
the outcome and the leakage of information) to be ``vanishing'' as a
function of the number of copies (i.e., eventually, the error becomes less
than any given positive constant).  While asymptotically correct interactive
function computation without any privacy requirement has been
investigated~\cite{OrlitskyRoche},\cite{MaIs11},\cite{MaIsGu12}, as far as
we know, there is very little work on asymptotically secure
computation~\cite{LeeAb14}. Lee and Abbe~\cite{LeeAb14} considered the
communication requirements for asymptotically secure computation under a
restrictive model of protocols in which  no private randomness is available
to the users.
In this work we provide communication and randomness lower bounds for
asymptotically secure computation of any function with arbitrary input
distribution.  We also establish a gap between the communication
requirements (and also randomness requirements) under asymptotic security
and under perfect security by studying the modular addition function.

It is instructive to compare the problem of communication complexity lower
bounds for secure multi-party computation with that when there is no
security requirement involved. This latter problem has been extensively
studied --- over the last three and a half decades, starting with
\cite{Yao79} --- resulting in a rich collection of results and techniques.
Unfortunately, many of the techniques in the communication complexity
setting are not relevant in the setting of secure computation:\footnote{Of course, communication complexity lower bounds continue to hold
for secure computation as well, but these bounds as such are (apparently)
very loose. There is a trivial upper bound for communication complexity,
which is at most the size of all inputs and outputs. This is often
insufficient for secure computation~\cite{FeigeKiNa94}; also see
Section~\ref{subsec:examples}.1.}
for instance, for communication complexity, Yao's minimax theorem allows one
to consider only deterministic protocols with public randomness, but in the
secure computation setting, one must allow private randomness, and hence it
is not sufficient to consider only deterministic protocols. This rules out
several powerful combinatorial approaches from the communication complexity
literature. But over the last decade or so (see for example, \cite{KerenidisLLRX12} and
references therein), several information theoretic tools have been
developed, which in many cases subsume more complicated combinatorial
approaches. Information-theoretic techniques have also been successfully used in deriving bounds in various 
cryptographic problems like key agreement (e.g.\ \cite{MaurerWo03}), secure two-party 
computation (e.g.\ \cite{DodisMi99}), and secret-sharing and its variants 
(e.g.\ \cite{BeimelOr11} and \cite{BlundoSaCrGaVa94}). 
Following this lead, the approach we take in this work is to develop an information-theoretic approach to obtain communication and randomness lower bounds for secure computation. 
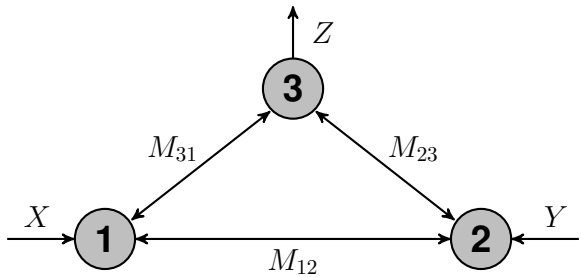
\begin{figure}[tb]
\centering
\begin{tikzpicture}[>=stealth', font=\sffamily\Large\bfseries, thick]
\draw [fill=lightgray] (-2.5,0) circle [radius=0.4]; \node at (-2.5,0) {1};
\draw [fill=lightgray] (2.5,0) circle [radius=0.4]; \node at (2.5,0) {2};
\draw [fill=lightgray] (0,2) circle [radius=0.4]; \node at (0,2) {3};

\draw [<->] (-2.1,0) -- (2.1,0); \node [scale=0.8] at (0,-0.3) {$M_{12}$};
\draw [<->] (-2.15,0.25) -- (-0.3,1.7); \node [scale=0.8] at (-1.6,1.2) {$M_{31}$};
\draw [<->] (2.15,0.25) -- (0.3,1.7); \node [scale=0.8] at (1.6,1.2) {$M_{23}$};

\draw [->] (-3.8,0) -- (-2.9,0); \node [scale=0.8] at (-3.4,0.3) {${X}$};
\draw [->] (3.8,0) -- (2.9,0); \node [scale=0.8] at (3.5,0.3) {${Y}$};
\draw [->] (0,2.4) -- (0,3.1); \node [scale=0.8] at (0.4,2.75) {${Z}$};

\end{tikzpicture}
\caption{A 3-user secure computation problem. Alice (user-1) has input $X$ and Bob (user-2) has $Y$. We require that (i) Charlie (user-3) obtains an output $Z$, where $Z$ is a (possibly randomized) function of the other two 
users' inputs, (ii) Alice and Bob learn no additional information about each other's inputs and the output, and (iii) Charlie learns nothing more about $X,Y$ than what is revealed by $Z$. All users can talk to each other, over multiple rounds over bidirectional pairwise private links.} 
\label{fig:setup}
\end{figure}
\subsection{Results and Techniques}
In this work we restrict our study to a concrete setting involving three
users (with security against passive corruption of any single user), where
two users, Alice and Bob have inputs, $X$ and $Y$, and only the third user 
Charlie produces an
output $Z$ as a (possibly randomized) function of the inputs; see
\Figureref{setup}. This is arguably the simplest setting of~\cite{BenorGoWi88} where
all functions can be securely evaluated, against passive corruption of a
single user. Indeed, the functions considered in this model are
the same as in the model of \cite{FeigeKiNa94}; however, we allow fully
interactive communication between all three users (as
in~\cite{BenorGoWi88}), whereas in the model of \cite{FeigeKiNa94}, there is
no other communication except a single message each from Alice and Bob to Charlie, and
a random string shared between Alice and Bob. Since we allow more general
protocols, it is harder to establish lower bounds in our model. We obtain
lower bounds on the entropy of the transcript between each pair of users
which will imply lower bounds on the expected number of bits exchanged by
these users. We also obtain lower bounds on the amount of randomness needed
for secure computation.

At a high-level, the main ingredients in deriving our lower bounds for perfectly secure computation are the following:
\begin{itemize}
\item Firstly, we observe (in \Lemmaref{general_cutset}) that, since Alice
and Bob do not obtain any outputs and therefore must not learn any
additional information about each other's inputs, they are both forced to
reveal their inputs fully (up to equivalent inputs) to the rest of the
system, and further, Charlie's output depends on the inputs only through all
the communication he has with the rest of the system. Combined with the
privacy requirements, one can immediately obtain na\"ive lower bounds on the
entropies of the transcripts: specifically, writing $X,Y,Z$ as
$X_1,X_2,X_3$, we have $H(M_{ij})\geq H(X_i,X_j|X_k)$, where
$\{i,j,k\}=\{1,2,3\}$.\footnote{\label{foot:addition} We point out a simple example for which one
can obtain a tight bound from this na\"ive bound: addition (in any group)
requires one group element to be communicated between every pair of players,
even with amortization over several independent instances. Previous lower
bounds for secure evaluation of addition \cite{FranklinYu92,ChorKu93}, while
considering an arbitrary number of users, either restricted themselves to
bounding the {\em number of messages} required, or relied on non-standard
security requirements. In comparison, for the 3-user case, for semi-honest security,
results of \cite{FranklinYu92,ChorKu93} only imply that all three links
should be used. \cite{FranklinYu92} did give a lower bound on the number of
bits communicated as well, but this was shown only under a non-standard
security requirement called {\em unstoppability}.}

We strengthen the na\"ive lower bounds by relying on a ``secure
data-processing inequality'' (\Lemmaref{monotone} due to Wolf and
Wullschleger~\cite{WolfWu08} and generalized in~\cite{PrabhakaranPr14}) for
{\em residual information} --- i.e., the gap between mutual-information and
(G\'{a}cs-K\"{o}rner) common information --- which lets us relate the
residual information of real-world views of a pair of users to the residual
information of their ideal-world views.\\

\item We can improve the lower bounds by exploiting the fact that, in a
protocol, the transcripts have to be generated by the users interactively,
rather than be created by an omniscient ``dealer.'' (We formalize the latter
notion of secure transcripts generated by a dealer as {\em Correlated
Multi-Secret Sharing Schemes}.) A technical contribution of this work is a
{\em new information inequality for 3-user protocols}~(\Lemmaref{infoineq}),
which serves as a tool to separate the transcript generated by a secure
protocol from one generated by a dealer.\\

\item Our final tool, that is used to significantly improve the above lower
bounds, is called {\em distribution switching}. The key idea is that the
security requirement forces the distribution of the transcript on certain
links to be independent of certain inputs. Hence we can optimize our bound
using an appropriate distribution of inputs. In fact, we can take the
different terms in our bound and {\em optimize each of them separately
using different distributions over the inputs}. The resulting bounds are
often stronger than what can be obtained by considering a single input
distribution for the entire expression. Further, this shows that even if
the protocol is allowed to depend on the input distribution, our bounds
(which depend only on the function being evaluated) hold for every input
distribution that has full support over the input domain.
\end{itemize}

For asymptotically secure computation, we show that for the same secure
computation problems, the protocols for asymptotic security can provably be
more communication efficient than the protocols for perfect security.
Hence, the lower bounds derived for perfect security do not hold for
asymptotic security. In deriving the lower bounds for asymptotic security,
we use some of the basic ideas (cut-set, secure data-processing inequality,
information inequality) from the bounds for perfectly secure computation,
but the lower bound proofs here are slightly more involved. For instance,
in any perfectly secure protocol, the information about a user's input must
flow out through the links she/he is part of (\Lemmaref{general_cutset});
but this is not true for asymptotically secure computation
-- in fact, in some examples we show that our optimal
protocol for asymptotically secure computation does not require users to
reveal their inputs. The analogous (\Lemmaref{cutset}) turns out to be more
involved.

While we restrict our attention to a 3-user setting, to the best of our
knowledge, our lower bounds (for perfectly and asymptotically secure
computations) are the first {\em generic} lower bounds which apply to
any function. To illustrate their use, we apply them to several interesting
example functions. In particular, we show the following:
\begin{itemize}
\item For several functions 
we prove that there are secure protocols which achieve {\em optimal
communication complexity simultaneously on each link}. We call such a
protocol a {\em communication-ideal} protocol.\\ 

\item We show an {\em explicit} deterministic function $f:\{0,1\}^n \times
\{0,1\}^n \rightarrow \{0,1\}^{n-1}$, which has a communication-ideal
protocol in which Charlie's total communication cost is (and must be at
least) $3n-1$ bits. In contrast, \cite{FeigeKiNa94} showed that {\em there
exist} functions $f:\{0,1\}^n \times \{0,1\}^n \rightarrow \{0,1\}$, for
which Charlie must receive at least $3n-4$ bits, where the protocol is
required to be in their non-interactive model. (Note that our bound is
incomparable to that of \cite{FeigeKiNa94} since we require the output of
our function to be longer; on the other hand, our bound uses an explicit
function and continues to hold even if we allow unrestricted interaction.)\\

\item Our lower bounds for communication complexity also yield lower bounds
on the amount of randomness needed in secure computation protocols. We
analyze secure protocols for several functions
and prove that these protocols are {\em randomness-optimal},
i.e., they use the least amount of randomness.\\

\item We also use our lower bounds to establish a separation between secret
sharing and secure computation: we show that there exists a function (in
fact, the {\sc and} function) which has a secret sharing scheme with a
share strictly smaller than the number of bits in the transcript on the
corresponding link in any secure computation protocol for that function.
While such a separation is natural to expect, we note that proving it
requires exploiting the properties of an interactive protocol.\\

\item For asymptotically secure computation: we analyze asymptotically
secure protocols for some functions 
and show that, under independent input
distribution, these protocols are communication-ideal as well as
randomness-optimal.
\end{itemize}

\subsection{Outline of the Paper}
We discuss the problem setup and some preliminaries in \Sectionref{prelims}.
In \Sectionref{ps_lowerbounds}, we prove our lower bound results for perfectly secure computation; this is also the setting for classical positive results like that of \cite{BenorGoWi88}. Secure protocols for some functions of interest are given, and our bounds are analyzed for those functions. In \Sectionref{asymp_lowerbounds}, we prove our lower bound results for asymptotically secure computation and apply them to a few functions. In \Sectionref{conclusion}, we conclude and give some open problems.

The lower bounds derived in this paper are for the honest-but-curious model against passive corruption of a single user. Typically these bounds continue to hold for active corruption as well -- for many functionalities, every protocol secure against active corruption is a protocol secure against passive corruption.
\section{Preliminaries}\label{sec:prelims}

\paragraph{Notation.}
We write $p_X$ to denote the distribution of a discrete random variable
$X$; $p_X(x)$ denotes $\Pr[X=x]$. When clear from the context, the
subscript of $p_X$ will be omitted.
The conditional distribution denoted by $p_{Z|U}$ specifies
$\Pr[Z=z|U=u]$,
for each value $z$ that $Z$ can take and each
value $u$ that $U$ can take.
A {\em randomized function} of two variables,  is specified by
a probability distribution $p_{Z|XY}$, where $X,Y$ denote the two input
variables, and $Z$ denotes the output variable.
For a sequence of random variables $X_1,X_2,\ldots,$ we denote by $X^n$ the vector $(X_1,\ldots,X_n)$. We abbreviate independent and identically distributed by i.i.d.

For random variables $T,U,V,$ we write the {\em Markov chain} $T-U-V$ to
indicate that $T$ and $V$ are conditionally independent conditioned on $U$.  All logarithms are to the base 2. The binary entropy function
is denoted by $H_2(p)=-p\log p -(1-p)\log(1-p),\;p\in(0,1).$

\paragraph{Problem Definition.}
We consider three user computation functionalities, in which Alice and Bob (users 1 and 2) receive as inputs blocks of random variables $X^n\in{\X^n}$ and $Y^n\in{\Y^n}$, respectively, where $(X_i,Y_i)\sim p_{XY}$, i.i.d., and Charlie (user 3) wants to produce an output $Z^n\in{\Z^n}$, where $Z_i$'s are distributed according to a specified distribution $p_{Z|XY}$.
In particular, we can consider a {\em deterministic function evaluation} functionality where $p_{Z|XY}(z|x,y)=1_{z=f(x,y)}$ for some function $f:{\X}\times{\Y}\to{\Z}$. The set ${\X}$, ${\Y}$, and ${\Z}$ are always finite. We assume that every pair of users is connected by a noiseless, bidirectional link, which is secure from the other user, i.e., the other user cannot read or tamper with any message sent on that link. All the users have access to private randomness, which is independent between the users and also independent of their inputs. We study the secure computation problem in two settings: perfect security and asymptotically perfect security. Below, we consider protocols which can depend not only on $p_{Z|XY}$, but also $p_{XY}$ (a protocol is formally defined later in this section); note that since we are interested in establishing lower bounds, this strengthens our results.

\begin{enumerate}
\item {\bf Perfectly secure computation:} A perfectly secure computation protocol $\Pi(p_{XY},p_{Z|XY})$ satisfies the following conditions:
\begin{itemize}
\item {\em Correctness:} Charlie's output $Z^n$ should be distributed according to $p(z^n|x^n,y^n) = \Pi_{i=1}^n p_{Z|XY}(z_i|x_i,y_i)$, where $x^n$ and $y^n$ are the inputs to Alice and Bob, respectively.
\item {\em Privacy:} Corresponding to privacy against Alice, Bob, and Charlie, respectively, we have the following three conditions:
\begin{align*}
I(M_{12},M_{31};Y^n,Z^n|X^n)&=0, \\
I(M_{12},M_{23};X^n,Z^n|Y^n)&=0, \\
I(M_{23},M_{31};X^n,Y^n|Z^n)&=0,
\end{align*}
where $M_{ij}$ is the collection of all the messages exchanged between users $i$ and $j$ on the $ij$ link in either direction during the entire execution of the protocol (a formal definition is given later, along with the definition of a protocol).
\end{itemize}
Intuitively, the privacy conditions guarantee that even if one user, say
Alice, is curious and retains her view (i.e., her input and all the messages exchanged during the entire execution of the protocol), this view reveals nothing more to her about the input and output of the
other users (namely, $Y^n,Z^n$), than what her own input/output (namely, $X^n$) reveals. In other words, a
curious user may as well simulate a view for herself based on just its
input and output rather than retain the actual view it obtained from the
protocol execution. A more formal definition of perfectly secure protocols is given in \Definitionref{ps_secure-protocol} in \Sectionref{ps_lowerbounds}.\\

\item {\bf Asymptotically secure computation:} For asymptotically secure computation, for simplicity, in this paper we restrict ourselves to deterministic functions $f:\X\times\Y\to\Z$. A sequence of asymptotically secure protocols $\Pi_n(f,p_{XY})$ satisfies the following conditions:
\begin{itemize}
\item {\em Correctness:} Charlie's output $\hat{Z}^n$ should be close to the true output $Z^n$, where $Z_i=f(X_i,Y_i)$, $i=1,2,\hdots,n$, in the sense that $\Pr\{\hat{Z}^n\neq Z^n\}\to0$ as $n\to\infty$, where probability is taken over the randomness of the input distribution and the protocol.
\item {\em Privacy:} Corresponding to privacy against Alice, Bob, and Charlie, respectively, we have the following three conditions as $n\to\infty$:
\begin{align*}
I(M_{12},M_{31};Y^n,Z^n|X^n) &\to 0, \\
I(M_{12},M_{23};X^n,Z^n|Y^n) &\to 0, \\
I(M_{23},M_{31};X^n,Y^n|Z^n) &\to 0.
\end{align*}
\end{itemize}
Intuitively, the privacy conditions guarantee that, from the protocol, any one user does not obtain non-negligible additional information about other users' inputs and output (if any). A more formal definition of asymptotically secure protocols is given in \Definitionref{asymp_secure-protocol} in \Sectionref{asymp_lowerbounds}.
\end{enumerate}

\paragraph{A Normal Form for $(p_{XY},p_{Z|XY})$.}
For a pair $(p_{XY},p_{Z|XY})$, define the relations $x\cong x'$, $y\cong y'$, and $z\cong z'$ as follows.
\begin{enumerate}
\item For any $x,x'\in\X$, let $\mathcal{S}_{x,x'}=\{y\in\Y : p_{XY}(x,y)>0, p_{XY}(x',y)>0\}$.
We say that $x\cong x'$, if $\forall y\in\mathcal{S}_{x,x'}$ and $z\in\Z$, 
we have $p_{Z|XY}(z|x,y)=p_{Z|XY}(z|x',y)$.
\item For any $y,y'\in\Y$, let $\mathcal{S}_{y,y'}=\{x\in\X : p_{XY}(x,y)>0, p_{XY}(x,y')>0\}$.
We say that $y\cong y'$, if $\forall x\in\mathcal{S}_{y,y'}$ and $z\in\Z$, 
we have $p_{Z|XY}(z|x,y)=p_{Z|XY}(z|x,y')$.
\item Let $\mathcal{S}=\{(x,y) : p_{XY}(x,y)>0\}$. For any $z,z'\in\Z$, we say that 
$z\cong z'$, if $\exists c\geq 0$ such that $\forall (x,y)\in\mathcal{S}$, we have 
$p_{Z|XY}(z|x,y)=c\cdot p_{Z|XY}(z'|x,y)$.
\end{enumerate}
A pair $(p_{XY},p_{Z|XY})$ is said to be in {\em normal form} if $x\cong x'\Rightarrow x=x'$, 
$y\cong y'\Rightarrow y=y'$, and $z\cong z'\Rightarrow z=z'$.

In the paper we mostly deal with $p_{XY}$ having full support. If $p_{XY}$ has full support, 
then the above definition reduces to the following definition of normal form.

\paragraph{A Normal Form for Functionality $p_{Z|XY}$ (for $p_{XY}$ with full support).} 
For a randomized functionality $p_{Z|XY}$, we define the relation  $x\equiv x'$ for $x,x'\in\X$
 to hold if $\forall y\in\Y, z\in\Z$, $p(z|x,y) = p(z|x',y)$; similarly we define $y\equiv y'$.
 For $z,z'\in\Z$, we define $z\equiv z'$ if there exists a constant $c$ such that 
$\forall x\in\X,y\in\Y$, $p(z|x,y)= c\cdot p(z'|x,y)$. 
We say that $p_{Z|XY}$ is in {\em normal form} if $x\equiv x' \Rightarrow x=x'$,
$y\equiv y' \Rightarrow y=y'$, and $z\equiv z' \Rightarrow z=z'$.

Note that if $p_{Z|XY}$ is a deterministic mapping $f:\X\times\Y\to\Z$, then $x\equiv x'$ for $x,x'\in\X$ implies that $\forall y\in\Y$, $f(x,y)=f(x',y)$; similarly $y\equiv y'$ is defined. We say that $f$ is in normal form if $x\equiv x' \Rightarrow x=x'$ and $y\equiv y' \Rightarrow y=y'$.
 
It is easy to see that if $p_{XY}$ has full support then one can transform any randomized function $p_{Z|XY}$ to one in normal form $p_{Z^*|X^*Y^*}$ with possibly smaller alphabets, so that any secure computation protocol for the former can be transformed to one for the latter with the same communication costs, and vice versa.
To define $X^*$, $\X$ is modified by replacing all $x$ in an equivalence class of $\equiv$ with a single representative; $Y^*$ and $Z^*$ are defined similarly.
The modification to the protocol, in either direction, is for each user to locally map $X$ to $X^*$ etc., or vice versa; notice that the $Z^*$ to $Z$ map is potentially randomized.

\paragraph{Protocols.}
Given inputs $X^n$ and $Y^n$ to Alice and Bob, respectively, all the users engage in a protocol $\Pi_n$ where they send messages to each other over several rounds, and at the end Charlie produces the output. We will omit the subscript $n$ where it is clear from the context.
A protocol consists of ``next message functions'' $(\Pi^1,\Pi^2,\Pi^3)$ and an output function $\Pi^{3,\text{out}}$. The next message function $\Pi^i, i=1,2,3$ specifies a distribution over $\mathbb{N}\times\{0,1\}^*\times \{0,1\}^*$ (which corresponds to the number of round and the messages on the two links to which user-$i$ is associated) conditioned on the input of user-$i$ (if any) and all the messages on these two links so far. The output function $\Pi^{3,\text{out}}$ defines the output of user-3 as a probabilistic function of all the messages it has seen so far. Specifically, it is a distribution over ${\mathcal Z}^n$ conditioned on all the messages on the links of Charlie.
We allow protocols to depend on the distribution of inputs to the users which would allow one to tune a protocol to be efficient for a given input distribution.
We require that a valid protocol must terminate with probability 1, i.e., on each link, the (potentially random) number of rounds after which the link remains unused must be finite with probability 1.
We denote by $M_{\vec{ij},t}$, the message sent from user $i$ to user $j$ during the round $t$ and by $M_{\vec{ij}}^{t-1}$, all the messages sent by user $i$ to user $j$ up to round $t-1$.
Let $M_{ij}^{t-1}=(M_{\vec{ij}}^{t-1},M_{\vec{ji}}^{t-1})$ denote all the message exchanged between user $i$ and user $j$ up to round $t-1$.
We denote by $M_{ij}$, the final transcript on link $ij$, which is the collection of all the messages exchanged between users $i$ and $j$ during the entire execution of the protocol.
The message $M_{\vec{ij},t}$ (which may be an empty string) is a function of user $i$'s input (if any), $M_{ij}^{t-1}$, and its private randomness.
Furthermore, we restrict the message $M_{\vec{ij},t}$ to be a codeword of a (potentially random) prefix-free binary code $\mathcal{C}_{\vec{ij},t}$, which itself can be determined (with probability 1) by the messages $M_{ij}^{t-1}$ exchanged between users $i$ and $j$.
Although this restricts the kind of protocols we allow, (e.g., our definition does not allow $\mathcal{C}_{\vec{12},t}$ to be determined by messages exchanged between users 1 and 2 via user-3 and not over the 12 link), this encompasses a fairly general class of protocols.
Having the prefix-free requirement and also that the code be determined by previously exchanged messages allow the participating nodes to know when each message and the exchange over the link connecting them has come to an end without the need for an explicit end-of-message symbol.
While we allow this generality, the protocols we provide have a deterministic number of rounds with deterministic message lengths. The generality is in order to prove impossibility results (communication and randomness lower bounds) with wide applicability.

We define $M_1=(M_{12},M_{31})$ as the transcripts that user 1 can see; $M_2$ and $M_3$ are defined similarly. We define the view of the $i^{\text{th}}$ user, $V_i$ to consist of $M_i$ and that user's input and output (if any).
Observe that a protocol along with an input distribution fully defines the joint distribution over all the inputs, outputs, and the joint transcripts on all the links.

\paragraph{Expected Number of Bits Exchanged and Entropy.} As mentioned earlier, we require that the
message sent at every round is a codeword in a prefix-free binary code which can be dynamically determined based on the previous messages exchanged over the link. This allows us to lower-bound the expected number of bits communicated in each link by the entropy of the transcript in that link.

Let $L_{\vec{ij},t} \in \{0,1,2,\hdots\}$ be the (potentially random) length of the message $M_{\vec{ij},t}$. Similarly, let $L_{ij,t}, L_{ij}^t$, and $L_{ij}$  be the lengths of $M_{ij,t}, M_{ij}^t$, and $M_{ij}$, respectively. For a protocol $\Pi_n$, we define the rate quadruple $(R_{12},R_{23},R_{31},\rho)$ as $R_{ij} := \frac{1}{n}\mathbb{E}[L_{ij}]$, $i,j=1,2,3$, $i\neq j$, and $\rho := \frac{1}{n}H(M_{12},M_{23},M_{31},Z^n|X^n,Y^n)$.

We are interested in lower bounds for ${\mathbb E}[L_{ij}]$. We have
\begin{align*}
H(M_{ij}) &= \sum_{t=1}^\infty H(M_{\ij,t},M_{\ji,t}|M_{ij}^{t-1})\\
&\leq \sum_{t=1}^\infty H(M_{\ij,t}|M_{ij}^{t-1}) + H(M_{\ji,t}|M_{ij}^{t-1})\\
&\stackrel{\text{(a)}}{=} \sum_{t=1}^\infty H(M_{\ij,t}|M_{ij}^{t-1},\mathcal{C}_{\ij,t}) + H(M_{\ji,t}|M_{ij}^{t-1},\mathcal{C}_{\ji,t})\\
&\leq \sum_{t=1}^\infty H(M_{\ij,t}|\mathcal{C}_{\ij,t}) + H(M_{\ji,t}|\mathcal{C}_{\ji,t})\\
& \stackrel{\text{(b)}}{\leq} \sum_{t=1}^\infty {\mathbb E}[L_{\ij,t}] +{\mathbb E}[L_{\ji,t}]\\
&= {\mathbb E}[L_{ij}],
\end{align*}
where (a) follows from the fact that the prefix-free codes $\mathcal{C}_{\ij,t},\mathcal{C}_{\ji,t}$, of which $M_{\ij,t},M_{\ji,t}$ are codewords, respectively, are functions of $M_{ij}^{t-1}$. (b) follows from the fact that the expected length $L$ of a prefix-free binary code for a random variable $U$ is lower-bounded by its entropy $H(U)$ \cite[Theorem 5.3.1]{CoverThomas06}.

\paragraph{Conditional Graph Entropy.}
Given a graph $G=(V,E)$, where $V$ is a finite collection of nodes and $E$ is a collection of pairs of vertices from $V$. 
A subset $U\subseteq V$ of $G$ is called an {\em independent set}  of $G$ if no two vertices of $U$ have an edge (an edge is a pair of distinct vertices) between them in $G$. 
Let $\varGamma(G)$ denote the collection of all independent sets of $G$.

Witsenhausen \cite{Witsenhausen76} defined the {\em characteristic graph} $G_X=(V,E)$ for a 
pair $(p_{XY},f)$, where $f:\X\times\Y\to\Z$ is a deterministic function, as follows: 
its vertex set is the support set of $X$, and $E=\{\{x,x'\}:\exists y\in\Y \text{ such that } 
p_{XY}(x,y)\cdot p_{XY}(x',y)>0 \text{ and } f(x,y)\neq f(x',y)\}$. $G_Y$ can be defined similarly.
\begin{defn}[Conditional Graph Entropy \cite{OrlitskyRoche}]\label{defn:conditional-graph-entropy}
For a given pair $(p_{XY},f)$, the conditional graph entropy of $G_X$ is defined as follows:
\begin{align}
H_{G_X}(X|Y) \quad := \displaystyle \min_{\substack{p_{W|X}: \\ W-X-Y \\ X\in W}} I(W;X|Y), \label{eq:conditional-graph-entropy}
\end{align}
\end{defn}
\noindent where the alphabet of $W$ is $\varGamma(G_X)$ -- the set of all independent sets of the characteristic graph $G_X$ defined above. 
By the data-processing inequality, the minimization in \eqref{eq:conditional-graph-entropy} can be restricted to $W$ ranging over maximal independent sets.
Note that $0\leq H_{G_X}(X|Y) \leq H(X|Y)$ hold in general; and if $G_X$ is a complete graph then $H_{G_X}(X|Y)=H(X|Y)$.

\paragraph{Common Information and Residual Information.}
G\'acs and K\"orner~\cite{GacsKorner} introduced the notion of common
information to measure a certain aspect of correlation between two random
variables. The G\'acs-K\"orner common information of a pair of correlated
random variables $(U,V)$ can be defined as $H(U\sqcap V)$, where $U\sqcap V$
is a random variable with maximum entropy among all random variables $Q$ that
are determined both by $U$ and by $V$ (i.e., there are functions $g$ and $h$
such that $Q =g(U)=h(V)$). It is not hard to see that $U\sqcap V$ is equal to the
 random variable corresponding to the set of connected components of the {\em characteristic bipartite graph} of $p_{UV}$ -- 
for a distribution $p_{UV}$, a bipartite graph on vertex set $\mathcal{U}\cup\mathcal{V}$ 
is said to be the characteristic bipartite graph of $p_{UV}$, 
if $u\in\mathcal{U}$ and $v\in\mathcal{V}$ are connected whenever $p_{UV}(u,v)>0$.
Note that if $p_{UV}$ is such that the characteristic bipartite graph is connected, then 
$U\sqcap V$ is constant and $H(U\sqcap V)=0$.
In \cite{PrabhakaranPr14}, the gap between mutual information and common information was termed {\em residual information}: $\RI(U;V):=I(U;V)-H(U\sqcap V)$.

In \cite{WolfWu08}, Wolf and Wullschleger identified (among other things) 
the following secure {\em data processing inequality} for residual information.
\begin{lem}[Secure data processing inequality \cite{WolfWu08}]\label{lem:monotone}
If $T,U,V,W$ are jointly distributed random variables such that the following two Markov chains hold: (i) $U-T-W$ and (ii) $T-W-V$, then
\[  \RI(T;W) \leq \RI((U,T);(V,W)).\]
\end{lem}
The Markov chain conditions can be viewed as follows: let $(U,T)$ and $(V,W)$ be the views of any pair of users, where, for the user holding $(U,T)$, $U$ can be thought of as all the messages exchanged during the protocol and $T$ can be thought of as its data (input and output) which is the ideal-world view. $(U,T)$ is the real-world view.
Now the Markov chain $U-T-W$ corresponds to the privacy requirement that $U$ (the rest of this user's view) can be simulated based on its data $T$, independent of the other user's data $W$; and similarly for the second user.
The lemma states that under this privacy condition, the residual information between the real-world views must be at least as large as that between the ideal-world views (i.e., the data).

In \cite{PrabhakaranPr14}, the following alternate definition of residual information
was given, which will be useful in lower-bounding conditional mutual information
terms.
\begin{align}
RI(U;V) := \displaystyle \min_{\substack{p_{Q|UV}: \\ I(Q;V|U)=0 \\ I(Q;U|V)=0}} I(U;V|Q). \label{eq:residual_info}
\end{align}
The random variable $Q$ which achieves the minimum is, in fact, $U\sqcap V$. Note that the residual information is always non-negative.
\begin{lem}[$RI$ tensorizes \cite{WolfWu08,PrabhakaranPr14}]\label{lem:RI_tensorizes}
For $(U^n;V^n)$, where $(U_i,V_i)$ are i.i.d., $RI(U^n;V^n) = nRI(U;V)$.
\end{lem}

\section{Outer Bounds on the Rate-Region for Perfectly Secure Computation}\label{sec:ps_lowerbounds}
This section is divided into four parts: in \Subsectionref{prelim_lbs} we derive preliminary lower bounds for secure computation; in \Subsectionref{improved_lbs} we give some techniques which significantly improve the preliminary bounds and lead to our main theorems; in \Subsectionref{randomness} we derive lower bounds on the amount of randomness required in secure computation protocols; and in \Subsectionref{examples} we consider some interesting examples -- secure protocols are given, and our lower bound results are analyzed for these example functions.

We consider a 3-user secure computation problem specified by $(p_{XY},
p_{Z|XY})$, see \Figureref{ps_setup}. Input $X^n$ to Alice (user-1) and
$Y^n$ to Bob (user-2) are distributed according to $p_{X,Y}$, i.i.d. Charlie (user-3)
wants to compute an output $Z^n$, which should be distributed according to
$p(z^n|x^n,y^n)=\Pi_{i=1}^n p_{Z|XY}(z_i|x_i,y_i)$. We say that a protocol
$\Pi_n(p_{XY},p_{Z|XY})$ for this setup is {\em perfectly secure} if the
output satisfies this and the protocol is perfectly secure against any
single user as defined by
\eqref{eq:ps_privacy_alice}-\eqref{eq:ps_privacy_charlie} below.
Recall that for a protocol $\Pi_n$, we define the rate quadruple $(R_{12},R_{23},R_{31},\rho)$ as $R_{ij} := \frac{1}{n}\mathbb{E}[L_{ij}]$, $i,j=1,2,3$, $i\neq j$, and $\rho := \frac{1}{n}H(M_{12},M_{23},M_{31},Z^n|X^n,Y^n)$.
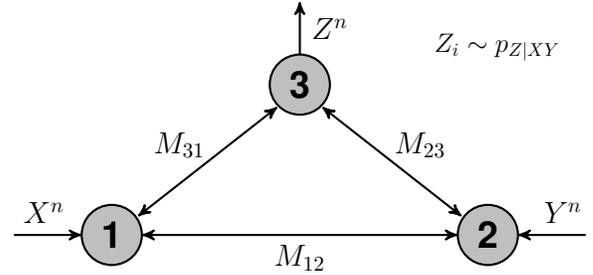
\begin{figure}[tb]
\centering
\begin{tikzpicture}[>=stealth', font=\sffamily\Large\bfseries, thick]
\draw [fill=lightgray] (-2.5,0) circle [radius=0.4]; \node at (-2.5,0) {1};
\draw [fill=lightgray] (2.5,0) circle [radius=0.4]; \node at (2.5,0) {2};
\draw [fill=lightgray] (0,2) circle [radius=0.4]; \node at (0,2) {3};

\draw [<->] (-2.1,0) -- (2.1,0); \node [scale=0.8] at (0,-0.3) {$M_{12}$};
\draw [<->] (-2.15,0.25) -- (-0.3,1.7); \node [scale=0.8] at (-1.6,1.2) {$M_{31}$};
\draw [<->] (2.15,0.25) -- (0.3,1.7); \node [scale=0.8] at (1.6,1.2) {$M_{23}$};

\draw [->] (-3.8,0) -- (-2.9,0); \node [scale=0.8] at (-3.4,0.3) {${X^n}$};
\draw [->] (3.8,0) -- (2.9,0); \node [scale=0.8] at (3.5,0.3) {${Y^n}$};
\draw [->] (0,2.4) -- (0,3.1); \node [scale=0.8] at (0.4,2.75) {${Z^n}$};

\node [right, scale=0.7] at (1.7,2.5) {$Z_i \sim p_{Z|XY}$};

\end{tikzpicture}
\caption{A setup for 3-user secure computation; privacy is required against single users (i.e., no collusion). Here $(X,Y)\sim p_{XY}$ and $Z_i\sim p_{Z|XY}$ for all $i$.} 
\label{fig:ps_setup}
\end{figure}
\begin{defn}\label{defn:ps_secure-protocol}
For a secure computation problem $(p_{XY},p_{Z|XY})$, the rate $(R_{12},R_{23},R_{31},\rho)$ is {\em achievable with perfect security for block-length} $n$, if there is a protocol $\Pi_n(p_{XY},p_{Z|XY})$ with rate $(R_{12},R_{23},R_{31},\rho)$ such that conditioned on inputs $X^n,Y^n$ of Alice and Bob, Charlie's output $Z^n$ is distributed according to $p(z^n|x^n,y^n)=\Pi_{i=1}^n p_{Z|XY}(z_i|x_i,y_i)$, and the following holds:
\begin{align}
I(M_{12},M_{31};Y^n,Z^n|X^n)&=0, \label{eq:ps_privacy_alice} \\
I(M_{12},M_{23};X^n,Z^n|Y^n)&=0, \label{eq:ps_privacy_bob} \\
I(M_{23},M_{31};X^n,Y^n|Z^n)&=0. \label{eq:ps_privacy_charlie}
\end{align}
Rate-region $\R^{n,\PS}$ is the closure of the set of all rate quadruples achievable with perfect security for block-length $n$. 
We say that $(R_{12},R_{23},R_{31},\rho)$ is {\em achievable with perfect security} if it is achievable with perfect security for some block-length $n$. And $\R^{\PS}$ is the closure of the set of all rate quadruples achievable with perfect security. 
\end{defn}
Here \eqref{eq:ps_privacy_alice} ensures that Alice learns no additional information about $(Y^n,Z^n)$; similarly for Bob; and \eqref{eq:ps_privacy_charlie} ensures that Charlie learns no additional information about $(X^n,Y^n)$ than revealed by $Z^n$.
\begin{remark}\label{remark:direct-sum_result}
{\em For perfectly secure computation, all our bounds are direct
sum bounds, i.e., our outer bound on $\R^{1,\PS}$ will also be an outer
bound for $\R^{\PS}$. So, for simplicity, we prove all our bounds for $n=1$
and then show that it holds for $\R^{\PS}$.}
\end{remark}
To make the presentation clear, in \Subsectionref{prelim_lbs} and \Subsectionref{improved_lbs}  we derive lower bounds only on the rates $R_{12},R_{23},R_{31}$, and lower bound on the randomness $\rho$ is derived in  \Subsectionref{randomness}. We prove bounds on the entropies $H(M_{ij})$, which, as argued in \Sectionref{prelims}, is a lower bound on the expected length of the transcript $M_{ij}$.

\subsection{Preliminary Lower Bounds}\label{subsec:prelim_lbs}
We first state the following basic lemma for any protocol for perfectly secure computation. Similar results have appeared in the literature earlier (for instance, special cases of \Lemmaref{general_cutset} appear in \cite{DodisMi00,WinklerWu10}).
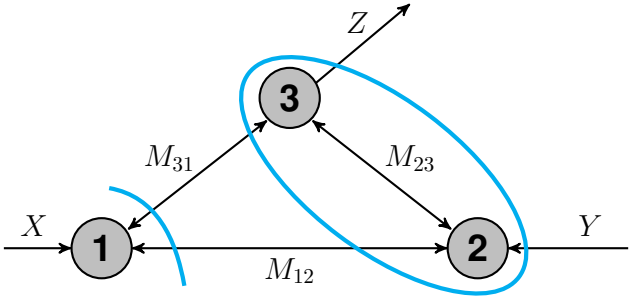
\begin{figure}
\centering
\begin{tikzpicture}[>=stealth', font=\sffamily\Large\bfseries, thick]
\draw [fill=lightgray] (-2.5,0) circle [radius=0.4]; \node at (-2.5,0) {1};
\draw [fill=lightgray] (2.5,0) circle [radius=0.4]; \node at (2.5,0) {2};
\draw [fill=lightgray] (0,2) circle [radius=0.4]; \node at (0,2) {3};

\draw [<->] (-2.1,0) -- (2.1,0); \node [scale=0.8] at (0,-0.3) {$M_{12}$};
\draw [<->] (-2.15,0.25) -- (-0.3,1.7); \node [scale=0.8] at (-1.6,1.2) {$M_{31}$};
\draw [<->] (2.15,0.25) -- (0.3,1.7); \node [scale=0.8] at (1.6,1.2) {$M_{23}$};

\draw [->] (-3.8,0) -- (-2.9,0); \node [scale=0.8] at (-3.4,0.3) {${X}$};
\draw [->] (4.5,0) -- (2.9,0); \node [scale=0.8] at (4,0.3) {${Y}$};
\draw [->] (0.35,2.2) -- (1.6,3.25); \node [scale=0.8] at (0.9,3) {$Z$};

\draw[ultra thick, cyan] (-2.4,0.8) to [out=-10,in=100] (-1.4,-0.5);
\draw [ultra thick, cyan] (1.25,0.97) ellipse [x radius=1cm, y radius=2.25cm, rotate=53];

\end{tikzpicture}
\caption{A cut separating Alice from Bob \& Charlie. Protocol $\Pi$ induces a 2-user secure computation protocol between Alice and {\em combined Bob-Charlie} with privacy requirement only against Alice.}
\label{fig:setup_two-party}
\end{figure}
\begin{lem}\label{lem:general_cutset}
In any secure protocol $\Pi_1(p_{XY},p_{Z|XY})$, where $(p_{XY},p_{Z|XY})$ is in normal form,
the following must hold:
\begin{align}
H(X|M_{12},M_{31}) &= 0, \label{eq:general_cutset_alice} \\
H(Y|M_{12},M_{23}) &= 0, \label{eq:general_cutset_bob} \\
H(Z|M_{23},M_{31}) &= 0. \label{eq:general_cutset_charlie}
\end{align}
\end{lem}
We prove this lemma in \Appendixref{proofs}. \Lemmaref{general_cutset} states
the simple fact that, for $(p_{XY},p_{Z|XY})$ in normal form,\footnote{\label{foot:comment_zero-error}
Note that when the input distribution $p_{XY}$ has full support, this
assumption is without loss of generality (see \Sectionref{prelims}). When
$p_{XY}$ has full support, information theoretic tools have proved to be
successful in deriving optimal bounds for zero-error computation. But
deriving bounds for zero-error computation with arbitrary input
distribution is more amenable to combinatorial arguments \cite{AlonOr96}.
Hence, for perfectly secure computation, in this paper we do not deal with
arbitrary $(p_{XY},p_{Z|XY})$; we confine our attention to either $p_{XY}$
which have full support, or more generally, to $(p_{XY},p_{Z|XY})$ which
satisfy some technical conditions.}
the cut separating Alice from Bob and Charlie must reveal Alice's input $X$
(see \Figureref{setup_two-party}). The intuition is that, since Alice is
not allowed to learn any new information about $Y$, correctness condition
forces Alice to reveal $X$. Note that this conclusion crucially depends on
the privacy requirement against Alice. For example, consider $X=(X_0,X_1),
X_0,X_1\in\{0,1\}$, $Y\in\{0,1\}$, and $X_0,X_1,Y$ are i.i.d. Bern(1/2).
Let $f((X_0,X_1),Y)=X_Y$. Without the privacy condition, Bob may send $Y$
to Alice who can compute $Z=X_Y$ and send this to Charlie. $H(X)=2$, but
here the cut $(M_{12},M_{31})$ reveals only 1 bit of information about $X$.
Similarly, the cut separating Bob from the rest of the users must reveal
his input, and the cut separating Charlie must reveal his output. This
relies on the fact that Alice and Bob obtain no output, and Charlie has no
input in our model. We obtain a preliminary lower bound below by using the
above lemma and the secure data-processing inequality for residual
information (\Lemmaref{monotone}).
\begin{thm}\label{thm:prelim_lbs}
For a secure computation problem $(p_{XY},p_{Z|XY})$, where $(p_{XY},p_{Z|XY})$ is in
normal form, if $(R_{12},R_{23},R_{31},\rho)\in\R^{\PS}$, then,
\begin{align}
R_{31} &\geq \max\{RI(X;Y), RI(Y;Z)\} + H(X,Z|Y) \label{eq:prelim_lb_M31},\\
R_{23} &\geq \max\{RI(X;Y), RI(X;Z)\} + H(Y,Z|X) \label{eq:prelim_lb_M23},\\
R_{12} &\geq \max\{RI(X;Z), RI(Y;Z)\} + H(X,Y|Z) \label{eq:prelim_lb_M12}.
\end{align}
\end{thm}
\begin{proof}
We shall prove \eqref{eq:prelim_lb_M31} for $\R^{1,\PS}$. The fact that it also holds for $\R^{\PS}$ follows from \Lemmaref{RI_tensorizes}. The other two inequalities 
can be shown similarly.
\begin{align}
H(M_{31}) &\geq \max\{H(M_{31}|M_{12}),H(M_{31}|M_{23})\} \notag \\
&= \max\{I(M_{31};M_{23}|M_{12}),I(M_{31};M_{12}|M_{23})\} \notag \\
&\qquad + H(M_{31}|M_{12},M_{23})
\label{eq:prelim_bound_M31}
\end{align}
We can bound the last term of \eqref{eq:prelim_bound_M31} as follows (to
already get a na\"ive bound):
\begin{align*}
H(M_{31}|M_{12},M_{23}) &\stackrel{\text{(a)}}{=} H(M_{31},X,Z|M_{12},M_{23},Y)\\
&\geq H(X,Z|M_{12},M_{23},Y) \stackrel{\text{(b)}}{=} H(X,Z | Y),
\end{align*}
where (a) follows from \Lemmaref{general_cutset} and (b) follows from \eqref{eq:ps_privacy_bob}, i.e., privacy against Bob.
Next, we lower-bound the first term inside the $\max$ of
\eqref{eq:prelim_bound_M31} by $RI(X;Y)$ as follows.
\begin{align}
I(M_{31};M_{23}|M_{12}) &= I(M_{12},M_{31};M_{12},M_{23}|M_{12}) \nonumber \\
&\stackrel{\text{(c)}}{=} I(M_{12},M_{31},X ; M_{12},M_{23},Y | M_{12}) \nonumber \\
& \geq RI(M_{12},M_{31},X;M_{12},M_{23},Y), \nonumber
\end{align}
where (c) follows from \Lemmaref{general_cutset}, and the last inequality 
follows from \eqref{eq:residual_info} by taking $Q=M_{12}$.
Now, by privacy against Alice we have $(M_{12},M_{31})-X-Y$, and by privacy against
Bob we have $(M_{12},M_{23})-Y-X$. Applying \Lemmaref{monotone} with the
above Markov chains, we get
\[RI(M_{12},M_{31},X;M_{12},M_{23},Y) \geq RI(X;Y).\]
Similarly, we can lower-bound the second term inside $\max$ of \eqref{eq:prelim_bound_M31} by $RI(Y;Z)$,
completing the proof.
\end{proof}
A consequence of \Lemmaref{general_cutset} is that the transcripts in a
secure computation protocol form shares in a ``correlated multi-secret
sharing scheme'' (CMSS) for the same distribution $p_{XYZ}=p_{XY}p_{Z|XY}$;
see \Appendixref{CMSS_sampling} for details.\footnote{We remark that our
notion of multiple secret sharing schemes is different from that of
\cite{BlundoSaCrGaVa94}, which (implicitly) required that secrets with
different access structures be independent of each other. In our case, $Z$
is typically strongly correlated with $X,Y$, often via a deterministic
function.} Hence, lower bounds on the entropies of the shares in a CMSS
imply lower bounds on the entropies of the messages in a secure computation
protocol. In \Appendixref{CMSS_sampling} we also derive stronger bounds on the sizes of these shares in CMSS.

To strengthen the preliminary bounds in \Theoremref{prelim_lbs}, we will restrict our 
attention in the rest of the paper to $p_{XY}$ having full support, which allows us to 
assume, without loss of generality, that the function $p_{Z|XY}$ is in normal form;
see \Sectionref{prelims} for details.

\subsection{Improved Lower Bounds}\label{subsec:improved_lbs}
{To improve the bounds in \Theoremref{prelim_lbs}, we $(i)$ give a technique, which we call 
{\em distribution switching}, and $(ii)$ prove {\em an information inequality for 3-user interactive 
protocols}, using which we improve the above bounds and obtain our main theorems.
We first prove the following lemma which gives an upper bound on the mutual information between the transcript on 
any link and the data (inputs and output) in terms of the G\'acs-K\"orner common information 
(which is equal to the difference between mutual information and residual information) between 
the data of the two users associated on that link.

\begin{lem}\label{lem:XYZ_M123}
For any secure protocol $\Pi_1(p_{XY},p_{Z|XY})$, where $p_{XY}$ need not have full support and $p_{Z|XY}$ need not be in the normal form, the following  hold:
\begin{align}
I(M_{12};X,Y,Z) &\leq I(X;Y) - RI(X;Y), \label{eq:ps_M12_XYZ} \\
I(M_{31};X,Y,Z) &\leq I(X;Z) - RI(X;Z), \label{eq:ps_M31_XYZ} \\
I(M_{23};X,Y,Z) &\leq I(Y;Z) - RI(Y;Z). \label{eq:ps_M23_XYZ}
\end{align}
\end{lem}
\begin{proof}
We first show the bound in \eqref{eq:ps_M12_XYZ}.
Since $I(M_{12};X,Y,Z)=I(M_{12};X) + I(M_{12};Y,Z|X)$, where second term is equal to zero by 
privacy against Alice \eqref{eq:ps_privacy_alice}, it is enough to show 
$I(M_{12};X)\leq I(X;Y)-RI(X;Y)$.
\begin{align}
I(M_{12};X) &= I(M_{12},Y;X) - I(Y;X|M_{12}) \nonumber \\
&= I(Y;X) + I(M_{12};X|Y) - I(Y;X|M_{12}) \nonumber \\
&= I(X;Y) - I(X;Y|M_{12}) \label{eq:ps_M12_XYZ1} \\
&\leq I(X;Y) - RI(X;Y). \label{eq:ps_M12_XYZ2}
\end{align}
We get \eqref{eq:ps_M12_XYZ1} by substituting $I(M_{12};X|Y)=0$, which follows from privacy 
against Bob \eqref{eq:ps_privacy_bob}. The inequality \eqref{eq:ps_M12_XYZ2} is obtained 
by substituting $I(X;Y|M_{12})\geq RI(X;Y)$, which can be proved by taking $Q=M_{12}$ in the 
definition of residual information \eqref{eq:residual_info} (where the Markov chain conditions 
$M_{12}-X-Y$ and $M_{12}-Y-X$ follow from privacy against Alice and privacy against Bob, respectively).

Similarly, we can show \eqref{eq:ps_M31_XYZ} using privacy against Alice and privacy against 
Charlie, and \eqref{eq:ps_M23_XYZ} using privacy against Bob and privacy against Charlie.
\end{proof}
As mentioned in \Sectionref{prelims}, for any jointly distributed random variables $U,V$, 
if the characteristic bipartite graph of $p_{UV}$ is connected, then $I(U;V)=RI(U;V)$.
Hence, as a simple consequence of the above lemma we obtain the following Lemma,
which states that privacy requirements imply the independence of the transcript $M_{12}$ 
generated by a secure protocol computing $p_{Z|XY}$ and the inputs. 
Moreover, if the function $p_{Z|XY}$ satisfies some additional constraints, the other two 
transcripts also become independent of the inputs. 
\begin{lem}\label{lem:XYZ_inde_M123}
Consider a function $p_{Z|XY}$ not necessarily in normal form.
\begin{enumerate}
\item Suppose $p_{XY}$ is such that the characteristic bipartite graph of $p_{XY}$ is connected.
Then any secure protocol $\Pi_1(p_{XY},p_{Z|XY})$ satisfies $I(M_{12};X,Y,Z)=0$.

\item Suppose $(p_{XY},p_{Z|XY})$ is such that the characteristic bipartite graph of the 
induced distribution $p_{XZ}$ is connected.
Then any secure protocol $\Pi_1(p_{XY},p_{Z|XY})$ satisfies $I(M_{31};X,Y,Z)=0$.

The characteristic bipartite graph of $p_{XZ}$ is connected if 
$p_{XY}$ has full support and $p_{Z|XY}$ satisfies the following condition:\\[0.15cm]
{\bf Condition 1.} There is no non-trivial partition $\mathcal{X} = \mathcal{X}_1 \cup \mathcal{X}_2$ (i.e., $\mathcal{X}_1 \cap \mathcal{X}_2 = \varnothing$ and neither $\mathcal{X}_1$ nor $\mathcal{X}_2$ is empty), such that if $\mathcal{Z}_k = \{z \in {\mathcal Z} : x \in \mathcal{X}_k, y \in \mathcal{Y}, p(z|x,y)>0\}, k=1,2$, their intersection $\mathcal{Z}_1 \cap \mathcal{Z}_2$ is empty.
\item Suppose $(p_{XY},p_{Z|XY})$ is such that the characteristic bipartite graph of the 
induced distribution $p_{YZ}$ is connected.
Then any secure protocol $\Pi_1(p_{XY},p_{Z|XY})$ satisfies $I(M_{23};X,Y,Z)=0$.

The characteristic bipartite graph of $p_{YZ}$ is connected if 
$p_{XY}$ has full support and $p_{Z|XY}$ satisfies the following condition:\\[0.15cm]
{\bf Condition 2.} There is no non-trivial partition $\mathcal{Y} = \mathcal{Y}_1 \cup \mathcal{Y}_2$ such that if $\mathcal{Z}_k = \{z\in{\mathcal Z} : x \in \mathcal{X}, y \in \mathcal{Y}_k, p(z|x,y)>0\}, k=1,2$, their intersection $\mathcal{Z}_1 \cap \mathcal{Z}_2$ is empty.
\end{enumerate}
\end{lem}
In the context of 1) above, we point out that $p_{XY}$ having a connected characteristic bipartite graph is a weaker 
condition than $p_{XY}$ having full support.
}

\subsubsection{Distribution Switching}
We will argue that even if the protocol is allowed to depend on the input distribution (as we do here), privacy requirements will require that the lower bounds derived for when the distributions of the inputs are changed continue to hold for the original setting. The main idea can be summarized as follows: Any secure protocol $\Pi(p_{XY},p_{Z|XY})$, where distribution $p_{XY}$ has full support, continues to be a secure protocol even if we switch the input distribution to a different one $p_{X'Y'}$. This follows, as we show below, directly from examining the (correctness and privacy) conditions required for a protocol to be secure.
\begin{itemize}
\item {\bf Correctness:} Note that we only change the input distribution, but the function being computed remains unchanged, i.e., $p_{Z'|X'Y'}(z|x,y)=p_{Z|XY}(z|x,y)$, for every $(x,y,z)\in\mathcal{X}\times\mathcal{Y}\times\mathcal{Z}$. The correctness condition requires that with the new input distribution, Charlie's output $Z'$ should be distributed according to $p_{Z'|X'=x,Y'=y}$, where $x$ and $y$ are inputs of Alice and Bob respectively, which, as we argued before is equal to $p_{Z|X=x,Y=y}$.
\item {\bf Privacy:} We have to show that privacy conditions against Alice, Bob, and Charlie remain intact if we change the original input distribution $p_{XY}$ with a different one $p_{X'Y'}$. In our secure protocol model, once Alice and Bob are given inputs $X=x$ and $Y=y$, respectively, the protocol produces $(m_{12},m_{23},m_{31},z)$ according to the conditional distribution $p_{M_{12}M_{23}M_{31}Z|XY}(m_{12},m_{23},m_{31},z|x,y)$. Note that this conditional distribution does not depend on $p_{XY}$. So, if we change the input distribution $p_{XY}$ to $p_{X'Y'}$, the conditional distribution $p_{M'_{12}M'_{23}M'_{31}Z'|X'Y'}(m_{12},m_{23},m_{31},z|x,y)$ does not change. More precisely, the following holds for every distribution $p_{X'Y'}$ and $(x,y)\in\mathcal{X}\times\mathcal{Y}$ such that $p_{X'Y'}(x,y)>0$.
\begin{align}
& p_{M'_{12}M'_{23}M'_{31}Z'|X'Y'}(m_{12},m_{23},m_{31},z|x,y) \nonumber \\
&\qquad = p_{M_{12}M_{23}M_{31}Z|XY}(m_{12},m_{23},m_{31},z|x,y).\label{eq:dist_switch}
\end{align}
We only show the result for Alice, that is, we will prove that $I(M_1;(Y,Z)|X)=0$ $\implies$ $I(M'_1;(Y',Z')|X')=0$, where $M'_1=(M'_{12},M'_{31})$ is generated by the original secure protocol when the inputs to Alice and Bob are distributed according to $p_{X'Y'}$. From \eqref{eq:dist_switch}, we get
\begin{align*}
& p_{M'_1Z'|X'Y'}(m_1,z|x,y) = p_{M_1Z|XY}(m_1,z|x,y) \\
\Rightarrow \ & p_{Z'|X'Y'}(z|x,y)p_{M'_1|X'Y'Z'}(m_1|x,y,z) \\
&\qquad \qquad = p_{Z|XY}(z|x,y)p_{M_1|XYZ}(m_1|x,y,z) \\
\stackrel{\text{(a)}}{\Rightarrow} \ & p_{M'_1|X'Y'Z'}(m_1|x,y,z) = p_{M_1|XYZ}(m_1|x,y,z) \\
\stackrel{\text{(b)}}{\Rightarrow} \ & p_{M'_1|X'Y'Z'}(m_1|x,y,z) = p_{M_1|X}(m_1|x),
\end{align*}
where (a) holds because $p_{Z'|X'Y'}(z|x,y) = p_{Z|XY}(z|x,y)$, and these are non-zero because the protocol produces $z$ when Alice and Bob are given their respective inputs $x$ and $y$; (b) follows from the privacy against Alice, i.e., $I(M_1;(Y,Z)|X)=0$. Note that the last equality holds true for any distribution $p_{X'Y'}$, and it implies that $M'_1-X'-(Y',Z')$, i.e., $I(M'_1;(Y',Z')|X')=0$.\\

Privacy against Bob and Charlie are similarly proved.
\end{itemize}

\subsubsection{An Information Inequality for Protocols}
We exploit the fact that, in a protocol, transcripts are generated by the users interactively rather than by an omniscient dealer. Towards this, we derive an information inequality relating the transcripts on different links in general 3-user protocols, in which users do not share any common or correlated randomness or correlated inputs at the beginning of the protocol. Note that our model for protocols does indeed satisfy these conditions when the inputs are independent of each other.
\begin{lem}
\label{lem:infoineq}
In any 3-user protocol (not necessarily secure), if the inputs to the users are
independent of each other, then, for $\{\alpha,\beta,\gamma\}=\{1,2,3\}$,
\[ I(M_{\gamma\alpha};M_{\beta\gamma}) \geq
I(M_{\gamma\alpha};M_{\beta\gamma}|M_{\alpha\beta}).\]
\end{lem}
This inequality provides us with a means to exploit the protocol structure
behind transcripts and helps us to lower-bound
$I(M_{\gamma\alpha};M_{\beta\gamma})$ in terms of the distribution
$p_{XYZ}$. This is achieved in a few easy steps:
\Lemmaref{infoineq}, combined with \eqref{eq:residual_info}, lets us
derive that $I(M_{\gamma\alpha};M_{\beta\gamma}) \ge \RI(M_{\gamma\alpha},M_{\alpha\beta};M_{\beta\gamma},M_{\alpha\beta})$.
Further combined with \Lemmaref{general_cutset}, the right hand side can be
equated with the residual information of the views of the users
$\alpha$ and $\beta$. Finally, by the secure data processing inequality
(\Lemmaref{monotone}), this can be lower-bounded in terms of the
residual information of (one of) the inputs and the output (if $\{\alpha,\beta\}\neq \{1,2\}$) or both the inputs (if $\{\alpha,\beta\}=\{1,2\}$).

\begin{proof}[Proof of \Lemmaref{infoineq}]
For any choice of distinct $\alpha,\beta,\gamma$ in $\{1,2,3\}$, the inequality of \Lemmaref{infoineq} is equivalent to the following inequality:
\begin{align}
&H(M_{12})+H(M_{23})+H(M_{31}) - H(M_{23},M_{31}) \nonumber \\
& - H(M_{31},M_{12}) - H(M_{12},M_{23}) + H(M_{12},M_{23},M_{31}) \geq 0. \label{eq:equi_infoineq}
\end{align}
To prove \eqref{eq:equi_infoineq} we apply induction on the number of rounds of the protocol.\\
{\em Base case:} At the beginning of the protocol, all the transcripts $M_{12}, M_{23}, M_{31}$ are empty. So the inequality \eqref{eq:equi_infoineq} is trivially true. \\
{\em Inductive step:} Assume that the inequality \eqref{eq:equi_infoineq} is true at the end of round $t$, and we prove it for $t+1$. Notice that the new message in round $t+1$ can be one of the six different messages -- from user-$i$ to user-$j$, and vice versa; but since \eqref{eq:equi_infoineq} is symmetric in all the transcripts, it is enough to prove the inequality when the new message sent in round $t+1$ is, say, from user-1 to user-2.

For simplicity, let us denote the transcript on 1-2 link at the end of round $t$ by $M_{12}$ itself, the new message sent by user-1 to user-2 in round $t+1$ by $\Delta M_{\vec{12}}$, and the transcript  on 1-2 link at the end of round $t+1$ by $\widetilde{M}_{12}$. Notation for the remaining transcripts are defined similarly. Hence, $\widetilde{M}_{12}=(M_{12},\Delta M_{\vec{12}})$, $\widetilde{M}_{23}=M_{23}$, $\widetilde{M}_{31}=M_{31}$. Since the users do not share any common or correlated randomness, the new message $\Delta M_{\vec{12}}$ is conditionally independent of the transcript $M_{23}$ between the other two users, conditioned on transcripts $(M_{12},M_{31})$ on both the links to which user-1 is associated with. So we have the following Markov chain: 
\begin{align}
\Delta M_{\vec{12}} - (M_{12},M_{31}) - M_{23} \label{eq:mc_infoineq}
\end{align}
Now we can show that the inequality in \eqref{eq:equi_infoineq} holds at the end of round $t+1$ as follows:
\begin{align}
&H(\widetilde{M}_{12}) + H(\widetilde{M}_{23}) + H(\widetilde{M}_{31}) - H(\widetilde{M}_{23},\widetilde{M}_{31}) \nonumber \\
&\quad - H(\widetilde{M}_{31},\widetilde{M}_{12}) - H(\widetilde{M}_{12},\widetilde{M}_{23}) + H(\widetilde{M}_{12},\widetilde{M}_{23},\widetilde{M}_{31}) \nonumber \\
&= H(M_{12},\Delta M_{\vec{12}})+H(M_{23})+H(M_{31}) - H(M_{23},M_{31}) \nonumber \\
&\qquad - H(M_{31},M_{12},\Delta M_{\vec{12}}) - H(M_{12},\Delta M_{\vec{12}},M_{23}) \nonumber \\
&\qquad \qquad + H(M_{12},\Delta M_{\vec{12}},M_{23},M_{31}) \nonumber \\
&\geq H(\Delta M_{\vec{12}}|M_{12}) - H(\Delta M_{\vec{12}}|M_{12},M_{31}) \nonumber \\
&\ - H(\Delta M_{\vec{12}}|M_{12},M_{23}) + H(\Delta M_{\vec{12}}|M_{12},M_{23},M_{31}) \label{eq:using_induction-hypo} \\
&= I(\Delta M_{\vec{12}} ; M_{23}|M_{12}) - I(\Delta M_{\vec{12}};M_{23}|M_{12},M_{31}) \nonumber \\
&\geq 0 \nonumber
\end{align}
In \eqref{eq:using_induction-hypo} we use the induction hypothesis, and last inequality follows from \eqref{eq:mc_infoineq} and the fact that the conditional mutual information is always non-negative.
\end{proof}

\subsubsection{Main Lower Bounds}
\begin{thm} \label{thm:lowerbound}
For a secure computation problem $(p_{XY},p_{Z|XY})$, where $p_{XY}$ has full support and $p_{Z|XY}$ is in normal form, if $(R_{12},R_{23},R_{31},\rho)\in\R^{\PS}$, then, 
\begin{align}
R_{31} &\geq \max \left\{ \begin{array}{l} \displaystyle \max_{p_{X}p_{Y'}} RI(Y';Z') \\ \quad + \displaystyle \max_{p_{XY''}} RI(X;Y'') + H(X,Z''|Y''),\\ 
\displaystyle \max_{p_{XY'}} RI(Y';Z') + H(X,Z'|Y')
\end{array}\right\}, \label{eq:M31_NoCond} \\
R_{23} &\geq \max \left\{ \begin{array}{l} \displaystyle \max_{p_{X'}p_{Y}} RI(X';Z') \\ \quad + \displaystyle \max_{p_{X''Y}} RI(X'';Y) + H(Y,Z''|X''),\\ 
\displaystyle \max_{p_{X'Y}} RI(X';Z') + H(Y,Z'|X')
\end{array}\right\}, \label{eq:M23_NoCond} \\
R_{12} &\geq \max \left\{ \begin{array}{l} \displaystyle \max_{p_{X'}p_{Y'}} RI(Y';Z') \\ \quad + \displaystyle \max_{p_{X'Y''}} RI(X';Z'') + H(X',Y''|Z''),\\ 
\displaystyle \max_{p_{X'}p_{Y'}} RI(X';Z') \\ \quad + \displaystyle \max_{p_{X''Y'}} RI(Y';Z'') + H(X'',Y'|Z'')
\end{array}\right\}, \label{eq:M12_bound}
\end{align}
where the maxima are over distributions having full support. The terms in the right hand side of \eqref{eq:M31_NoCond} are evaluated using the distribution $p_X$ of the data $X$ of Alice. The terms in~\eqref{eq:M31_NoCond}, for instance, are evaluated using
\begin{align*}
\text{$1^{\text{st}}$ bound: } p_{X,Y',Z'}(x,y,z)&=p_{X}(x)p_{Y'}(y)p_{Z|X,Y}(z|x,y),\\
p_{X,Y'',Z''}(x,y,z)&=p_{XY''}(x,y)p_{Z|X,Y}(z|x,y), \\
\text{$2^{\text{nd}}$ bound: } p_{X,Y',Z'}(x,y,z)&=p_{XY'}(x,y)p_{Z|X,Y}(z|x,y).
\end{align*}
Similarly, the terms in \eqref{eq:M23_NoCond} are evaluated using the distribution $p_Y$ of the data $Y$ of Bob. The lower bound in \eqref{eq:M12_bound} does not depend on the distributions $p_X$ and $p_Y$ of the data. The terms in the first bound of~\eqref{eq:M12_bound}, for instance, are evaluated using
\begin{align*}
\text{$1^{\text{st}}$ bound: }p_{X',Y',Z'}(x,y,z)&=p_{X'}(x)p_{Y'}(y)p_{Z|X,Y}(z|x,y),\\
p_{X',Y'',Z''}(x,y,z)&=p_{X'Y''}(x,y)p_{Z|X,Y}(z|x,y).
\end{align*}
\end{thm}
\begin{proof}
We again prove this only for $n=1$. To see the general case of, say, $n=m$
we simply invoke the $n=1$ result with inputs $X^m,Y^m$, and 
$p_{Z^m|X^m,Y^m}$ defined as the $m$-wise product of $p_{Z|X,Y}$. Making
memoryless choices for the primed random variables in the bound, the result
follows from \Lemmaref{RI_tensorizes}.

Suppose we have a secure protocol for computing $p_{Z|XY}$ in the normal form under $p_{XY}$ which has full support. Consider $H(M_{31})$,
\begin{align*}
H(M_{31}) &= I(M_{31};M_{12}) \\
&\qquad + I(M_{31};M_{23}|M_{12}) + H(M_{31}|M_{12},M_{23}).
\end{align*}
By privacy against Alice, conditioned on $X$, $M_{31}$ is independent of $Y$. 
So, by distribution switching, keeping the marginal distribution of $X$ same, we may switch the distribution $p_{XY}$ to, say, $p_{XY''}$, which also has full support, and the resulting $M_{31}$ has the same distribution as under $p_{XY}$, i.e.,
\begin{align*}
H(M_{31}) &= \max_{p_{XY''}} I(M_{31};M_{12}) \\
&\qquad + I(M_{31};M_{23}|M_{12}) + H(M_{31}|M_{12},M_{23}).
\end{align*}
Under this switched distribution, let us consider the first term $I(M_{31};M_{12})$. 
Let us notice that, by privacy against Alice, $(M_{31},M_{12})$ must again be independent of $Y''$. 
Hence, even if we switch the distribution of $Y$ to, say, $p_{Y'}$ (keeping the marginal of $X$ same), the joint distribution of $(M_{31},M_{12})$ must remain unchanged. 
Hence, we have that $I(M_{31};M_{12})$ under the distribution $p_{XY''}$ is the same as that under $p_{XY'}$. Therefore,
\begin{align*}
H(M_{31}) &= \max_{p_{XY'}} I(M_{31};M_{12}) \\
&\quad + \max_{p_{XY''}} I(M_{31};M_{23}|M_{12}) + H(M_{31}|M_{12},M_{23}).
\end{align*}
If we take the distribution under the first $\max$ to be the product distribution, i.e., $p_{XY'}=p_Xp_{Y'}$, then we can apply \Lemmaref{infoineq} and get the following:
\begin{align*}
H(M_{31}) &\geq \max_{p_Xp_{Y'}} I(M_{31};M_{12}|M_{23}) \\
&\quad + \max_{p_{XY''}} I(M_{31};M_{23}|M_{12}) + H(M_{31}|M_{12},M_{23}).
\end{align*}
We can bound each of the three terms separately as follows: 
$(i)$ Using the definition of residual information, \Lemmaref{general_cutset}, and \Lemmaref{monotone}, we can show that $I(M_{31};M_{12}|M_{23}) \geq RI(Y;Z)$ and $I(M_{31};M_{23}|M_{12}) \geq RI(X;Y)$; and 
$(ii)$ using \Lemmaref{general_cutset} and privacy against Bob, we get $H(M_{31}|M_{12},M_{23}) \geq H(X,Z|Y)$. 
This gives the first bound on $H(M_{31})$ (first row of \eqref{eq:M31_NoCond}).
\begin{align*}
H(M_{31}) &\geq \max_{p_Xp_{Y'}} RI(Y';Z') \\
&\quad + \max_{p_{XY''}} RI(X;Y'') + H(X,Z''|Y'').
\end{align*}
For the second bound on $H(M_{31})$ (second row of \eqref{eq:M31_NoCond}), we first expand $H(M_{31})$ in a different way as follows:
\begin{align*}
H(M_{31})&=I(M_{31};M_{23}) \\
&\qquad + I(M_{31};M_{12}|M_{23}) + H(M_{31}|M_{12},M_{23}).
\end{align*}
Now drop the first term $I(M_{31},M_{23})$ and proceed as above.

The bounds on $H(M_{23})$ follow in an identical fashion. 
To see the bounds on $H(M_{12})$, let us recall that $M_{12}$ is independent of $X,Y$ (by \Lemmaref{XYZ_inde_M123}),
 and hence we may switch the distributions of both $X$ and $Y$. 
Furthermore, let us note that we may write $H(M_{12})$ in two different ways.
\begin{align}
H(M_{12}) &= I(M_{12};M_{31}) \nonumber \\
&\quad + I(M_{12};M_{23}|M_{31})+H(M_{12}|M_{23},M_{31})\label{eq:M12intwoways1}\\
H(M_{12}) &= I(M_{12};M_{23}) \nonumber \\
&\quad + I(M_{12};M_{31}|M_{23})+H(M_{12}|M_{23},M_{31}).\label{eq:M12intwoways2}
\end{align}
Using \eqref{eq:M12intwoways1} and proceeding as we did for $H(M_{31})$ leads to the top row of the right hand side of \eqref{eq:M12_bound}, and \eqref{eq:M12intwoways2} leads to the bottom row. 
\end{proof}

When the function satisfies certain additional constraints, we can strengthen the lower bounds on the $H(M_{23})$ and $H(M_{31})$ as shown below.
\begin{thm} \label{thm:conditionallowerbounds}
For a secure computation problem $(p_{XY},p_{Z|XY})$, where $p_{XY}$ has full support and $p_{Z|XY}$ is in normal form, if $(R_{12},R_{23},R_{31},\rho)\in\R^{\PS}$, then, 
\begin{enumerate}
\item Suppose the function $p_{Z|XY}$ satisfies Condition~1 of \Lemmaref{XYZ_inde_M123}, that is, there is no non-trivial partition $\mathcal{X} = \mathcal{X}_1 \cup \mathcal{X}_2$ (i.e., $\mathcal{X}_1 \cap \mathcal{X}_2 = \varnothing$ and neither $\mathcal{X}_1$ nor $\mathcal{X}_2$ is empty), such that if
$\mathcal{Z}_k = \{z \in {\mathcal Z} : x \in \mathcal{X}_k, y \in \mathcal{Y}, p(z|x,y)>0\}, k=1,2$, their intersection $\mathcal{Z}_1 \cap \mathcal{Z}_2$ is empty. Then we have the following strengthening of \eqref{eq:M31_NoCond}.
\begin{align}
R_{31} &\geq \max \left\{ \begin{array}{l} \displaystyle \max_{p_{X'}p_{Y'}} RI(Y';Z') \\ + \displaystyle \max_{p_{X'Y''}} RI(X';Y'') + H(X',Z''|Y''), \\ 
\displaystyle \max_{p_{X'Y'}} RI(Y';Z') + H(X',Z'|Y')
\end{array}\right\} \label{eq:M31_WithCond}
\end{align}
where the maxima are over distributions having full support, and the terms in~\eqref{eq:M31_WithCond}, for instance, are evaluated using
\begin{align*}
&p_{X',Y',Z'}(x,y,z)=p_{X'}(x)p_{Y'}(y)p_{Z|X,Y}(z|x,y),\\
&p_{X',Y'',Z''}(x,y,z)=p_{X'Y''}(x,y)p_{Z|X,Y}(z|x,y), \\
&p_{X',Y',Z'}(x,y,z)=p_{X'Y'}(x,y)p_{Z|X,Y}(z|x,y).
\end{align*}
\item Suppose the function $p_{Z|XY}$ satisfies Condition~2 of \Lemmaref{XYZ_inde_M123}, that is, there is no non-trivial partition $\mathcal{Y} = \mathcal{Y}_1 \cup \mathcal{Y}_2$ such that if $\mathcal{Z}_k = \{z\in{\mathcal Z} : x \in \mathcal{X}, y \in \mathcal{Y}_k, p(z|x,y)>0\}, k=1,2$, their intersection $\mathcal{Z}_1 \cap \mathcal{Z}_2$ is empty. Then we have the following strengthening of \eqref{eq:M23_NoCond}.
\begin{align}
R_{23} &\geq \max \left\{ \begin{array}{l} \displaystyle \max_{p_{X'}p_{Y'}} RI(X';Z') \\ + \displaystyle \max_{p_{X''Y'}} RI(X'';Y') + H(Y',Z''|X''), \\ 
\displaystyle \max_{p_{X'Y'}} RI(X';Z') + H(Y',Z'|X')
\end{array}\right\} \label{eq:M23_WithCond}
\end{align}
where the maxima are over distributions having full
support, and the terms in~\eqref{eq:M23_WithCond}, for instance, are evaluated using
\begin{align*}
&p_{X',Y',Z'}(x,y,z)=p_{X'}(x)p_{Y'}(y)p_{Z|X,Y}(z|x,y),\\
&p_{X'',Y',Z''}(x,y,z)=p_{X''Y'}(x,y)p_{Z|X,Y}(z|x,y), \\
&p_{X',Y',Z'}(x,y,z)=p_{X'Y'}(x,y)p_{Z|X,Y}(z|x,y).
\end{align*}

\end{enumerate}
\end{thm}
\begin{proof} We again prove this only for the case of $n=1$. The general
case follows for the same reasons as in the proof of
\Theoremref{lowerbound}.
By \Lemmaref{XYZ_inde_M123}, $M_{31}$ is independent of $p_{XY}$ under
condition~1. Hence, when condition~1 is satisfied, we may switch the
distribution $p_{XY}$ to $p_{X'Y'}$. Note that this is {\em unlike} what we
did in the proof of the lower bound on $H(M_{31})$ in
\Theoremref{lowerbound} where we switched $p_{XY}$ to $p_{XY'}$ keeping the
marginal distribution of $X$ same. Now proceeding similarly as there leads
us to~\eqref{eq:M31_WithCond}. Similarly, under condition~2, $M_{23}$ is
independent of $X,Y$ which leads to~\eqref{eq:M23_WithCond}.
\end{proof}

Note that in \Theoremref{lowerbound} and
\Theoremref{conditionallowerbounds}, any choice of $p_{X'},p_{Y'},p_{X'Y'}$, $p_{X'Y},p_{XY'}$, $p_{X''Y},p_{XY''}$, $p_{X''Y'},p_{X'Y''}$ (with full support) yields a lower
bound. 
For a given function, while all choices do yield valid lower bounds,
one is often able to obtain the {\em best} of these lower bounds analytically (as in
\Theoremref{ot}, where it is seen to be the best as it matches an upper
bound) or numerically (as in \Theoremref{and}).

To summarize, for any secure computation problem $(p_{XY},p_{Z|XY})$, where $(p_{XY},p_{Z|XY})$
is in normal form, \Theoremref{prelim_lbs} gives lower bounds on entropies of the transcripts 
on all three links. If $p_{XY}$ has full support and $p_{Z|XY}$ is in normal form, then 
\Theoremref{lowerbound} gives improved lower bounds on entropies of the transcripts on all 
three links. In addition, if $p_{Z|XY}$ satisfies condition~1 of \Lemmaref{XYZ_inde_M123}, 
then \eqref{eq:M31_WithCond} gives further improvements on the lower bound for $H(M_{31})$; if $p_{Z|XY}$ 
satisfies condition~2 of \Lemmaref{XYZ_inde_M123}, then \eqref{eq:M23_WithCond} further improves the 
 lower bound for $H(M_{23})$. The fact that these improvements
can be strict will be seen through examples in
Section~\ref{subsec:examples}; see the paragraph following \Theoremref{and} therein.

{
\subsection{Lower Bounds on Randomness}\label{subsec:randomness}
In this section we provide lower bounds on the amount of randomness required in secure 
computation protocols. Although our focus in this paper is to prove communication lower 
bounds, it turns out that we may apply the above lower bounds on communication to derive 
bounds on the amount of randomness required. We show in \Subsectionref{examples} that they 
give tight bounds on randomness required for the specific functions we analyze.
\begin{thm}\label{thm:general_randomness}
For a secure computation problem $(p_{XY},p_{Z|XY})$, where $(p_{XY},p_{Z|XY})$ is in normal form, if $(R_{12},R_{23},R_{31},\rho)\in\R^{\PS}$, then
\begin{align*}
\rho &\geq \max\{RI(X;Y)+RI(X;Z), \\
&\hspace{2.5cm} RI(X;Y)+RI(Y;Z), \\
&\hspace{3.75cm} RI(X;Z)+RI(Y;Z)\} \\
&\quad + H(Y,Z|X)+H(X,Z|Y)+H(X,Y|Z)-H(X,Y).
\end{align*}
\end{thm}
\begin{proof}
Again, we prove this for $n=1$, and the general result follows from
\Lemmaref{RI_tensorizes}.
Fix a protocol $\Pi(p_{XY},p_{Z|XY})$. We bound the randomness required by this protocol as follows:
\begin{align}
\rho &= H(M_{12},M_{23},M_{31},Z|X,Y), \qquad(\text{by definition})\nonumber \\
& = H(M_{12},M_{23},M_{31}|X,Y,Z) + H(Z|X,Y) \nonumber \\
&\geq \max\{H(M_{12},M_{31}|X,Y,Z), H(M_{12},M_{23}|X,Y,Z), \nonumber \\
&\hspace{1.5cm} + H(M_{23},M_{31}|X,Y,Z)\} + H(Z|X,Y) \label{eq:ps_rand_interim} \\
&= \max\{H(M_{12},M_{31}|X), H(M_{12},M_{23}|Y), \nonumber \\
&\hspace{1.5cm} + H(M_{23},M_{31}|Z)\} + H(Z|X,Y) \label{eq:ps_rand_interim1} \\
&\geq \max\{H(M_{12},M_{31})-H(X), H(M_{12},M_{23})-H(Y) \nonumber \\
&\hspace{1.5cm} + H(M_{23},M_{31})-H(Z)\} + H(Z|X,Y), \label{eq:ps_rand_interim2}
\end{align}
where, in \eqref{eq:ps_rand_interim1} we used privacy conditions \eqref{eq:ps_privacy_alice}-\eqref{eq:ps_privacy_charlie}, 
and in \eqref{eq:ps_rand_interim2} we used $H(M_{12},M_{31}|X)=H(M_{12},M_{31},X)-H(X)\geq H(M_{12},M_{31})-H(X)$.
We can bound $H(M_{12},M_{31})$, $H(M_{12},M_{23})$, and $H(M_{23},M_{31})$ using the 
bounds we have already developed.
First we bound $H(M_{12},M_{31})$ as follows:
\begin{align}
\label{eq:M31-M12_bound}
H(M_{12},M_{31}) &= H(M_{12}) + H(M_{31}|M_{12}).
\end{align}
Notice that one of the bounds on $H(M_{31})$ in \Theoremref{prelim_lbs} was obtained by bounding $H(M_{31}|M_{12})$ as follows:
\begin{align}
\label{eq:M31-cond-M12_bound}
H(M_{31}|M_{12}) &\geq RI(X;Y) + H(X,Z|Y).
\end{align}
Substituting the bound on $H(M_{12})$ from \Theoremref{prelim_lbs} and bound on $H(M_{31}|M_{12})$ from \eqref{eq:M31-cond-M12_bound} into \eqref{eq:M31-M12_bound} we get:
\begin{align*}
H(M_{12},M_{31}) &\geq \max\{RI(X;Z),RI(Y;Z)\} + H(X,Y|Z) \\
&\qquad + RI(X;Y) + H(X,Z|Y).
\end{align*}
Substituting this into the first term inside the $\max$ in \eqref{eq:ps_rand_interim2} we get:
\begin{align*}
\rho &\geq \max\{RI(X;Z),RI(Y;Z)\} + H(X,Y|Z) + RI(X;Y) \\
&\quad + H(X,Z|Y) - H(X) + H(Z|X,Y) \\
&= \max\{RI(X;Y)+RI(X;Z),RI(X;Y)+RI(Y;Z)\}  \\
&\quad + H(Y,Z|X) + H(X,Z|Y) + H(X,Y|Z) - H(X,Y).
\end{align*}
If we expand $H(M_{12},M_{31})$ in another way as $H(M_{31})+H(M_{12}|M_{31})$ and proceed similarly as above, we get the following:
\begin{align*}
\rho &\geq \max\{RI(X;Y)+RI(X;Z),RI(X;Z)+RI(Y;Z)\} \\
&\quad + H(Y,Z|X) + H(X,Z|Y) + H(X,Y|Z) - H(X,Y).
\end{align*}
Combining the above two bounds gives the desired result. 
Notice that bounding $H(M_{12},M_{23})$ and $H(M_{23},M_{31})$ result in the same bound.
\end{proof}
If $p_{XY}$ has full support (which allows us to assume, without loss of generality, that 
$p_{Z|XY}$ is in normal form), we can strengthen \Theoremref{general_randomness}  
using the ideas from \Subsectionref{improved_lbs}. For instance, now 
we can bound $H(M_{12})$ from \Theoremref{lowerbound}. Notice that by \eqref{eq:ps_privacy_bob}, i.e., privacy against Alice, 
$(M_{12},M_{31})$ is conditionally independent of $Y$ conditioned on $X$. 
This means that the distribution of $(M_{12},M_{31})$ will not change if we change the input 
distribution to $p_{XY'}$ keeping the marginal of Alice's input the same.
Applying this observation to \eqref{eq:M31-cond-M12_bound} we get the following:
\begin{align}
H(M_{31}|M_{12}) &\geq \max_{p_{Y'|X}} RI(X;Y') + H(X,Z'|Y'), \label{eq:M31-cond-M12_ds-bound}
\end{align}
where the terms on the right hand side are evaluated using the following distribution:
\[p_{XY'Z'}(x,y,z)=p_{X}(x)\cdot p_{Y'|X}(y|x)\cdot p_{Z|XY}(z|x,y).\]
Similarly we can prove the following bounds, where right hand side expressions are evaluated using appropriate distributions:
\begin{align}
H(M_{23}|M_{12}) &\geq \max_{p_{X'|Y}} RI(X';Y) + H(Y,Z'|X'),\label{eq:M23-cond-M12_ds-bound}\\
H(M_{12}|M_{31}) &\geq \max_{p_{Y'|X}} RI(X;Z') + H(X,Y'|Z'),\label{eq:M12-cond-M31_ds-bound}\\
H(M_{12}|M_{23}) &\geq \max_{p_{X'|Y}} RI(Y;Z') + H(X',Y|Z'),\label{eq:M12-cond-M23_ds-bound}\\
H(M_{23}|M_{31}) &\geq \max_{p_{X'Y'}} RI(X';Z) + H(Y',Z|X'),\label{eq:M23-cond-M31_ds-bound}\\
H(M_{31}|M_{23}) &\geq \max_{p_{X'Y'}} RI(Y';Z) + H(X',Z|Y'),\label{eq:M31-cond-M23_ds-bound}
\end{align}
where the $p_{X'Y'}$ in the maximization in
\eqref{eq:M23-cond-M31_ds-bound} and \eqref{eq:M31-cond-M23_ds-bound} are
over distributions such that they result in the same marginal distribution of $Z$, i.e., $p_{X'Y'}$ should be such that the following holds for every value $z$ in the alphabet $\Z$:
\[\sum_{x,y}p_{X'Y'}(x,y)p_{Z|XY}(z|x,y)=\sum_{x,y}p_{XY}(x,y)p_{Z|XY}(z|x,y).\]
The above observations, along with \eqref{eq:ps_rand_interim2}, lead to the following theorem.
\begin{thm}\label{thm:rand_lowerbound}
For a secure computation problem $(p_{XY},p_{Z|XY})$, where $p_{XY}$ has full support and 
$p_{Z|XY}$ is in normal form, $(R_{12},R_{23},R_{31},\rho)\in\R^{\PS}$ only if
\begin{align*}
\rho &\geq H(M_{12}) + H(M_{31}|M_{12}) - H(X) + H(Z|XY), \\
\rho &\geq H(M_{12}) + H(M_{23}|M_{12}) - H(Y) + H(Z|XY), \\
\rho &\geq H(M_{31}) + H(M_{12}|M_{31}) - H(X) + H(Z|XY), \\
\rho &\geq H(M_{31}) + H(M_{23}|M_{31}) - H(Z) + H(Z|XY), \\
\rho &\geq H(M_{23}) + H(M_{12}|M_{23}) - H(Y) + H(Z|XY), \\
\rho &\geq H(M_{23}) + H(M_{31}|M_{23}) - H(Z) + H(Z|XY),
\end{align*}
where lower bounds on the entropy terms may be taken from \Theoremref{lowerbound}, and lower bounds on
the conditional entropy terms may be taken from \eqref{eq:M31-cond-M12_ds-bound}-\eqref{eq:M31-cond-M23_ds-bound}.
\end{thm}
Note that a consequence of the first bound in \Theoremref{rand_lowerbound},
combined with \eqref{eq:M31-cond-M12_ds-bound} (with $Y'$ taken to be
independent of $X$, so that $H(X,Z'|Y') \ge H(X)$), is that
\begin{align}
\label{eq:rho-M12}
\rho &\geq H(M_{12}).
\end{align}
When the function satisfies certain additional constraints, we can further strengthen 
\Theoremref{rand_lowerbound} using the improved bounds on $H(M_{31})$ and $H(M_{23})$ from 
\Theoremref{conditionallowerbounds}. 
The resulting bounds on randomness are stated in the following theorem.
\begin{thm}\label{thm:rand_conditionallowerbounds}
For a secure computation problem $(p_{XY},p_{Z|XY})$, where $p_{XY}$ has full support and 
$p_{Z|XY}$ is in normal form, if $(R_{12},R_{23},R_{31},\rho)\in\R^{\PS}$ then following must hold.
\begin{enumerate}
\item If $p_{Z|XY}$ satisfies condition 1 of \Lemmaref{XYZ_inde_M123}, then
\begin{align*}
\rho &\geq H(M_{31}) + H(M_{12}|M_{31}) - H(X) + H(Z|XY), \\
\rho &\geq H(M_{31}) + H(M_{23}|M_{31}) - H(Z) + H(Z|XY),
\end{align*}
where lower bound on $H(M_{31})$ may be taken from \Theoremref{conditionallowerbounds}, 
and lower bounds on the conditional entropy terms may be taken from \eqref{eq:M12-cond-M31_ds-bound} and~\eqref{eq:M23-cond-M31_ds-bound}.
\item If $p_{Z|XY}$ satisfies condition 2 of \Lemmaref{XYZ_inde_M123}, then 
\begin{align*}
\rho &\geq H(M_{23}) + H(M_{12}|M_{23}) - H(Y) + H(Z|XY), \\
\rho &\geq H(M_{23}) + H(M_{31}|M_{23}) - H(Z) + H(Z|XY),
\end{align*}
where lower bound on $H(M_{23})$ may be taken from \Theoremref{conditionallowerbounds}, 
and lower bounds on the conditional entropy terms may be taken from \eqref{eq:M12-cond-M23_ds-bound} and~\eqref{eq:M31-cond-M23_ds-bound}.
\end{enumerate}
\end{thm}

To summarize, for any secure computation problem $(p_{XY},p_{Z|XY})$, 
where $(p_{XY},p_{Z|XY})$ is in normal form, \Theoremref{general_randomness} gives a lower 
bound on the randomness. If $p_{XY}$ has full support and $p_{Z|XY}$ is in normal form, then 
\Theoremref{rand_lowerbound} gives an improved lower bound on randomness. In addition, if 
$p_{Z|XY}$ satisfies condition~1 or condition~2 of \Lemmaref{XYZ_inde_M123}, then lower 
bound of \Theoremref{rand_conditionallowerbounds} may, in some cases, be better than that of 
\Theoremref{rand_lowerbound}.

\begin{remark}\label{remark:ps_remark_rand}
{\em As argued in \Sectionref{prelims}, for communication requirements, if $p_{XY}$ has full support, we 
can assume, without loss of generality, that $p_{Z|XY}$ is in normal form. In case $p_{Z|XY}$ is not in normal form, we 
redefine the problem to $(p_{X^*Y^*},p_{Z^*|X^*Y^*})$, where $p_{Z^*|X^*Y^*}$ is in normal 
form, such that any secure protocol for the former can be transformed to a secure protocol for 
the latter {\em with the same communication costs, and vice versa}.
This implies that the communication lower bounds developed for the modified problem 
$(p_{X^*Y^*},p_{Z^*|X^*Y^*})$ are equally good for the original problem $(p_{XY},p_{Z|XY})$.
But this may not hold for randomness complexity, i.e., we do not know how to transform a 
secure protocol for $(p_{X^*Y^*},p_{Z^*|X^*Y^*})$ to a secure protocol for $(p_{XY},p_{Z|XY})$ 
with the same randomness requirement. 
Hence, although the lower bounds for randomness developed for the problem 
$(p_{X^*Y^*},p_{Z^*|X^*Y^*})$ also serve as valid lower bounds for the original problem 
$(p_{XY},p_{Z|XY})$, improvements may be possible by considering the original problem.}
\end{remark}
}

\subsection{Application to Specific Functions}\label{subsec:examples}
In this section we consider a few important examples where we will apply
our generic lower bounds from above and also give secure protocols. In some cases
we will obtain the optimal rate regions. While some of these results are
natural to conjecture, they are not easy to prove (see, for
instance,~\Footnoteref{addition}). In all the examples except
one\footnote{The exceptional case arises as a special case of {\sc controlled erasure} where
we need to exploit blocks of input for source compression.}, our protocols are for a
block-length of 1; hence, where they are optimal, $\R^{1,\PS}=\R^{\PS}$.
Before presenting our examples, we make the following two definitions:

\paragraph{Communication-Ideal Protocol.}
We say that a protocol $\Pi(p_{XY},p_{Z|XY})$ for securely computing a randomized function 
$p_{Z|XY}$ for a distribution $p_{XY}$ is {\em communication-ideal}, if for each 
$ij \in \{12,23,31\}$, \[ H(M^{\Pi}_{ij}) = \inf_{\Pi'(p_{XY},p_{Z|XY})} H(M^{\Pi'}_{ij}), \]
where the infimum is over all secure protocols for $p_{Z|XY}$ with the same distribution 
$p_{XY}$. 
That is, a communication-ideal protocol achieves the least entropy possible for 
every link, simultaneously. We remark that it is not clear, {\em a priori}, how to determine 
if a given function $p_{Z|XY}$ has a communication-ideal protocol for a given distribution 
$p_{XY}$.

\paragraph{Randomness-Optimal Protocol.}
We say that a protocol $\Pi(p_{XY},p_{Z|XY})$ for securely computing a randomized function 
$p_{Z|XY}$ for a distribution $p_{XY}$ is {\em randomness-optimal}, if
\[ \rho^{\Pi}(p_{XY},p_{Z|XY}) = \inf_{\Pi'(p_{XY},p_{Z|XY})} \rho^{\Pi'}(p_{XY},p_{Z|XY}), \]
where the infimum is over all secure protocols for $p_{Z|XY}$ with the same distribution 
$p_{XY}$. That is, a protocol is randomness optimal if the randomness (measured in bits) used by 
the protocol is the least among all protocols.
As defined in the beginning of \Sectionref{ps_lowerbounds}, the amount of randomness required 
by a protocol $\Pi$ is
$\rho^{\Pi}(p_{XY},p_{Z|XY})=H(M_{12},M_{23},M_{31},Z|X,Y)$.\\

{\bf 1. Secure Computation of Remote Oblivious Transfer {\sc remote $\binom{m}{1}$-OT$^n_2$}:}
The {\sc remote $\binom{m}{1}$-OT$^n_2$} function, is defined as follows:
Alice's input $X=(X_0,X_1,\hdots,X_{m-1})$ is made up of $m$ bit-strings
each of length $n$, and Bob has an input $Y\in \{0,1,\hdots,m-1\}$. Charlie
wants to compute $Z=f(X,Y)=X_Y$. This can be seen as a 3 user variant of
oblivious transfer~\cite{Wiesner83,EvenGL85}.
\Figureref{ot} gives the simple protocol for this function from \cite{FeigeKiNa94}
(rephrased as a protocol in our model). It requires $nm$ bits to be
exchanged over the Alice-Charlie (31) link, $n+\log m$  bits over the
Bob-Charlie (23) link, and $nm+\log m$ bits over the Alice-Bob (12) link. The total number of random bits used in the protocol is $nm+\log m$. We show that this protocol achieves the optimal rate region, i.e., it is a {\em communication-ideal} as well as {\em randomness-optimal} protocol.
\begin{figure}[htb]
\hrule height 1pt
\vspace{.05cm}
{\bf Algorithm 1:} {Secure Computation of \textsc{remote $\binom{m}{1}$-OT$^n_2$}} 
\hrule
\begin{algorithmic}[1]
\REQUIRE Alice has $m$ input bit strings $X_0,X_1,\hdots,X_{m-1}$ each of length $n$ \& Bob has an input $Y \in \{0,1,\hdots,m-1\}$.
\ENSURE Charlie securely computes the {\sc remote $\binom{m}{1}$-OT$^n_2$}: $Z=X_Y$.

\medskip

\STATE Alice samples $nm + \log m$ independent, uniformly distributed bits from her private randomness. Denote the first $m$ blocks each of length $n$ of this random string by $K_0,K_1,\hdots,K_{m-1}$ and the last $\log m$ bits by $\pi$. Alice sends it to Bob as $M_{\vec{12},1}=(K_0,K_1,\hdots,K_{m-1},\pi)$.

\STATE Alice computes $M^{(i)}=X_{\pi+i\; (\text{mod }m)}\oplus K_{\pi+i\; (\text{mod } m)}, \quad i\in\{0,1,\hdots,m-1\}$ and sends to Charlie 
$M_{\vec{13},2}=(M^{(0)},M^{(1)},\hdots,M^{(m-1)})$. Bob computes $C=Y-\pi\; (\text{mod } m), K=K_Y$ and sends to Charlie $M_{\vec{23},2}=(C, K)$.

\STATE Charlie outputs $Z=M^{(C )}\oplus K$. 

\end{algorithmic}

\hrule
\caption{A protocol to securely compute {\sc remote
$\binom{m}{1}$-OT$^n_2$}, which is a special case of the general protocol
given in \cite{FeigeKiNa94}. The protocol requires $nm$ bits to be
exchanged over the Alice-Charlie (13) link, $n+\log m$  bits over the
Bob-Charlie (23) link, and $nm+\log m$ bits over the Alice-Bob (12) link. We
show optimality of our protocol by showing that any protocol must exchange
an expected $nm$ bits  over the Alice-Charlie (31) link, $n+\log m$ bits
over the Bob-Charlie (23) link, and $nm+\log m$ bits over the Alice-Bob (12)
link.}
\label{fig:ot}
\end{figure}
\begin{thm}\label{thm:ot}
Any secure protocol for computing {\sc remote $\binom{m}{1}$-OT$^n_2$} for inputs $X$ and $Y$, where $p_{XY}$ has full support, must satisfy
\begin{align*}
R_{31} \geq\ &nm, \quad R_{23} \geq n+\log m, \text{ and}\\
&R_{12}, \rho \geq nm+\log m.
\end{align*}
Hence the protocol in \Figureref{ot} is optimal.
\end{thm}
\begin{proof}
\textsc{remote $\binom{m}{1}$-OT$^n_2$} satisfies Condition~1 and Condition~2 of \Lemmaref{XYZ_inde_M123}, which implies $RI(Y;Z)=I(Y;Z)$ and $RI(Z;X)=I(Z;X)$. For $H(M_{31})$ and $H(M_{23})$ the bottom rows in~\eqref{eq:M31_WithCond} and~\eqref{eq:M23_WithCond} simplify to the following:
\begin{align*}
H(M_{31}) \geq \max_{p_{X'Y'}} I(Y';Z') + H(X'|Y'), \\
H(M_{23}) \geq \max_{p_{X'Y'}} I(X';Z') + H(Y'|X').
\end{align*}
Taking $X'$ and $Y'$ to be independent and uniform in their respective domains, we get $H(M_{31}) \geq nm$.
To derive a lower bound on $H(M_{23})$, take $X',Y'$ to be independent with $Y' \sim \text{unif}\{0,1,\hdots,m-1\}$ and $X'$ distributed as below:
\begin{align*}
&p_{X_0',X_1',\hdots,X_{m-1}'}(x_0,x_1,\hdots,x_{m-1}) \\
&\qquad \qquad =\left\{\begin{array}{ll} \frac{1}{2^n}-\epsilon, & x_0=x_1=\hdots=x_{m-1},\\ 
\epsilon/(2^{n(m-1)}-1),& \text{otherwise},\end{array}\right.
\end{align*}
where $\epsilon>0$ can be made arbitrarily small to make $I(Z';X')$ as close
to $n$ as desired. This gives a bound of $H(M_{23})\geq n+\log m$. For $H(M_{12})$, the bottom row of \eqref{eq:M12_bound} simplifies to
\begin{align*}
H(M_{12}) &\geq \max_{p_{X'}p_{Y'}} I(X';Z') \\
&\qquad \qquad + \max_{p_{X''Y'}} I(Y';Z'') + H(X'',Y'|Z'').
\end{align*}
Taking $X''$ and $Y'$ to be independent and uniform (in the second maximum), and $X'$ to be distributed as below
\begin{align*}
&p_{X_0',X_1',\hdots,X_{m-1}'}(x_0,x_1,\hdots,x_{m-1}) \\
&\qquad \qquad =\left\{\begin{array}{ll} \frac{1}{2^n}-\epsilon, & x_0=x_1=\hdots=x_{m-1},\\ 
\epsilon/(2^{n(m-1)}-1),& \text{otherwise},\end{array}\right.
\end{align*}
where $\epsilon>0$ can be made arbitrarily small to make $I(X';Z')$ as close to $n$ as desired. This gives a bound of $H(M_{12})\geq nm+\log m$.

Finally, from \eqref{eq:rho-M12} and the above bound on $H(M_{12})$, we get
$\rho\ge H(M_{12}) \ge nm+\log m$.  This implies that the above protocol is
randomness-optimal.
\end{proof}

The above optimality result for {\sc remote ot} function has the following implications:
\begin{enumerate}
\item[(i)]{\em Optimality of the FKN Protocol.} Feige et al.\ \cite{FeigeKiNa94}
provided a generic (non-interactive) secure computation protocol for all
3-user functions in our model. This protocol uses a straight-forward (but
``inefficient'') reduction from an arbitrary function to the remote OT function, 
and then gives a simple protocol for the remote OT function. 
While the resulting protocol turns out to be suboptimal for most functions, \Theoremref{ot}
shows that the protocol that \cite{FeigeKiNa94} used for {\sc remote OT}
itself is optimal. 

\item[(ii)] {\em Separating Secure and Insecure Computation.}
A basic question of secure computation is whether it needs more bits to be
communicated to the user who wants to learn the output than the input-size itself (which suffices for insecure
computation). While natural to expect, it is not easy to prove this.
\cite{FeigeKiNa94}, in their restricted
model\footnote{Recall that the model of~\cite{FeigeKiNa94} can be thought of as our
protocol model with the following restrictions: (i) Alice and Bob share a
common random variable independent of their inputs, but are otherwise
unable communicate with each other. (ii) Alice and Bob may send only one
message each to Charlie who may not send any messages to the other users.},
showed a non-explicit result, that for securely computing {\em most}
Boolean functions on the domain $\{0,1\}^n \times \{0,1\}^n$, Charlie is
required to receive at least $3n-4$ bits, which is significantly more than
the $2n$ bits sufficient for insecure computation. 

{\sc remote $\binom{2}{1}$-OT$^n_2$} from above already gives us an
explicit example of a function where this is true: the total input size is
$2n+1$, but the communication is at least $H(M_{31})+H(M_{23}) \ge 3n+1$.
To present an easy comparison to the lower bound of \cite{FeigeKiNa94}, we
can consider a symmetrized variant of {\sc remote $\binom{2}{1}$-OT$^n_2$},
in which two instances of {\sc remote $\binom{2}{1}$-OT$^n_2$} are
combined, one in each direction. More specifically, $X=(A_0,A_1,a)$ where
$A_0,A_1$ are of length $(n-1)/2$ (for an odd $n$) and $a$ is a single bit;
similarly $Y=(B_0,B_1,b)$; the output of the function is defined as an
$(n-1)$ bit string $f(X,Y)=(A_b,B_a)$. Considering (say) the uniform input
distribution over $X,Y$, the bounds for {\sc remote
$\binom{2}{1}$-OT$^n_2$} add up to give us $H(M_{31}) \ge 3(n-1)/2+1$ and
$H(M_{23}) \ge 3(n-1)/2+1$, so that the communication with Charlie is
lower-bounded by $H(M_{31})+H(M_{23}) \ge 3n-1$.

This compares favourably with the bound of \cite{FeigeKiNa94} in many ways:
our lower bound holds even in a model that allows interaction; in particular,
this makes the gap between insecure computation ($n-1$ bits in our case,
$2n$ bits for \cite{FeigeKiNa94}) and secure computation (about $3n$ bits
for both) somewhat larger. More importantly, our lower bound is explicit (and
tight for the specific function we use), whereas that of \cite{FeigeKiNa94}
is existential. However, our bound does not subsume that of
\cite{FeigeKiNa94}, who considered {\em Boolean} functions. Our results do
not yield a bound significantly larger than the input size, when the output
is a single bit. It appears that this regime is more amenable to
combinatorial arguments, as pursued in \cite{FeigeKiNa94}, rather than
information theoretic arguments. 
We leave it as a fascinating open problem to
obtain tight bounds in this regime, possibly by combining combinatorial and
information-theoretic approaches.
\end{enumerate}

\renewcommand{\star}{+}
{\bf 2. Secure Computation of {\sc group-add}:}
Let $\mathbb{G}$ be a (possibly non-abelian) group with binary operation $\star$. The 
function {\sc group-add} is defined as follows: Alice has an input $X\in\mathbb{G}$, Bob has an input
$Y\in\mathbb{G}$ and Charlie should get $Z=f(X,Y)=X \star Y$.

In \Figureref{group}, we recapitulate a well-known simple protocol for
securely computing the above function. The protocol requires a
$|\mathbb{G}|$-ary symbol to be exchanged per computation over each link. We show below that 
this protocol is a communication-ideal as well as randomness-optimal for any input
distribution with full support.  As mentioned in
\Footnoteref{addition}, this is easy to see for the uniform distribution,
and using distribution switching, we can see that the same holds as long as
the input distribution has full support.
\begin{figure}[htb]
\hrule height 1pt
\vspace{.05cm}
{\bf Algorithm 2:} {Secure Computation of {\sc group-add}} 
\hrule
\begin{algorithmic}[1]
\REQUIRE Alice \& Bob have input $X,Y\in\mathbb{G}$, respectively.
\ENSURE Charlie securely computes $Z=X\star Y$.

\medskip

\STATE Charlie samples a uniformly distributed element $K$ from $\mathbb{G}$ using his private randomness; sends it to Bob as $M_{\vec{32}}=K$.

\STATE Bob sends $M_{\vec{21}}=Y \star M_{\vec{32}}$ to Alice.

\STATE Alice sends $M_{\vec{13}}=X \star M_{\vec{21}}$ to Charlie.

\STATE Charlie outputs $Z= M_{\vec{13}} - K$.
\end{algorithmic}
\hrule
\caption{An optimal protocol for secure computation in any group
$\mathbb{G}$. The protocol requires a $|\mathbb{G}|$-ary symbol to be exchanged over each link.}
\label{fig:group}
\end{figure}
\begin{thm}\label{thm:group}
Any secure protocol for computing in a Group $\mathbb{G}$, where $p_{XY}$ has full support over $\mathbb{G}\times\mathbb{G}$, must satisfy 
\begin{align*}
R_{12}, R_{23}, R_{31},\rho \geq \log|\mathbb{G}|.
\end{align*}
Hence the protocol in \Figureref{group} is optimal.
\end{thm}
\begin{proof}
It is easy to see that the above function satisfies Condition~1 and Condition~2 of \Lemmaref{XYZ_inde_M123}. 
We will only need the last terms (corresponding to the na\"ive bounds $H(X',Y''|Z'')$ etc., but
with distribution switching) of \eqref{eq:M12_bound}, \eqref{eq:M31_WithCond}, and \eqref{eq:M23_WithCond} for $H(M_{12})$, $H(M_{31})$, and $H(M_{23})$, respectively.
Since we are computing a deterministic function, and $Y$ can be determined from $(X,Z)$, the last terms in each of the these bounds will reduce to the following:
\begin{align*}
H(M_{12}) &\geq \max_{p_{X'Y''}}  H(X'|Z''),\\
H(M_{31}) &\geq \max_{p_{X'Y''}}  H(X'|Y''),\\
H(M_{23}) &\geq \max_{p_{X''Y'}}  H(Y'|X'').
\end{align*}
The optimum bounds for $M_{12}$, $M_{31}$, and $M_{23}$ are obtained by evaluating all the expressions above with product distributions, where each random variable is uniformly distributed over $\mathbb{G}$; this gives $H(M_{12}), H(M_{31}), H(M_{23}) \geq \log|\mathbb{G}|$.

Finally, from \eqref{eq:rho-M12} and the above bound on $H(M_{12})$, we get
$\rho \ge H(M_{12}) \ge \log |\mathbb{G}|$, which implies that the above protocol is randomness-optimal.
\end{proof}

{\bf 3. Secure Computation of {\sc controlled erasure}:}
The controlled erasure function is defined as follows: Alice and Bob have bits $X$ and $Y$, respectively. Alice's input $X$ acts as
the ``control'', which decides whether Charlie receives an erasure ($\Delta$) or Bob's input $Y$. 
\begin{center}
\begin{tabular}{ccc}
\toprule
& \multicolumn{2}{c}{y}\\
\cmidrule(r){2-3} 
x&\multicolumn{1}{c}{0}&\multicolumn{1}{c}{1}\\
\midrule
0& $\Delta$ & $\Delta$\\
1& 0 & 1\\
\bottomrule
\end{tabular}
\end{center}
Notice that Charlie always finds out Alice's control bit, but does not learn
Bob's bit when it is erased. This function does not satisfy Condition~1 of
\Lemmaref{XYZ_inde_M123}.

\Figureref{erasure} gives a protocol for securely computing this
function on each location of strings of length $n$. Bob sends his input
string to Charlie under the cover of a one-time pad and reveals the key
used to Alice. Alice sends his input to Charlie compressed using a Huffman
code (replaced by Lempel-Ziv if we want the protocol to be distribution
independent). He also sends to Charlie those key bits he received
from Bob that corresponds to the locations where there is no erasure (i.e., where his
input bit is 1). When $X\sim\text{Bernoulli}(p)$ and
$Y\sim\text{Bernoulli}(q)$, i.i.d., where $p,q\in(0,1)$, the expected
message length for Alice-Charlie link is ${\mathbb E}[L_{31}] < nH_2(p)+1
+ np$, the messages lengths on the other two links are deterministically
$n$ each, $L_{12}=L_{23}=n$. Here we prove the optimality of this protocol
for $X\sim\text{Bernoulli}( p)$ and $Y\sim\text{Bernoulli}( q)$, where $p,q
\in (0,1)$. We also prove that this protocol is randomness-optimal.
\begin{figure}[htb]
\hrule height 1pt
\vspace{.05cm}
{\bf Algorithm 3:} {Secure Computation of \textsc{controlled erasure}} 
\hrule
\begin{algorithmic}[1]
\REQUIRE Alice \& Bob have input bits $X^n,Y^n\in\{0,1\}^n$.
\ENSURE Charlie securely computes the {\sc controlled erasure} function
\[Z_i=f(X_i,Y_i),\qquad i=1,\ldots,n.\]

\medskip

\STATE Bob samples $n$ i.i.d. uniformly distributed bits $K^n$ from his private randomness; sends it to Alice as $M_{\vec{21},1}=K^n$. Bob sends to Charlie his input $Y^n$ masked (bit-wise) with $K^n$ as $M_{\vec{23},1}=Y^n\oplus K^n$. 

\STATE Alice sends his input $X^n$ to Charlie compressed using a Huffman
code (or Lempel-Ziv if we want the protocol to not depend on the input
distribution of $X^n$); let $c(X^n)$ be the codeword. Alice also sends to
Charlie the sequence of key bits $K_i$ corresponding to the locations where
his input $X_i$ is 1.
\[ M_{\vec{12},2}= c(X^n),(K_i)_{i:X_i=1}.\]

\STATE Charlie outputs \[Z_i=\left\{ \begin{array}{ll}\Delta& \text{ if } X_i=0,\\ (Y_i\oplus K_i)\oplus K_i& \text{ if } X_i=1.\end{array}\right.\]
\end{algorithmic}
\hrule

\caption{A protocol to compute {\sc controlled erasure} function. For
$X\sim\text{Bernoulli}(p)$ and $Y\sim\text{Bernoulli}(q)$, both i.i.d and
$p,q\in(0,1)$, the
expected message lengths are ${\mathbb E}[L_{31}] < n(H_2(p)+p)+1$,
$L_{12} = n$, and $L_{23} =n$. We show that these are asymptotically
optimal by showing the following lower bounds: $H(M_{31}) \geq n(H_2(p) +
p)$, $H(M_{12})\geq n$, and $H(M_{23})\geq n$.}
\label{fig:erasure}
\end{figure}

\begin{thm}\label{thm:erasure}
Any secure protocol for computing {\sc controlled erasure} with inputs $X^n,Y^n$, where $(X_i,Y_i)\sim p_{XY}$, i.i.d., has full support, and induced $X\sim$ {\em Bernoulli}$(p)$ and $Y\sim$ {\em Bernoulli}$(q)$ with $p,q\in(0,1)$, must satisfy
\begin{align*}
R_{31} \geq n(H_2(p) + p),\quad R_{12}, R_{23},\rho &\geq n.
\end{align*}
Hence the protocol in \Figureref{erasure} is optimal.
\end{thm}
\begin{proof}
It is easy to see that this function satisfies only Condition~2 of \Lemmaref{XYZ_inde_M123}, which implies $RI(Y;Z)=I(Y;Z)$; but Condition~1 of \Lemmaref{XYZ_inde_M123} is not satisfied -- in fact $RI(X;Z)=0$. Our best bounds for $H(M_{31})$ and $H(M_{23})$ are given by \eqref{eq:M31_NoCond} and \eqref{eq:M23_WithCond}, respectively. The bottom row of \eqref{eq:M31_NoCond} simplifies to the following:
\begin{align*}
H(M_{31}) &\geq \max_{p_{{X}^n{Y'}^n}}I({Y'}^n;{Z'}^n) + H(X^n|{Y'}^n).
\end{align*}
The optimum bound for $H(M_{31})$ is obtained by taking ${Y'}^n$, i.i.d., Bernoulli(1/2), independent of $X^n$; this gives $H(M_{31}) \geq n(p+H_2(p ))$. For $H(M_{23})$, the bottom row of \eqref{eq:M23_WithCond} simplifies to the following:
\begin{align*}
H(M_{23}) &\geq \max_{p_{{X'}^n}p_{{Y'}^n}} H({Y'}^n | {X'}^n).
\end{align*}
Taking ${Y'}^n$, i.i.d., Bernoulli(1/2), independent of ${X'}^n$ gives $H(M_{23}) \geq n$. For $H(M_{12})$, putting $RI(X';Z')=0$ in the bottom row of \eqref{eq:M12_bound} and simplifying further, we get the following:
\begin{align*}
H(M_{12}) &\geq \max_{p_{{X''}^n{Y'}^n}} I({Y'}^n;{Z''}^n) + H({X''}^n,{Y'}^n|{Z''}^n).
\end{align*}
Since $X''$ is a function of $Z''$ for {\sc controlled-erasure}, the above expression simplifies to $H(M_{12})\geq \sup_{p_{{Y'}^n}}H({Y'}^n)$, which gives $H(M_{12})\geq n$ by taking ${Y'}^n$, i.i.d., Bernoulli(1/2).

Finally, from \eqref{eq:rho-M12} and the above bound on $H(M_{12})$, we get
$\rho \ge H(M_{12}) \ge n$, which implies that the above protocol is randomness-optimal.
\end{proof}

{\bf 4. Secure Computation of {\sc sum}:}
The {\sc sum} function is defined as follows: Alice and Bob have one bit input $X\in\{0,1\}$ and $Y\in\{0,1\}$, respectively, and Charlie wants to compute the arithmetic sum $Z=f(X,Y)=X+Y$. \Figureref{sum} recapitulates a simple protocol for this function. This protocol requires a ternary symbol to be exchanged per computation over each link. We show in below that our bounds give $H(M_{31}),H(M_{23})\geq \log(3)$ and $H(M_{12})\geq 1.5$. Thus, while the protocol matches the lower bound on $H(M_{31})$ and $H(M_{23})$, there is a gap for $H(M_{12})$; while the protocol requires $H(M_{12})=\log(3)$, the lower bound is only $H(M_{12})\geq1.5$. We also show that this protocol is randomness-optimal, which proves a recent conjecture of \cite{LeeAb14} for three users. 
For $U,V\in\{0,1,2\}$, we write $U+V$ to denote the addition modulo-3.
\begin{figure}[htb]
\hrule height 1pt
\vspace{.05cm}
{\bf Algorithm 4:} {Secure Computation of \textsc{sum}}
\hrule
\begin{algorithmic}[1]
\REQUIRE Alice and Bob have input $X,Y\in\{0,1\}$, respectively.
\ENSURE Charlie securely computes {\sc sum} $Z=X+Y$.

\medskip

\STATE Charlie samples a uniformly distributed element $K$ from $\{0,1,2\}$ using his private randomness; sends it to Alice as $M_{\vec{31}}=K$.

\STATE Alice sends $M_{\vec{12}}=M_{\vec{31}}+X$ to Bob.

\STATE Bob sends $M_{\vec{23}}=M_{\vec{12}}+Y$ to Charlie.

\STATE Charlie outputs $Z=M_{\vec{23}}-K$.
\end{algorithmic}
\hrule
\caption{A protocol to compute {\sc sum}. The protocol requires a ternary symbol to be exchanged over all the three links. We show a lower bound of $\log(3)$ both on Alice-Charlie and Bob-Charlie links and a lower bound of 1.5 on Alice-Bob link.}
\label{fig:sum}
\end{figure}
\begin{thm}\label{thm:sum}
Any secure protocol for computing {\sc sum}, where $p_{XY}$ has full support over $\{0,1\}\times\{0,1\}$, must satisfy 
\begin{align*}
R_{31},R_{23},\rho \geq \log(3) \text{ and } R_{12} \geq 1.5.
\end{align*}
\end{thm}
\begin{proof}
It is easy to see that \textsc{sum} satisfies Condition~1 and Condition~2 of \Lemmaref{XYZ_inde_M123}, which implies $RI(Y;Z)=I(Y;Z)$ and $RI(Z;X)=I(Z;X)$. Since $X$ can be determined from $(Y,Z)$ and $Y$ can be determined from $(X,Z)$, the bottom rows of the bounds in~\eqref{eq:M31_WithCond} and~\eqref{eq:M23_WithCond} for $H(M_{31})$ and $H(M_{23})$, respectively, simplify to the following:
\begin{align*}
H(M_{31}) &\geq \max_{p_{X'Y'}} H(Z'), \\
H(M_{23}) &\geq \max_{p_{X'Y'}} H(Z').
\end{align*}
For $H(M_{31})$, taking $p_{X'Y'}(0,0)$ = $p_{X'Y'}(1,1)=1/3$ and $p_{X'Y'}(0,1)$ = $p_{X'Y'}(1,0)=1/6$ gives $H(M_{31}),H(M_{23})\geq \log(3)$. For $H(M_{12})$, the bound in top row in \eqref{eq:M12_bound} simplifies to
\begin{align*}
H(M_{12}) &\geq \max_{p_{X'}p_{Y'}}\left( H(Z') - H(X') \right) + \max_{p_{X'Y''}} H(X').
\end{align*}
Taking $X',Y'\sim$ Bern(1/2) gives $H(M_{12})\geq 1.5$.

For the {\sc sum} function the first bound in \Theoremref{rand_conditionallowerbounds} 
simplifies to $\rho\geq H(M_{31})$. This together with the above bound on $H(M_{31})$ gives 
$\rho \geq \log(3)$, which implies randomness-optimality of the above protocol.
\end{proof}

{\bf 5. Secure Computation of {\sc and}:}
The {\sc and} function is defined as follows:
Alice and Bob have one bit input $X\in\{0,1\}$ and $Y\in\{0,1\}$, respectively, and  
Charlie wants to compute $Z=f(X,Y)=X\wedge Y$.
The best known protocol for {\sc and} first appeared in \cite{FeigeKiNa94}, and 
we recapitulate it here in \Figureref{and} (rephrased as a protocol in our model). 
This protocol requires a ternary symbol to be exchanged over Alice-Charlie (13) and Bob-Charlie (23) links, and symbols from an alphabet of size 6 over the Alice-Bob (12) link. We show in below that our bounds give $H(M_{31}),H(M_{23})\geq \log(3)$ and $H(M_{12}),\rho\geq 1.826$. Thus, while the protocol matches the lower bound on $H(M_{31})$ and $H(M_{23})$, there is a gap for $H(M_{12})$ and $\rho$; 
while the protocol requires $H(M_{12})=\rho=\log(6)\approx 2.585$, the lower bound is only $H(M_{12}),\rho\geq1.826$.
\begin{figure}[htb]
\hrule height 1pt
\vspace{.05cm}
{\bf Algorithm 5:} {Secure Computation of \textsc{and}}
\hrule
\begin{algorithmic}[1]
\REQUIRE Alice has an input bit $X$ \& Bob has a bit $Y$.
\ENSURE Charlie securely computes the {\sc and} $Z=X\wedge Y$.

\medskip

\STATE Alice samples a uniform random permutation $(\alpha,\beta,\gamma)$
of $(0,1,2)$ from her private randomness; sends it to Bob
$M_{\vec{12}}=(\alpha,\beta,\gamma)$ (using a symbol from an alphabet of
size 6).

\STATE Alice sends $\alpha$ to Charlie if $X=1$, and $\beta$ if $X=0$. Bob
sends $\alpha$ to Charlie if $Y=1$ and $\gamma$ if $Y=1$.
\begin{align*}
M_{31}=\begin{cases}
\alpha \qquad \text{if }X=1, \\
\beta  \qquad \text{if }X=0,
\end{cases}
&
\qquad
M_{23}=\begin{cases}
\alpha \qquad \text{if }Y=1, \\
\gamma  \qquad \text{if }Y=0.
\end{cases}
\end{align*}

\STATE Charlie outputs $Z=1$ if $M_{31}=M_{23}$, and 0 otherwise.
\end{algorithmic}
\hrule
\caption{A protocol to compute {\sc and} \cite{FeigeKiNa94}. The protocol requires a ternary symbol to be exchanged over the Alice-Charlie (13) and Bob-Charlie (23) links and symbols from an alphabet of size 6 over the Alice-Bob (12) link.}
\label{fig:and}
\end{figure}
\begin{thm}\label{thm:and}
Any secure protocol for computing {\sc and} for inputs $X$ and $Y$, where $p_{XY}$ has full support over $\{0,1\}\times\{0,1\}$, must satisfy
\begin{align*}
R_{31},R_{23} \geq \log(3) \text{ and } R_{12}, \rho \geq 1.826.
\end{align*}
\end{thm}
\begin{proof}It is easy to see that \textsc{and} satisfies Condition~1 and Condition~2 of \Lemmaref{XYZ_inde_M123}, which implies $RI(Y;Z)=I(Y;Z)$ and $RI(Z;X)=I(Z;X)$. For $H(M_{31})$ and $H(M_{23})$ the bottom rows of \eqref{eq:M31_WithCond} and \eqref{eq:M23_WithCond} simplify to the following:
\begin{align*}
H(M_{31}) \geq \max_{p_{X'Y'}} I(Y';Z') + H(X',Z'|Y'), \\
H(M_{23}) \geq \max_{p_{X'Y'}} I(X';Z') + H(Y',Z'|X').
\end{align*}
For $H(M_{31})$, take $p_{X'Y'}(0,0)=p_{X'Y'}(1,0)=p_{X'Y'}(1,1)=(1-\epsilon)/3$ and $p_{X'Y'}(0,1)=\epsilon$, where $\epsilon>0$ can be made arbitrarily small to make $H(M_{31})$ as close to $\log(3)$ as we desire. For $H(M_{23})$, take $p_{X'Y'}(0,0)=p_{X'Y'}(0,1)=p_{X'Y'}(1,1)=(1-\epsilon)/3$ and $p_{X'Y'}(1,0)=\epsilon$, where $\epsilon>0$ can be made arbitrarily small to make $H(M_{23})$ as close to $\log(3)$ as we desire.

For $H(M_{12})$, the top row of \eqref{eq:M12_bound} simplifies to
\begin{align*}
H(M_{12}) &\geq \max_{p_{X'}p_{Y'}} I(Y';Z') \\
&\qquad \qquad + \max_{p_{X'Y''}}I(X';Z'') + H(X',Y''|Z'').
\end{align*}
The expression after the second maximum simplifies to $H(X') + p_{X'}(0)H(Y''|X'=0)$, which is always upper-bounded by $H(X') + p_{X'}(0)$ and can be made equal to this by taking $Y''\sim$ Bernoulli(1/2) and independent of $X'$. Now taking $p_{X'}(1)=0.456$ and $p_{Y'}(1)=0.397$ gives $H(M_{12}) \geq 1.826$.

Finally, from \eqref{eq:rho-M12} and the above bound on $H(M_{12})$, we get
$\rho \ge H(M_{12}) \ge 1.826$, whereas the protocol requires $1+\log 3 \approx 2.585$ random bits.
\end{proof}
Here we explicitly show, through the above example {\sc and}, the progressive improvements on the communication lower bounds from applying \Theoremref{prelim_lbs} to \Theoremref{conditionallowerbounds}. Let $X,Y$ be i.i.d. binary random variables distributed uniformly in $\{0,1\}$. For the secure computation of {\sc and} for this input distribution, \Theoremref{prelim_lbs} gives $(R_{31},R_{23},R_{12})\geq(1.311,1.311,1.5)$, \Theoremref{lowerbound} gives $(R_{31},R_{23},R_{12})\geq(1.5,1.5,1.826)$, and \Theoremref{conditionallowerbounds} gives $(R_{31},R_{23},R_{12})\geq(\log(3),\log(3),1.826)$. Notice that in \Theoremref{conditionallowerbounds} we only improve bounds on $R_{31},R_{23}$ over \Theoremref{lowerbound}.

{\em Separating Secure Computation and Secret Sharing:}
Another natural separation one expects is between the amount of
communication needed when the views (or transcripts) are generated by a
secure computation protocol, versus when they are generated by an
omniscient ``dealer'' so that the security requirements are met. The latter
setting corresponds to the share sizes in a CMSS scheme (see
\Appendixref{CMSS_sampling}). Again, while such a separation is expected,
it is not very easy to establish this, especially with explicit examples.
It requires us to establish a strong lower bound for the secure computation
problem as well as provide a CMSS scheme that is better.

We establish the separation using the 3-user {\sc and} function. There is a
CMSS scheme that achieves $\log(3)\le 1.6$ bits of entropy for all three
shares $M_{12},M_{23}$, and $M_{31}$ (see \Theoremref{gap_CMSS} in \Appendixref{CMSS_sampling}). However,
\Theoremref{and} shows that in a secure computation protocol, $H(M_{12})$
should be strictly larger than this.  

{\bf Note:} We need the use of \Lemmaref{infoineq} (information inequality) only to improve the bound on $H(M_{12})$, in particular for {\sc sum, remote-ot}, and {\sc and}. Bounds on the other two links in all the functions above do not need the use of information inequality.

\section{Outer Bounds on the Rate-Region for Asymptotically Secure Computation}\label{sec:asymp_lowerbounds}
In this section we restrict ourselves to the secure computation of deterministic functions $f:\X\times\Y\to\Z$, i.e., where $p_{Z|XY}$ is a deterministic mapping of the inputs to the output. We consider block-wise computation (with block length $n$), where Alice has input $X^n=(X_1,X_2,\hdots,X_n)$, Bob has input $Y^n=(Y_1,Y_2,\hdots,Y_n)$, and Charlie wants to compute the output $Z^n=(Z_1,Z_2,\hdots,Z_n)$, where $Z_i=f(X_i,Y_i)$ and $(X_i,Y_i)\sim p_{XY}$, i.i.d.; see \Figureref{asymp_setup}. Protocol is allowed to be asymptotically secure, i.e., it can make an error in the function computation -- Charlie may produce an output $\hat{Z}^n$ such that $\Pr[\hat{Z}^n\neq Z^n]\to0$ as $n\to\infty$, and it allows for small information leakage; see \Definitionref{asymp_secure-protocol} for a formal definition. Recall that for a protocol $\Pi_n$, we define the rate quadruple $(R_{12},R_{23},R_{31},\rho)$ as $R_{ij} := \frac{1}{n}\mathbb{E}[L_{ij}]$, $i,j=1,2,3$, $i\neq j$, and $\rho := \frac{1}{n}H(M_{12},M_{23},M_{31}|X^n,Y^n)$.
\begin{figure}[tb]
\centering
\begin{tikzpicture}[>=stealth', font=\sffamily\Large\bfseries, thick]
\draw [fill=lightgray] (-2.5,0) circle [radius=0.4]; \node at (-2.5,0) {1};
\draw [fill=lightgray] (2.5,0) circle [radius=0.4]; \node at (2.5,0) {2};
\draw [fill=lightgray] (0,2) circle [radius=0.4]; \node at (0,2) {3};

\draw [<->] (-2.1,0) -- (2.1,0); \node [scale=0.8] at (0,-0.3) {$M_{12}$};
\draw [<->] (-2.15,0.25) -- (-0.3,1.7); \node [scale=0.8] at (-1.6,1.2) {$M_{31}$};
\draw [<->] (2.15,0.25) -- (0.3,1.7); \node [scale=0.8] at (1.6,1.2) {$M_{23}$};

\draw [->] (-3.8,0) -- (-2.9,0); \node [scale=0.8] at (-3.4,0.3) {${X^n}$};
\draw [->] (3.8,0) -- (2.9,0); \node [scale=0.8] at (3.5,0.3) {${Y^n}$};
\draw [->] (0,2.4) -- (0,3.1); \node [scale=0.8] at (0.4,2.75) {${\hat{Z}^n}$};

\node [right, scale=0.7] at (1.7,2.5) {$Z=f(X,Y)$};

\end{tikzpicture}
\caption{A setup for 3-user secure computation; privacy is required against single users (i.e., no collusion). Here $(X^n,Y^n)\sim p_{XY}$, i.i.d., and $Z=f(X,Y)$ where $f$ is the function being computed.}
\label{fig:asymp_setup}
\end{figure}
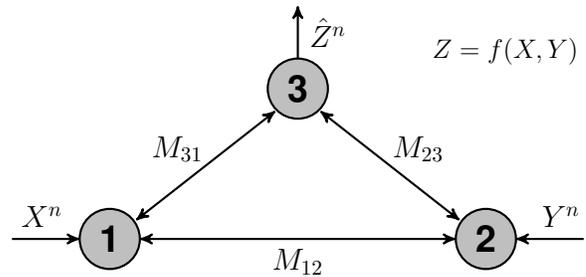
\begin{defn}\label{defn:asymp_secure-protocol}
For a secure computation problem $(f,p_{XY})$, the rate $(R_{12},R_{23},R_{31},\rho)$ is {\em achievable} with asymptotic security, if there exists a sequence of protocols $\Pi_n$ with rate $(R_{12},R_{23},R_{31},\rho)$, such that for every $\epsilon >0$, there is a large enough $n$, such that
\begin{align}
\Pr[\hat{Z}^n \neq Z^n] &\leq \epsilon, \label{eq:correct_cond}\\
I(M_{12},M_{31} ; Y^n | X^n ) &\leq \epsilon, \label{eq:privacy_alice} \\
I(M_{12},M_{23} ; X^n | Y^n ) &\leq \epsilon, \label{eq:privacy_bob} \\
I(M_{23},M_{31} ; X^n,Y^n | Z^n ) &\leq \epsilon. \label{eq:privacy_charlie}
\end{align}
The rate-region $\R^{\AS}$ is closure of the set of all achievable rate quadruples.
\end{defn}
Here, \eqref{eq:privacy_alice}-\eqref{eq:privacy_bob} ensure that Alice and Bob learn negligible additional information about each other's inputs, and \eqref{eq:privacy_charlie} ensures that Charlie learns negligible additional information about $(X^n,Y^n)$ than revealed by $Z^n$.

We prove bounds on the entropies $H(M_{ij})$, which, as argued in \Sectionref{prelims}, is a lower bound on the expected length of the transcript $M_{ij}$. Let $\Pi_n$ be a sequence of protocols which imply the achievability of rate $(R_{12},R_{23},R_{31},\rho)$ as per \Definitionref{asymp_secure-protocol}. Then, let $\epsilon >0$ and $n$ large enough be such that \eqref{eq:correct_cond}-\eqref{eq:privacy_charlie} are satisfied.

We need three key lemmas: \Lemmaref{cutset}, using cutset arguments, gives an upper bound on the amount of information present about inputs on different cuts; \Lemmaref{data-processing} gives a secure data-processing inequality for residual information; and \Lemmaref{infoineq}, which gives an information inequality for interactive protocols.

Our main results provide outer bounds on the rate-region for secure computation of a given function $f$ and an input distribution $p_{XY}$. The results are stated in \Theoremref{main_lbs_dependent} and \Theoremref{main_lbs_independent}. In \Subsectionref{asymp_examples}, we present some example functions for which our outer bounds are tight.

\subsection{Cutset Bounds}\label{subsec:cutset_bounds}
Our first lemma will imply that the cut separating Alice must reveal
information about $X^n$ at a rate of at least $H_{G_X}(X|Y)$ (see
Figure~\ref{fig:setup_two-user}). The rough intuition is that since Alice
is not allowed to learn any (significant amount of) new information about
Bob's input $Y^n$, this is essentially the function computation problem
with one-sided communication of Orlitsky and Roche~\cite{OrlitskyRoche} for
which the converse result there implies that Alice must send information
about $X^n$ at a rate of at least $H_{G_X}(X|Y)$.
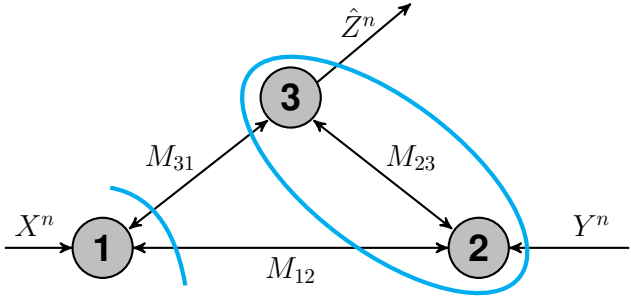
\begin{figure}
\centering
\begin{tikzpicture}[>=stealth', font=\sffamily\Large\bfseries, thick]
\draw [fill=lightgray] (-2.5,0) circle [radius=0.4]; \node at (-2.5,0) {1};
\draw [fill=lightgray] (2.5,0) circle [radius=0.4]; \node at (2.5,0) {2};
\draw [fill=lightgray] (0,2) circle [radius=0.4]; \node at (0,2) {3};

\draw [<->] (-2.1,0) -- (2.1,0); \node [scale=0.8] at (0,-0.3) {$M_{12}$};
\draw [<->] (-2.15,0.25) -- (-0.3,1.7); \node [scale=0.8] at (-1.6,1.2) {$M_{31}$};
\draw [<->] (2.15,0.25) -- (0.3,1.7); \node [scale=0.8] at (1.6,1.2) {$M_{23}$};

\draw [->] (-3.8,0) -- (-2.9,0); \node [scale=0.8] at (-3.4,0.3) {${X^n}$};
\draw [->] (4.5,0) -- (2.9,0); \node [scale=0.8] at (4,0.3) {${Y^n}$};
\draw [->] (0.35,2.2) -- (1.6,3.25); \node [scale=0.8] at (0.9,3) {${\hat{Z}^n}$};

\draw[ultra thick, cyan] (-2.4,0.8) to [out=-10,in=100] (-1.4,-0.5);
\draw [ultra thick, cyan] (1.25,0.97) ellipse [x radius=1cm, y radius=2.25cm, rotate=53];

\end{tikzpicture}
\caption{A cut separating Alice from Bob \& Charlie. Protocol $\Pi_n$ induces a 2-user secure computation protocol between Alice and {\em combined Bob-Charlie} with privacy requirement only against Alice.}
\label{fig:setup_two-user}
\end{figure}

\begin{lem}\label{lem:cutset}
The protocol $\Pi_n$ satisfies the following:
\begin{align}
H(X^n|M_{12},M_{31}) &\leq n(H(X)-H_{G_X}(X|Y) + \epsilon_1), \label{eq:cutset_alice} \\
H(Y^n|M_{12},M_{23}) &\leq n(H(Y)-H_{G_Y}(Y|X) + \epsilon_2), \label{eq:cutset_bob} \\
H(Z^n|M_{23},M_{31}) &\leq n\epsilon_3, \label{eq:cutset_charlie}
\end{align}
where $\epsilon_1,\epsilon_2,\epsilon_3 \to 0$ as $\epsilon \to 0$.
\end{lem}
\begin{remark}
{\em Note that this is {\em unlike} the case for perfectly secure
computation. By \Lemmaref{general_cutset}, which is analogous to the above
lemma for perfectly secure computation (if we restrict $p_{Z|XY}$ there to
be a deterministic function $f$ such that the pair $(p_{XY},f)$ is in
normal form), all three conditional entropies on the left-hand-side above
are equal to zero. However, for asymptotically secure computation (even if
we restrict the pair $(p_{XY},f)$ to be in normal form), these conditional
entropies may be far from zero -- in fact,
\eqref{eq:cutset_alice}-\eqref{eq:cutset_bob} above can hold with equality
asymptotically (see \Subsectionref{asymp_examples}).}
\end{remark}
\begin{proof}
\begin{align}
&H(X^n|M_{12},M_{31}) \nonumber \\
&= I(X^n;Y^n|M_{12},M_{31}) + H(X^n|M_{12},M_{31},Y^n) \nonumber \\
&\leq I(M_{12},M_{31},X^n;Y^n) + H(X^n|M_{12},M_{31},Y^n) \nonumber \\
&= I(X^n;Y^n) + \underbrace{I(M_{12},M_{31};Y^n|X^n)}_{\leq \ \epsilon \text{ by } \eqref{eq:privacy_alice}} + H(X^n|M_{12},M_{31},Y^n) \nonumber \\
&\leq nI(X;Y) + H(X^n|Y^n) - I(X^n;M_{12}M_{31}|Y^n) + \epsilon \nonumber \\
&\leq nH(X) - I(X^n;M_{12}M_{31}|Y^n) + \epsilon. \label{eq:interim_cutset}
\end{align}
We apply cutset arguments, and use correctness \eqref{eq:correct_cond} and privacy against Alice \eqref{eq:privacy_alice} to lower-bound the second term of \eqref{eq:interim_cutset}. Consider the cut separating Alice from the other two users; $\Pi_n$ induces a two-user secure computation protocol between Alice and {\em combined Bob-Charlie} (see Figure~\ref{fig:setup_two-user}), with privacy requirement only against Alice. For $0\leq D \leq 1$, we define
\begin{align}
R_f^{\WZ}(D) := \displaystyle \min_{\substack{p_{U|XY}: \\ U-X-Y \\ \exists g: \mathbb{E}[d_H(f(X,Y),g(U,Y))] \leq D}} I(U;X|Y), \label{eq:wyner-ziv}
\end{align}
where $d_H:\mathcal{U}\times\mathcal{Y}\to\{0,1\}$ is the Hamming distortion function.
$R_f^{\WZ}(D)$ is the optimal rate of Csisz\'ar-K\"orner's~\cite{CsiszarKo78} extension (also see \cite{yamamoto82}) of Wyner-Ziv problem \cite{WynerZi76} specialized to this function computation (without any privacy) as used by Orlitsky and Roche~\cite{OrlitskyRoche}.
\begin{lem}\label{lem:two-user_lower-bound}
$I(X^n;M_{12},M_{31}|Y^n) \geq n\left(R_f^{\WZ}(0)-\delta_{\epsilon}\right)$, where $\delta_{\epsilon}\to0$ as $\epsilon\to0$.
\end{lem}
The proof of the above lemma in \Appendixref{asymp_proofs} is along the lines of the converse of the Wyner-Ziv theorem in~\cite[Section 11.3]{ElgamalKim11} except for the following complication: in the Wyner-Ziv problem, communication is one-sided, but we allow messages in both directions over multiple rounds. However, as we show in \Appendixref{asymp_proofs}, privacy against Alice, $I(M_{12},M_{31};Y^n|X^n)\leq \epsilon$, which implies that very little new information about $Y^n$ flows back to Alice, allows us to prove the lemma.

We can relate $R_f^{\WZ}(D)$ with {\em conditional graph entropy} $H_{G_X}(X|Y)$ (defined 
in \Definitionref{conditional-graph-entropy}) using the following result from \cite{OrlitskyRoche}.

\begin{lem}[\cite{OrlitskyRoche}, Theorem 2]\label{lem:orlitsky-roche_graph-entropy}
For every $p_{XY}$ and $f$ \[ R_f^{\WZ}(0)=H_{G_X}(X|Y).\]
\end{lem}
\noindent From \Lemmaref{two-user_lower-bound} and \Lemmaref{orlitsky-roche_graph-entropy} we get the following:
\begin{align}
I(X^n;M_{12},M_{31}|Y^n) \geq n(H_{G_X}(X|Y)-\delta_{\epsilon}), \label{eq:entropyXgivenY_bound}
\end{align}
where $\delta_{\epsilon}\to0$ as $\epsilon\to0$.

From \eqref{eq:interim_cutset} and \eqref{eq:entropyXgivenY_bound} we get $H(X^n|M_{12},M_{31}) \leq n(H(X)-H_{G_X}(X|Y) + \epsilon_1)$, where $\epsilon_1 = \epsilon + \delta_{\epsilon}$, which proves \eqref{eq:cutset_alice}.
Similarly, by considering the cut separating Bob from Alice and Charlie, we can show the following (which proves \eqref{eq:cutset_bob}):
\begin{align}
I(Y^n;M_{12},M_{23}|X^n) \geq n(H_{G_Y}(Y|X)-\delta_{\epsilon}). \label{eq:entropyYgivenX_bound}
\end{align}
Applying Fano's inequality gives \eqref{eq:cutset_charlie} as follows:
\begin{align*}
H(Z^n|M_{23},M_{31}) &\stackrel{\text{(b)}}{=} H(Z^n|M_{23},M_{31},\hat{Z}^n) \\
&\leq H(Z^n|\hat{Z}^n) \leq 1 + \epsilon(n\log |\mathcal{Z}| -1)\leq n\epsilon_3,
\end{align*}
where (b) follows from the Markov chain $\hat{Z}^n-(M_{23},M_{31})-(X^n,Y^n,M_{12})$, and $\epsilon_3\to 0$ as $\epsilon\to 0$.
\end{proof}

\subsection{Asymptotically secure Data Processing Inequality}\label{subsec:asymp_data-processing}
To prove \Theoremref{main_lbs_dependent} and \Theoremref{main_lbs_independent}, we need to prove an asymptotic version of the secure data processing inequality of \Lemmaref{monotone}.
\begin{lem}[Asymptotically secure data processing inequality]\label{lem:data-processing}
Privacy conditions \eqref{eq:privacy_alice}-\eqref{eq:privacy_charlie} imply the following:
\begin{align*}
\frac{1}{n}RI(M_{12},M_{31},X^n;M_{23},M_{31},Z^n) &\geq RI(X;Z) -\epsilon_{31},\\
\frac{1}{n}RI(M_{12},M_{23},Y^n;M_{23},M_{31},Z^n) &\geq RI(Y;Z) - \epsilon_{23},\\
\frac{1}{n}RI(M_{12},M_{31},X^n;M_{12},M_{23},Y^n) &\geq RI(X;Y) - \epsilon_{12},
\end{align*}
where $\epsilon_{12},\epsilon_{23},\epsilon_{31}\to0$ as $\epsilon\to 0$.
\end{lem}
We prove this Lemma in \Appendixref{asymp_proofs}. Our proof only uses ``weak'' privacy constraints, i.e., the privacy conditions are upper bounded by $n\epsilon$, as opposed to the ``strong'' ones in \eqref{eq:privacy_alice}-\eqref{eq:privacy_charlie} which are upper bounded by $\epsilon$.

\subsection{Main Lower Bounds}\label{subset:main_lower-bounds}

\begin{thm}\label{thm:main_lbs_dependent}
For a secure computation problem $(f,p_{XY})$, if $(R_{12},R_{23},R_{31},\rho)\in\R^{\AS}$, then
\begin{align*}
R_{12} &\geq H_{G_X}(X|Y) + H_{G_Y}(Y|X) - H(Z) \\
&\hspace{3cm} + \max\{RI(X;Z),RI(Y;Z)\}, \\
R_{23} &\geq H_{G_Y}(Y|X) + RI(X;Z), \\
R_{31} &\geq H_{G_X}(X|Y) + RI(Y;Z),\\
\rho &\geq H_{G_X}(X|Y) + H_{G_Y}(Y|X) - H(Z) \\
&\qquad \qquad + RI(X;Z) + RI(Y;Z), \\
\rho &\geq H_{G_X}(X|Y) + H_{G_Y}(Y|X) - I(X;Y)   \\
&\ + RI(X;Y) + \max\{RI(X;Z),RI(Y;Z)\} - H(Z).
\end{align*}
\end{thm}
\begin{proof}
\begin{align}
&H(M_{12}) \geq H(M_{12}|M_{31}) \nonumber \\
&= H(M_{12}|M_{31},M_{23}) + I(M_{12};M_{23}|M_{31}) \nonumber \\
&= H(M_{12},X^n|M_{31},M_{23}) - H(X^n|M_{12},M_{31},M_{23}) \nonumber \\
&\qquad \quad + I(M_{12},X^n;M_{23}|M_{31}) - I(X^n;M_{23}|M_{12},M_{31}) \nonumber \\
&= H(M_{12},X^n|M_{31},M_{23}) - H(X^n|M_{12},M_{31}) \nonumber \\
&\qquad \quad + I(M_{12},X^n;M_{23}|M_{31}) \label{eq:M12_useful-later} \\
&\geq H(M_{12},X^n|M_{31},M_{23}) + I(M_{12},X^n;M_{23}|M_{31}) \nonumber \\
&\qquad \quad - n(H(X)-H_{G_X}(X|Y) + \epsilon_1), \label{eq:M12_interim}
\end{align}
where \eqref{eq:M12_interim} follows from \eqref{eq:cutset_alice}. The first two terms of \eqref{eq:M12_interim} can be bounded easily as follows:
{\allowdisplaybreaks
\begin{align}
&H(M_{12},X^n|M_{31},M_{23}) \nonumber \\
&= H(M_{12},X^n,Y^n|M_{31},M_{23}) - H(Y^n|M_{12},M_{23},M_{31},X^n) \nonumber \\
&\geq H(X^n,Y^n|M_{31},M_{23},Z^n) - H(Y^n|M_{12},M_{23},X^n) \nonumber \\
&\geq n(H(X,Y|Z)-\epsilon/n) - n(H(Y|X)-H_{G_Y}(Y|X)+\delta_{\epsilon}). \nonumber \\
&= n(H(X)-H(Z) + H_{G_Y}(Y|X) - \epsilon/n - \delta_{\epsilon}) \label{eq:M12_first_interim}
\end{align}
In the last inequality we use privacy against Charlie \eqref{eq:privacy_charlie} and $H(Y^n|M_{12},M_{23},X^n) \leq n(H(Y|X)-H_{G_Y}(Y|X)+\delta_{\epsilon})$ from \eqref{eq:entropyYgivenX_bound}.
}
{\allowdisplaybreaks
\begin{align}
& I(M_{12},X^n;M_{23}|M_{31}) \nonumber \\
&= I(M_{12},X^n;M_{23},Z^n|M_{31}) - \underbrace{I(M_{12},X^n;Z^n|M_{23},M_{31})}_{\leq \ H(Z^n|M_{23},M_{31}) \ \leq \ n\epsilon_3 \text{ by } \eqref{eq:cutset_charlie}} \nonumber \\
&\geq RI(M_{12},M_{31},X^n;M_{23},M_{31},Z^n) - n\epsilon_3 \label{eq:M12_second_intermediate} \\
&\geq n(RI(X;Z) - \epsilon_{31} - \epsilon_3) \quad \text{(by \Lemmaref{data-processing})}, \label{eq:M12_second_interim}
\end{align}
where \eqref{eq:M12_second_intermediate} follows from the definition of residual information in \eqref{eq:residual_info} by taking $U=(M_{12},M_{31},X^n)$, $V=(M_{23},M_{31},Z^n)$, and $Q=M_{31}$.
}
Substituting from \eqref{eq:M12_first_interim} and \eqref{eq:M12_second_interim} into \eqref{eq:M12_interim} and simplifying further, we get:
\begin{align*}
H(M_{12}) &\geq n(H_{G_X}(X|Y) + H_{G_Y}(Y|X) - H(Z) \\
&\hspace{4cm} + RI(X;Z) - \epsilon'_{12}),
\end{align*}
where $\epsilon'_{12}=\epsilon/n + \epsilon_1 + \epsilon_3 + \epsilon_{31} + \delta_{\epsilon}$ and  $\epsilon'_{12}\to0$ as $\epsilon\to0$.
By symmetry and letting $\epsilon \downarrow 0$, we have
\begin{align*}
R_{12} &\geq H_{G_X}(X|Y) + H_{G_Y}(Y|X) - H(Z) \\
&\hspace{3cm} + \max\{RI(X;Z),RI(Y;Z)\}.
\end{align*}
The remaining bounds on $R_{23}, R_{31}$ and $\rho$ are proved below.
{\allowdisplaybreaks
\begin{align*}
&H(M_{23}) \geq H(M_{23}|M_{31}) \\
&= H(M_{23}|M_{12},M_{31},X^n) + I(M_{23};M_{12},X^n|M_{31}) \\
&= H(M_{23},Y^n|M_{12},M_{31},X^n) - H(Y^n|M_{12},M_{23},M_{31},X^n) \\
&\quad + I(M_{23},Z^n;M_{12},X^n|M_{31}) - I(Z^n;M_{12},X^n|M_{23},M_{31}) \\
&\geq \underbrace{H(Y^n|M_{12},M_{31},X^n)}_{\geq\ H(Y^n|X^n) -\epsilon \text{ by } \eqref{eq:privacy_alice}} - \underbrace{H(Y^n|M_{12},M_{23},X^n)}_{\leq\ n(H(Y|X)-H_{G_Y}(Y|X)+\delta_{\epsilon}) \text{ by } \eqref{eq:entropyYgivenX_bound}} \\
&\quad + \underbrace{RI(M_{23},M_{31},Z^n;M_{12},M_{31},X^n)}_{\geq\ n(RI(Z;X)-\epsilon_{31})\text{ by \Lemmaref{data-processing}}} - \underbrace{H(Z^n|M_{23},M_{31})}_{\leq\ n\epsilon_3 \text{ by } \eqref{eq:cutset_charlie}} \\
&\geq n(H_{G_Y}(Y|X) + RI(X;Z) - \delta_{23}) \\
&\qquad (\text{where } \delta_{23}=\epsilon/n + \epsilon_3 + \epsilon_{31} + \delta_{\epsilon}, \text{ and } \delta_{23}\to0 \text{ as }\epsilon\to0)
\end{align*}
}
By letting $\epsilon \downarrow 0$, we get
\begin{align*}
R_{23}\geq H_{G_Y}(Y|X) + RI(X;Z).
\end{align*}
Similarly, we can prove the bound on $H(M_{31})$ by first expanding as $H(M_{31}) \geq H(M_{31}|M_{23})$ and then proceed as in $H(M_{23})$. For the rate of private randomness $\rho$ required, we bound $H(M_{12},M_{23},M_{31}|X^n,Y^n)$ as follows:
{\allowdisplaybreaks
\begin{align}
n\rho_n &= H(M_{12},M_{23},M_{31}|X^n,Y^n) \nonumber \\
&\geq H(M_{12},M_{31}|X^n,Y^n) \nonumber \\
&= H(M_{12},M_{31}|X^n) - \underbrace{I(M_{12},M_{31};Y^n|X^n)}_{\leq \ \epsilon \text{ by } \eqref{eq:privacy_alice}} \nonumber \\
&\geq H(M_{12},M_{31},X^n) - H(X^n) - \epsilon. \label{eq:rand_interim}
\end{align}
We bound the first term of \eqref{eq:rand_interim} as follows:
}
{\allowdisplaybreaks
\begin{align}
&H(M_{12},M_{31},X^n) \nonumber \\
&= H(M_{31}) + H(M_{12}|M_{31}) + H(X^n|M_{12},M_{31}). \label{eq:rand_first-way_interim}
\end{align}
We can bound the second and third term of \eqref{eq:rand_first-way_interim} together as follows:
\begin{align}
&H(M_{12}|M_{31}) + H(X^n|M_{12},M_{31}) \nonumber \\
&\qquad \geq H(M_{12},X^n|M_{23},M_{31}) + I(M_{12},X^n;M_{23}|M_{31}) \nonumber \\
&\qquad \geq n(H(X) - H(Z) + H_{G_Y}(Y|X) -\epsilon/n - \delta_{\epsilon}) \nonumber \\
&\qquad \quad + n(RI(X;Z)-\epsilon_{31}-\epsilon_3), \label{eq:rand_first-way_interim2}
\end{align}
where the first inequality follows from \eqref{eq:M12_useful-later}; \eqref{eq:rand_first-way_interim2} follows from \eqref{eq:M12_first_interim} and \eqref{eq:M12_second_interim}. We lower-bound the first term of \eqref{eq:rand_first-way_interim} as follows:
}
{\allowdisplaybreaks
\begin{align}
&H(M_{31}) \geq H(M_{31}|M_{23}) \nonumber \\
&= H(M_{31}|M_{12},M_{23},Y^n) + I(M_{31};M_{12},Y^n|M_{23}) \nonumber \\
&= H(M_{31},X^n|M_{12},M_{23},Y^n) - H(X^n|M_{12},M_{23},M_{31},Y^n) \nonumber \\
&\qquad + I(M_{31},Z^n;M_{12},Y^n|M_{23}) -I(Z^n;M_{12},Y^n|M_{23},M_{31}) \nonumber \\
&\stackrel{\text{(a)}}{\geq} \underbrace{H(X^n|M_{12},M_{23},Y^n)}_{\geq\ H(X^n|Y^n)-\epsilon \text{ by } \eqref{eq:privacy_bob}} - \underbrace{H(X^n|M_{12},M_{31},Y^n)}_{\leq\ n(H(X|Y) - H_{G_X}(X|Y)+\delta_{\epsilon}) \text{ by }\eqref{eq:entropyXgivenY_bound}} \nonumber \\
&\qquad + RI(M_{23},M_{31},Z^n;M_{12},M_{23},Y^n) -H(Z^n|M_{23},M_{31}) \nonumber \\
&\geq n(H_{G_X}(X|Y)-\epsilon/n-\delta_{\epsilon}) + n(RI(Y;Z)-\epsilon_{23}) - n\epsilon_3, \label{eq:rand_first-way_interim1}
\end{align}
where, in (a) we use the definition of residual information \eqref{eq:residual_info} and simple Shannon information inequalities. In \eqref{eq:rand_first-way_interim1} we use \Lemmaref{data-processing} and \eqref{eq:cutset_charlie}.
}
From \eqref{eq:rand_interim}-\eqref{eq:rand_first-way_interim1} and letting $\epsilon\downarrow0$, we get the following:
\begin{align}
\rho &\geq H_{G_X}(X|Y) + H_{G_Y}(Y|X) + RI(X;Z) \nonumber \\
&\qquad \qquad + RI(Y;Z) - H(Z). \label{eq:rand_first-bound} 
\end{align}
There is another way to bound the first term of \eqref{eq:rand_interim} as follows:
{\allowdisplaybreaks
\begin{align}
&H(M_{12},M_{31},X^n) \nonumber \\
&= H(M_{12}) + H(M_{31}|M_{12}) + H(X^n|M_{12},M_{31}) \nonumber \\
&\geq H(M_{12}|M_{31}) + H(X^n|M_{12},M_{31}) + H(M_{31}|M_{12}). \label{eq:rand_second-way_interim}
\end{align}
The first two terms of \eqref{eq:rand_second-way_interim} can be bounded as in \eqref{eq:rand_first-way_interim2}.
We bound the last term of \eqref{eq:rand_second-way_interim} as follows:
}
{\allowdisplaybreaks
\begin{align}
&H(M_{31}|M_{12}) \nonumber \\
&= H(M_{31}|M_{12},M_{23}) + I(M_{31};M_{23}|M_{12}) \nonumber \\
&= H(M_{31}|M_{12},M_{23},Y^n) + I(M_{31};Y^n|M_{12},M_{23}) \nonumber \\
&\hspace{3cm}+ I(M_{31};M_{23}|M_{12}) \nonumber \\
&= H(M_{31},X^n|M_{12},M_{23},Y^n) - H(X^n|M_{12},M_{23},M_{31},Y^n) \nonumber \\
&\hspace{3cm} + I(M_{31};M_{23},Y^n|M_{12}) \nonumber \\
&\geq \underbrace{H(X^n|M_{12},M_{23},Y^n)}_{\geq \ H(X^n|Y^n) - \epsilon \text{ by }\eqref{eq:privacy_bob}} + \underbrace{I(M_{31},X^n;M_{23},Y^n|M_{12})}_{\geq\ n(RI(X;Y)-\epsilon_{12}) \text{ by } \eqref{eq:residual_info} \text{ and \Lemmaref{data-processing}}} \nonumber \\
& - I(X^n;M_{23},Y^n|M_{12},M_{31}) - H(X^n|M_{12},M_{23},M_{31},Y^n) \nonumber \\
&\geq H(X^n|Y^n) - \epsilon + n(RI(X;Y)-\epsilon_{12}) - H(X^n|M_{12},M_{31}) \nonumber \\
&\geq n(H_{G_X}(X|Y) + RI(X;Y) - I(X;Y) -\epsilon/n - \epsilon_1 - \epsilon_{12}). \label{eq:rand_second-way_interim2}
\end{align}
Last inequality \eqref{eq:rand_second-way_interim2} follows from \eqref{eq:cutset_alice}.
}
From \eqref{eq:rand_interim}, \eqref{eq:rand_first-way_interim2}, \eqref{eq:rand_second-way_interim}, \eqref{eq:rand_second-way_interim2}, and by letting $\epsilon\downarrow0$, we get the following:
\begin{align*}
\rho &\geq H_{G_X}(X|Y) + H_{G_Y}(Y|X) + RI(X;Y) \\
&\qquad \qquad + RI(X;Z)  - I(X;Y) - H(Z).
\end{align*}
By symmetry, we have
\begin{align}
\rho &\geq H_{G_X}(X|Y) + H_{G_Y}(Y|X) + RI(X;Y) \nonumber \\
&\ + \max\{RI(X;Z),RI(Y;Z)\}  - I(X;Y) - H(Z). \label{eq:rand_second-bound}
\end{align}
\eqref{eq:rand_first-bound} and \eqref{eq:rand_second-bound} together prove the bounds on $\rho$ in \Theoremref{main_lbs_dependent}.
\end{proof}
{
If the input distribution $p_{XY}$ is a product distribution, i.e., $p_{XY}=p_Xp_Y$, then we 
can improve \Theoremref{main_lbs_dependent}.
In the case of independent inputs we can assume, without loss of generality, that $p_X$ and 
$p_Y$ have full support; and as observed earlier, the input distribution having full support 
allows us to assume, without loss of generality, that the function $f:\X\times\Y\to\Z$ is in 
normal form (see \Sectionref{prelims} for details), which implies that the characteristic 
graphs $G_X, G_Y$ are complete graphs, and therefore, the conditional graph entropies are 
equal to conditional entropies, i.e., $H_{G_X}(X|Y)=H(X|Y), H_{G_Y}(Y|X)=H(Y|X)$. For independent 
inputs we have $I(X;Y)=0$, and the bounds in \Lemmaref{cutset} reduce to the following: 
$H(X^n|M_{12},M_{31})\leq n\epsilon_1$, $H(Y^n|M_{12},M_{23})\leq n\epsilon_2$, and 
$H(Z^n|M_{23},M_{31})\leq n\epsilon_3$.}
\begin{thm}\label{thm:main_lbs_independent}
{Consider a secure computation problem $(f,p_Xp_Y)$, where $p_X$ and $p_Y$ have full support and $f$ is in normal form. If $(R_{12},R_{23},R_{31},\rho)\in\R^{\AS}$, then}
\begin{align*}
R_{12} &\geq H(X,Y|Z) + RI(X;Z) + RI(Y;Z), \\
R_{23} &\geq H(Y|X) + RI(X;Z), \\
R_{31} &\geq H(X|Y) + RI(Y;Z), \\
\rho &\geq H(X,Y|Z) + RI(X;Z) + RI(Y;Z).
\end{align*}
\end{thm}
\begin{proof}
{We crucially use the information inequality for interactive protocols (\Lemmaref{infoineq}) to improve the bounds in \Theoremref{main_lbs_dependent}.
The improvement is only on $R_{12}$. 
The other bounds on $R_{23},R_{31},\rho$ can directly be obtained by substituting the conditional graph entropies by conditional entropies in \Theoremref{main_lbs_dependent}.}
{\allowdisplaybreaks
\begin{align}
&H(M_{12}) = H(M_{12}|M_{23},M_{31}) + I(M_{12};M_{23},M_{31}) \nonumber \\
&= H(M_{12}|M_{23},M_{31}) + I(M_{12};M_{23}) + I(M_{12};M_{31}|M_{23}) \nonumber \\
&\geq H(M_{12}|M_{23},M_{31}) + I(M_{12};M_{23}|M_{31}) \nonumber \\
&\qquad + I(M_{12};M_{31}|M_{23}) \qquad \text{ (by \Lemmaref{infoineq})} \label{eq:interim_M12_bound}
\end{align}
We bound the first term of \eqref{eq:interim_M12_bound} from below as follows:
}
{\allowdisplaybreaks
\begin{align}
&H(M_{12}|M_{23},M_{31}) \nonumber \\
&= H(M_{12},X^n|M_{23},M_{31}) - \underbrace{H(X^n|M_{12},M_{23},M_{31})}_{\leq \ H(X^n|M_{12},M_{31})\ \leq\ n\epsilon_1} \nonumber \\
&\geq H(M_{12},X^n,Y^n|M_{23},M_{31}) - \underbrace{H(Y^n|M_{12},M_{23},M_{31},X^n)}_{\leq \ H(Y^n|M_{12},M_{23})\ \leq \ n\epsilon_2} \nonumber \\
&\hspace{3cm} - n\epsilon_1 \nonumber \\
&\geq H(X^n,Y^n|M_{23},M_{31},Z^n) - n(\epsilon_1 + \epsilon_2) \nonumber \\
&\geq n(H(X^n,Y^n|Z^n) - \epsilon/n - \epsilon_1 - \epsilon_2) \quad \text{ by } \eqref{eq:privacy_charlie}\label{eq:first_M12_bound}
\end{align}
}
We bound the second term of \eqref{eq:interim_M12_bound} from below as follows:
{\allowdisplaybreaks
\begin{align}
&I(M_{12};M_{23}|M_{31}) \nonumber \\
&= I(M_{12},X^n;M_{23}|M_{31}) - \underbrace{I(X^n;M_{23}|M_{12},M_{31})}_{\leq\ H(X^n|M_{12},M_{31})\  \leq \ n\epsilon_1} \nonumber \\
&\geq I(M_{12},X^n;M_{23},Z^n|M_{31}) - \underbrace{I(M_{12},X^n;Z^n|M_{23},M_{31})}_{\leq\ H(Z^n|M_{23},M_{31})\  \leq \ n\epsilon_3} \nonumber \\
&\hspace{3cm} - n\epsilon_1 \nonumber \\
&\stackrel{\text{(a)}}{\geq} RI(M_{12},M_{31},X^n ; M_{23},M_{31},Z^n) - n(\epsilon_1+\epsilon_3) \nonumber \\
&\geq n(RI(X;Z) - \epsilon_1 - \epsilon_3 - \epsilon_{31}), \quad  \text{(by \Lemmaref{data-processing})} \label{eq:second_M12_bound}
\end{align}
where (a) follows from the definition of residual information \eqref{eq:residual_info}, by taking $U=(M_{12},M_{31},X^n)$, $V=(M_{23},M_{31},Z^n)$, and $Q=M_{31}$.
}
Similarly, we can bound the third term of \eqref{eq:interim_M12_bound} as follows:
\begin{align}
I(M_{12};M_{31}|M_{23}) &\geq n(RI(Y;Z) - \epsilon_2 - \epsilon_3 - \epsilon_{23}). \label{eq:third_M12_bound}
\end{align}
Substituting the values from \eqref{eq:first_M12_bound}-\eqref{eq:third_M12_bound} into \eqref{eq:interim_M12_bound}, we get
\begin{align}
H(M_{12}) &\geq n(H(X,Y|Z) + RI(X;Z) + RI(Y;Z) - \gamma_{12}), \nonumber
\end{align}
where $\gamma_{12} = \epsilon/n + 2(\epsilon_1 + \epsilon_2 + \epsilon_3) + \epsilon_{31} + \epsilon_{23}$, and $\gamma_{12} \to 0$ as $\epsilon\to0$. By letting $\epsilon\downarrow0$, we have
\[R_{12} \geq H(X,Y|Z) + RI(X;Z) + RI(Y;Z).\]
\end{proof}
{\begin{remark}\label{remark:as_remark_rand}
{\em The observation in \Remarkref{ps_remark_rand} on the fact that working with normal form may not be without loss of generality as far as randomness requirement is concerned even if $p_{XY}$ has full support holds here as well.
}
\end{remark}}

\subsection{Application to Specific Functions}\label{subsec:asymp_examples}
In this section we present some examples and show that secure protocols for some of these achieve the optimal rate-region. \\

{\bf 1. Secure computation of {\sc addition} in a finite field:}
Let $(\mathbb{F},+,\times)$ be a finite field and $(X,Y)\sim p_{XY}$, where $p_{XY}$ is a joint distribution over $\mathbb{F}^2$. The function {\sc addition} is defined as follows: $Z=X+Y$, where $+$ is performed in $\mathbb{F}$. Our protocol uses the following fact about data compression of a discrete memoryless source.

{\em Fact 1 \cite{Elias55}. }Let $U^n$ be a sequence of $n$ i.i.d. random variables, each distributed over $\mathbb{F}$. For fix $\epsilon>0$, let $R=H(U)/\log|\mathbb{F}|+\epsilon$, then there is a sequence of linear encoder and decoder pairs $(A_n, D_n)$, where $A_n\in\mathbb{F}^{nR\times n}$ and $D_n:\mathbb{F}^{nR}\to\mathbb{F}^n$, such that $\Pr[D_n(A_nU^n)\neq U^n]\to 0$ as $n\to\infty$.

The protocol in \Figureref{addition} is a secure version of K{\"o}rner-Marton scheme \cite{KornerMa79} in finite fields. All the arithmetic is in $\mathbb{F}$. The protocol requires $n\rho = |M_{12}| = |M_{23}| = |M_{31}| = n(H(Z)+\epsilon)$ (in bits). Below we show that if $p_{XY}$ is a product distribution, i.e., $p_{XY}=p_Xp_Y$, then this protocol achieves the optimal rate-region.

\begin{figure}[htb]
\hrule height 1pt
\vspace{.05cm}
{\bf Algorithm 6:} {Secure Computation of {\sc addition} in a finite field} 
\hrule
\begin{algorithmic}[1]
\REQUIRE Alice \& Bob have input vectors $X^n,Y^n\in\mathbb{F}^n$.
\ENSURE Charlie securely computes $\hat{Z}^n$ with $\Pr[\hat{Z}^n\neq Z^n]\to 0$, where $Z^n=(Z_1,Z_2,\hdots,Z_n), Z_i=X_i+Y_i$.

\medskip

\STATE For fix $\epsilon$, let $R=H(Z)/\log|\mathbb{F}|+\epsilon$ and $A_n$ be a $nR\times n$ matrix in $\mathbb{F}$ whose existence is ensured by the fact 1. Alice and Bob share $K\sim \text{Unif}(\mathbb{F}^{nR})$ over 1-2 link.

\STATE Alice sends $M_{\vec{13}} := A_nX^n + K$ (component-wise addition) to Charlie.

\STATE Bob sends $M_{\vec{23}}:= A_nY^n - K$ to Charlie.

\STATE Charlie computes $M_{13}+M_{23}$, which is equal to $A_nZ^n$, and recovers $Z^n$ with high probability.
\end{algorithmic}
\hrule
\caption{An optimal protocol for block-wise secure computation of {\sc addition} in any finite field $\mathbb{F}$. The protocol requires roughly $nH(Z)$ bits to be exchanged on average over each link.}
\label{fig:addition}
\end{figure}

\begin{thm}\label{thm:as_addition}
For any secure protocol for computing {\sc addition} in a finite field $\mathbb{F}$ for independent $X\sim p_X$, $Y\sim p_Y$ with full support, we have the following optimal bound on the rate-region:
\[R_{12},R_{23},R_{31},\rho \geq H(Z).\]
\end{thm}
\begin{proof}
For this function with $p_Xp_Y$ having full support, we have $RI(X;Z)=I(X;Z)$ and $RI(Y;Z)=I(Y;Z)$. For a product distribution, i.e., $p_{XY}=p_Xp_Y$, it can be verified easily that the all four bounds on $R_{12},R_{23},R_{31},\rho$ in \Theoremref{main_lbs_independent} reduce to $H(Z)$ (in bits), thereby achieving the optimal rate-region. Note that the converse of the optimality of this protocol needs full force of \Lemmaref{infoineq}.\\
\end{proof}
{\em Separating Perfectly and Asymptotically Secure Computation:} Note that the result of \Theoremref{group} 
in \Subsectionref{examples} also holds when restricted to independent inputs taking values in finite fields. Comparing that with the result of above \Theoremref{as_addition} establishes a gap in the rate regions of perfectly secure computation and asymptotically secure computation.
 
We give a tight characterization of the rate-region of this function only for independent input distributions, and we leave it as an interesting open problem to characterize the rate-region of this function for arbitrary $p_{XY}$. For arbitrary $p_{XY}$ our bounds in \Theoremref{main_lbs_dependent} reduce to $\rho,R_{23},R_{31}\geq H(Z)$ but $R_{12} \geq \max\{H(X|Y),H(Y|X))\}$. In general, the bound on $R_{12}$ does not match what our protocol achieves.
But for the special case of the joint distribution of K{\"o}rner-Marton \cite{KornerMa79}: $p_{XY}(x,y)=\frac{p}{2}1_{x\neq y}+\frac{1-p}{2}1_{x=y}$, where $x,y\in\{0,1\}$ and $0\leq p\leq 1/2$, we have $H(X|Y)=H(Y|X)=H(Z)$. This distribution is sometimes referred to as the doubly symmetric binary source (DSBS) with parameter $p$. Thus the secure computation of modular addition in a binary field with DSBS source requires $R_{12},R_{23},R_{31},\rho\geq H(Z)$, which is also an example with dependent inputs that separates perfect secure computation from asymptotically secure computation. \\

{\bf 2. Secure computation of {\sc controlled-erasure}:}
We again study the controlled erasure function from
\Subsectionref{examples} here in the asymptotic setting. 
In this function Alice and Bob have one bit input $X$ and $Y$, respectively, where Alice's input $X$ acts as
the ``control'' which decides whether Charlie receives an erasure
($\Delta$) or Bob's input $Y$. 

Let $(X,Y)\sim p_{XY}$, where $p_{XY}$ is a joint distribution over $\{0,1\}^2$ with marginal distributions $X\sim$ Bern($p$) and $Y\sim$ Bern($q$), $p,q\in(0,1)$.
The protocol in \Figureref{as_erasure} requires $\mathbb{E}[L_{31}]< n(H_2(p)+p)+1$, $|M_{12}|=n, \rho = n$, and
\begin{align*}
|M_{23}|&=nH(Y\oplus K|G,X) =n(pH(Y|X=1)+(1-p)\cdot 1)\\ &= n(H(Y|X) + (1-p)(1-H(Y|X=0))).
\end{align*}

\begin{figure}[htb]
\hrule height 1pt
\vspace{.05cm}
{\bf Algorithm 7:} {Secure Computation of \textsc{controlled erasure}} 
\hrule
\begin{algorithmic}[1]
\REQUIRE Alice \& Bob have input vectors $X^n,Y^n\in\{0,1\}^n$ with $(X,Y)\sim p_{XY}$; let $X\sim$ Bern($p$) and $Y\sim$ Bern($q$).
\ENSURE Charlie securely computes $\hat{Z}^n$ with $\Pr[\hat{Z}^n\neq Z^n]\to0$, where $Z_i, i=1,\hdots,n.$ is the {\sc controlled-erasure} function of $X_i,Y_i$.

\medskip

\STATE Alice and Bob share $n$ random bits $K^n$ over 1-2 link.

\STATE Alice sends $M_{13}:=(C(X^n),(K_i)_{i\in\{j:X_j=1\}})$ to Charlie, where $C(X^n)$ is the Huffman compression of $X^n$.

\STATE Let
\[G_i=
\begin{cases}
K_i & \text{if $X_i=1$}, \\
\bot & \text{if $X_i=0$}.
\end{cases}
\]
Charlie decodes $C(X^n)$ to get $X^n$ and obtains $(i)_{i\in\{j:X_j=1\}}$, which, together with the second component of $M_{13}$ gives $G^n=(G_1,G_2,\hdots,G_n)$. Since Charlie has $(G^n,X^n)$, by Slepian-Wolf theorem \cite[Section 10.3]{ElgamalKim11}, Bob only needs to send at rate $H(Y\oplus K|G,X)$ for Charlie to recover $Y^n\oplus K^n$ with high probability.

\STATE Now, having access to $Y^n\oplus K^n$ and $G^n$, Charlie can recover $Z^n$.
\end{algorithmic}
\hrule
\caption{A protocol to securely compute {\sc controlled erasure} function asymptotically. For $(X,Y)\sim p_{XY}$ with
$X\sim\text{Bernoulli}(p)$ and $Y\sim\text{Bernoulli}(q)$, both i.i.d and $0<p,q\leq 1/2$.}
\label{fig:as_erasure}
\end{figure}

\begin{thm}\label{thm:as_erasure}
For any secure protocol for computing {\sc controlled-erasure} function with $(X,Y)\sim p_{XY}$, $X\sim$ Bern($p$) and $Y\sim$ Bern($q$), $p,q \in (0,1)$, we have the following bound on the rate region:
\begin{align*}
R_{12},R_{23},\rho &\geq H(Y|X), \\
R_{31} &\geq H(X) + pH(Y|X=1).
\end{align*}
If $X$ and $Y$ are independent and $q=1/2$ (irrespective of the value of $p$), we have $\rho,R_{12},R_{23}\geq n$ and $R_{31}\geq H_2(p)+p$, achieving the optimal rate region.
\end{thm}
\begin{proof}
For this function and $p_{XY}$ having full support, we have $RI(X;Z)=0$ and $RI(Y;Z)=I(Y;Z)$, bounds in \Theoremref{main_lbs_dependent} reduce to the following:
\begin{align*}
R_{12},R_{23},\rho &\geq H(Y|X), \\
R_{31} &\geq H(X) + pH(Y|X=1).
\end{align*}
It can be verified easily that for independent $X,Y$ and $q=1/2$ (irrespective of the value of $p$), the lower bounds match the protocol requirements, thereby achieving the optimal rate-region.
\end{proof}

\section{Conclusion}\label{sec:conclusion}

In this work we presented generic lower bounds on communication and randomness
for perfectly and asymptotically secure 3-user computation, and showed that
they yield tight bounds for some interesting examples. However, the general
problem of obtaining tight lower bounds for communication and randomness complexity of secure
computation remains open.

For perfectly secure computation, the standard upper bound on the total
communication exchanged between all three users is linear in the size of the
circuit computing the function \cite{BenorGoWi88, ChaumCrDa88}. This
implication to circuit lower bounds presents a ``barrier'' to obtaining
super-linear bounds for explicit functions since circuit complexity lower
bounds are notoriously difficult~\cite[Chapter~23]{AroraBarak09}. We propose a
possibly easier open problem: do there {\em exist} Boolean functions with
super-linear communication complexity for secure computation?  Note that lower
bounds on circuit complexity do not directly translate to lower bounds on
communication complexity of secure computation, as established by a
sub-exponential upper bound of $2^{\widetilde{O}(\sqrt{n})}$ for the latter
\cite{BeimelIsKuKu14}. Though it is plausible that for random Boolean
functions, the actual communication cost is $2^{\Omega(n^\epsilon)}$ for some
$\epsilon>0$, none of the current techniques appear capable of delivering such
a result. Another interesting problem we leave open is to find an explicit
example for a {\em Boolean} function in which the total communication to
Charlie must be significantly larger than the total input size. Note that
\cite{FeigeKiNa94} gave an existential result (in their restricted model)
and the explicit example in this work does not have Boolean output.

For asymptotically secure computation, the only generic feasibility results
are the ones that were developed for standard (statistically or perfectly)
secure computation, like the ones by Ben-Or, Goldwasser, and
Wigderson~\cite{BenorGoWi88} or Chaum, Cr\'epeau, and
Damg\r{a}rd~\cite{ChaumCrDa88}. In light of the result of
\Subsectionref{asymp_examples}, where we establish a gap between the
communication and randomness requirements of perfectly secure computation
and asymptotically secure computation, it is plausible that there are
generic protocols for asymptotically secure computation with lower
communication and randomness requirements than possible for standard secure
computation.

We presented a new information inequality for 3-user interactive
protocols, which was instrumental in obtaining our strongest bounds. This
inequality requires users to have independent inputs in the beginning. It would
be interesting to generalize this to settings where the inputs may be
dependent. More generally, proving information inequalities for interactive
protocols in larger networks is also of independent interest and might prove
useful in establishing strong communication lower bounds in multiuser setting.

Two other directions we leave as important open directions are to develop
communication and randomness lower bounds for secure multiparty computation
involving more than 3 parties, and to obtain stronger lower bounds for
security against active corruption than in the honest-but-curious setting
(when computation is feasible in both models; indeed, it is well-known that
general secure computation against active corruption is not possible when 1
out of 3 parties can be actively corrupted). There has been some prior work
in the first direction as mentioned in \Sectionref{intro}, these results
have been mostly only for the modular addition function.  While some of our
techniques can be extended to more than 3 parties, we will need entirely new
techniques for separating the communication requirements of the
honest-but-curious and active corruption settings.

\appendices

\section{Details Omitted from \Sectionref{ps_lowerbounds}}\label{app:proofs}
\begin{proof}[Proof of \Lemmaref{general_cutset}]
Fix a protocol $\Pi$. First we show $H(X|M_{12},M_{13})=0$. We apply a cut-set argument. 
Consider the cut isolating Alice from Bob \& Charlie.
We need to show that for every $m_{12},m_{31}$ with $p(m_{12},m_{31})>0$,
there is a (necessarily unique) $x\in {\mathcal X}$ such that
$p(x|m_{12},m_{31})=1$. Suppose, to the contrary, that we have a secure
protocol resulting in a p.m.f. $p(x,y,z,m_{12},m_{31})$ such that there
exists $x,x'\in{\mathcal X}$, $x\neq x'$, and $m_{12},m_{31}$ satisfying
$p(m_{12},m_{31})>0$, $p(x|m_{12},m_{31})>0$, and $p(x'|m_{12},m_{31})>0$. For these $x,x'$, 
since $(p_{XY},p_{Z|XY})$ is in the normal form, $\exists (y,z) \in 
{\mathcal{Y}\times\mathcal{Z}}$ such that $p_{XY}(x,y)>0, p_{XY}(x',y)>0$, and 
$p_{Z|X,Y}(z|x,y)\neq p_{Z|X,Y}(z|x',y)$.
\begin{enumerate}
\item[(i)] The definition of a protocol implies that $p(x,y,z,m_{12},m_{31})$ can be written 
as $p_{X,Y}(x,y)p(m_{12},m_{31}|x,y)p(z|m_{12},m_{31},y)$.
\item[(ii)] Privacy against Alice implies that $p(m_{12},m_{31}|x,y,z)=p(m_{12},m_{31}|x)$.
\item[(iii)] Using (ii) in (i), we get 
$p(x,y,z,m_{12},m_{31})=p_{X,Y}(x,y)p(m_{12},m_{31}|x)p(z|m_{12},m_{31},y)$.
\item[(iv)] Correctness and (ii) imply that we can also write 
$p(x,y,z,m_{12},m_{31})=p_{X,Y}(x,y)p_{Z|X,Y}(z|x,y)p(m_{12},m_{31}|x)$.
\item[(v)] Since $p_{X,Y}(x,y)p(m_{12},m_{31}|x)>0$, from (iii) and (iv), we get 
$p(z|m_{12},m_{31},y)=p_{Z|X,Y}(z|x,y)$.
\end{enumerate}
 Applying the above arguments to $(x',y,z,m_{12},m_{31})$ we get 
$p(z|m_{12},m_{31},y)=p_{Z|X,Y}(z|x',y)$, leading to the contradiction 
$p(z|m_{12},m_{31},y) \neq p(z|m_{12},m_{31},y)$, since by assumption 
$p_{Z|X,Y}(z|x,y) \neq p_{Z|X,Y}(z|x',y)$.

Similarly, by considering the cut separating Bob from Alice \& Charlie we can show \eqref{eq:general_cutset_bob}.

To show \eqref{eq:general_cutset_charlie}, i.e., $H(Z|M_{23},M_{31})=0$, we need to show that for every $m_{23},m_{31}$ with $p(m_{23},m_{31})>0$, there is a (necessarily unique) $z\in \mathcal{Z}$ such that $p(z|m_{23},m_{31})=1$. Suppose, to the contrary, that we have a secure protocol resulting in a p.m.f. $p(x,y,z,m_{23},m_{31})$ such that there exists $z,z'\in{\mathcal Z}$, $z\neq z'$ and $m_{23},m_{31}$
satisfying $p(m_{23},m_{31})>0$, $p(z|m_{23},m_{31})$, $p(z'|m_{23},m_{31})>0$. By the assumption that $(p_{XY},p_{Z|XY})$ is in normal form, there exists $(x,y)$ s.t. $p_{XY}(x,y)>0$ and $p_{Z|X,Y}(z|x,y)>0$.
\begin{enumerate}
\item[(i)] The definition of a protocol implies that $p(x,y,z,m_{23},m_{31})$ can be written as $p_{X,Y}(x,y)p(m_{23},m_{31}|x,y)p(z|m_{23},m_{31})$.
\item[(ii)] Privacy against Charlie implies that $p(x,y,z,m_{23},m_{31})$ can be written as $p_{X,Y}(x,y)p(z|x,y)p(m_{23},m_{31}|z)$.
\item[(iii)] (i), (ii), and correctness give $p(m_{23},m_{31}|x,y)p(z|m_{23},m_{31})=p_{Z|X,Y}(z|x,y)p(m_{23},m_{31}|z)$. 
\end{enumerate}
By assumption, $p(m_{23},m_{31})>0$ and $p(z|m_{23},m_{31})>0$, which imply that $p(m_{23},m_{31}|z)>0$. And since $p_{Z|X,Y}(z|x,y)>0$, we have from (iii) that $p(m_{23},m_{31}|x,y)>0$. Now consider $(x,y,z')$. 
By assumption, $p(m_{23},m_{31})>0$ and $p(z'|m_{23},m_{31})>0$, which imply $p(m_{23},m_{31}|z')>0$. 
Since $p(m_{23},m_{31}|x,y)>0$, running the same steps (i)-(iii) as above with $(x,y,z',m_{23},m_{31})$, (iii) implies that $p_{Z|X,Y}(z'|x,y)>0$. Define $\alpha \triangleq \frac{p(z|x,y)}{p(z'|x,y)}$. Since $(p_{XY},p_{Z|XY})$ is in normal form, $\exists (x',y')\in (\mathcal{X,Y})$ s.t. $p_{XY}(x',y')>0$ and $p_{Z|X,Y}(z|x',y') \neq \alpha \cdot p_{Z|X,Y}(z'|x',y')$. Since $\alpha \neq 0$, at least one of $p(z|x',y')$ or $p(z'|x',y')$ is non-zero. Assume that any one of these is non-zero, then applying the above arguments will give us that the other one should also be non-zero.
\begin{enumerate}
\item[(iv)] Repeating the steps (i)-(iii) with $(x,y,z',m_{23},m_{31})$ yields
$p(m_{23},m_{31}|x,y)p(z'|m_{23},m_{31})=p_{Z|X,Y}(z'|x,y)p(m_{23},m_{31}|z').$
\item[(v)] Dividing the expression in (iii) by the expression in (iv) gives $\frac{p(z|m_{23},m_{31})}{p(z'|m_{23},m_{31})} = \alpha \cdot \frac{p(m_{23},m_{31}|z)}{p(m_{23},m_{31}|z')}$.
\item[(vi)] Repeating (i)-(v) for $(x',y',z,m_{23},m_{31})$ and $(x',y',z',m_{23},m_{31})$, we get $\frac{p(z|m_{23},m_{31})}{p(z'|m_{23},m_{31})} \neq \alpha \cdot \frac{p(m_{23},m_{31}|z)}{p(m_{23},m_{31}|z')}$, which contradicts (v).
\end{enumerate}
\end{proof}
\section{Connections to Secure Sampling and Correlated Multi-Secret Sharing}\label{app:CMSS_sampling}
\paragraph{Secure Sampling.} In {\em secure sampling} functionalities, 
none of the users receives any input, but all three users produce outputs. 
The functionality is specified by a joint distribution $p_{XYZ}$, and the protocol for sampling $p_{XYZ}$ is specified by $\Pi(p_{XYZ})$. 
The correctness condition in this case is that the outputs of Alice, Bob, and Charlie are distributed according to $p_{XYZ}$. 
The security conditions remain the same as in the case of secure computation, that is, none of the users can infer anything about the other users' outputs other than what they can from their own outputs.

\paragraph{A Normal Form for $p_{XYZ}$.} For a joint distribution $p_{XYZ}$, define the relation $x\sim x'$ for $x,x'\in\X$ to hold if $\exists c\geq 0$ such that $\forall y\in\Y, z\in\Z$, $p(x',y,z)=c\cdot p(x,y,z)$. 
Similarly, we define $y\sim y'$ for $y,y'\in\Y$ and $z\sim z'$ for $z,z'\in\Z$. 
We say that $p_{XYZ}$ is in the {\em normal form} if $x\sim x' \Rightarrow x=x'$, $y\sim y' \Rightarrow y=y'$, and $z\sim z' \Rightarrow z=z'$.

It is easy to see that one can transform any distribution $p_{XYZ}$ to one in normal form $p_{X'Y'Z'}$, with possibly smaller alphabets, so that any secure sampling protocol for the former can be transformed to one for the latter with the same communication costs, and vice versa. 
To define $X'$, $X$ is modified by removing all $x$ such that $p(x) = 0$ and then replacing all $x$ in an equivalence class of $\sim$ with a single representative; $Y'$ and $Z'$ are defined similarly. The modification to the protocol, in either direction, is for each user to locally map $X$ to $X'$ etc., or vice versa. Hence it is enough to study the communication complexity of securely sampling distributions in the normal form.

Now, we show an analog of \Lemmaref{general_cutset} for secure sampling protocols.
\begin{lem}\label{lem:cutset_samp}
Suppose $p_{XYZ}$ is in normal form. Then, in any secure sampling protocol $\Pi(p_{XYZ})$,  the cut isolating Alice from Bob and Charlie must determine Alice's output $X$, i.e., $H(X| M_{12},M_{31}) =0$. Similarly, $H(Y|M_{12},M_{23})=0$ and $H(Z|M_{23},M_{31})=0$.
\end{lem}
\begin{proof}
We only prove $H(X|M_{12},M_{31})=0$; the other ones, i.e., $H(Y|M_{12},M_{23})=0$ and $H(Z|M_{23},M_{31})=0$ can be proved similarly. 
We need to show that for every $m_{12},m_{31}$ with $p(m_{12},m_{31})>0$, 
there is a (necessarily unique) $x\in {\mathcal X}$ such that $p(x|m_{12},m_{31})=1$. 
Suppose, to the contrary, that we have a secure sampling protocol resulting in a p.m.f. 
$p(x,y,z,m_{12},m_{31})$ such that there exists $x,x'\in{\mathcal X}$, $x\neq x'$ and $m_{12},m_{31}$
satisfying $p(m_{12},m_{31})>0$, $p(x|m_{12},m_{31})>0$, and $p(x'|m_{12},m_{31})>0$. Since $p(m_{12},m_{31})>0$ and $p(x|m_{12},m_{31})>0$ imply $p_X(x)>0$, there exists $(y,z)$ s.t. $p_{XYZ}(x,y,z)>0$.
\begin{enumerate}
\item[(i)] The definition of a protocol implies that $p(x,y,z,m_{12},m_{31})$ can be written as $p_{YZ}(y,z)p(m_{12},m_{31}|y,z)p(x|m_{12},m_{31})$.
\item[(ii)] Privacy against Alice implies that $p(x,y,z,m_{12},m_{31})$ can be written as $p_{XYZ}(x,y,z)p(m_{12},m_{31}|x)$.
\item[(iii)] (i) and (ii) gives $p_{YZ}(y,z)p(m_{12},m_{31}|y,z)p(x|m_{12},m_{31})$ = $p_{XYZ}(x,y,z)p(m_{12},m_{31}|x)$. 
\end{enumerate}
By assumption, $p(m_{12},m_{31})>0$ and $p(x|m_{12},m_{31})>0$, which imply that $p(m_{12},m_{31}|x)>0$. And since $p_{XYZ}(x,y,z)>0$, we have from (iii) that $p(m_{12},m_{31}|y,z)>0$. Now consider $(x',y,z,m_{12},m_{31})$. By assumption, $p(m_{12},m_{31})>0$ and $p(x'|m_{12},m_{31})>0$, which imply $p(m_{12},m_{31}|x')>0$. Since $p(m_{12},m_{31}|y,z)>0$ from above, (iii) implies that $p_{XYZ}(x',y,z)>0$. Define $\alpha \triangleq \frac{p(x,y,z)}{p(x',y,z)}$. Since $p_{XYZ}$ is in normal form, $\exists (y',z')\in (\mathcal{Y,Z})$ s.t. $p_{XYZ}(x,y',z') \neq \alpha \cdot p_{XYZ}(x',y',z')$. Since $\alpha \neq 0$, at least one of $p(x,y',z')$ or $p(x',y',z')$ is non-zero. Assume that any one of these is non-zero, then applying the above arguments will give us that the other one should also be non-zero.
\begin{enumerate}
\item[(iv)] Repeating the steps (i)-(iii) with $(x',y,z,m_{12},m_{31})$ yields $p_{YZ}(y,z)p(m_{12},m_{31}|y,z)p(x'|m_{12},m_{31})$ = $p_{XYZ}(x',y,z)p(m_{12},m_{31}|x')$.
\item[(v)] Dividing the expression in (iii) by the expression in (iv) gives $\frac{p(x|m_{12},m_{31})}{p(x'|m_{12},m_{31})} = \alpha \cdot \frac{p(m_{12},m_{31}|x)}{p(m_{12},m_{31}|x')}$.
\item[(vi)] Repeating (i)-(v) for $(x,y',z',m_{12},m_{31})$ and $(x',y',z',m_{12},m_{31})$, we get $\frac{p(x|m_{12},m_{31})}{p(x'|m_{12},m_{31})} \neq \alpha \cdot \frac{p(m_{12},m_{31}|x)}{p(m_{12},m_{31}|x')}$, which contradicts (v).
\end{enumerate}
\end{proof}

\begin{thm}\label{thm:samp_lbs}
Any secure sampling protocol $\Pi(p_{XYZ})$, where $p_{XYZ}$ is in normal form, should satisfy the following lower bounds on the entropy of the transcripts on each link.
\begin{align*}
H(M_{23}) &\geq RI(X;Z) + RI(X;Y) + H(Y,Z|X), \\
H(M_{31}) &\geq RI(Y;Z) + RI(X;Y) + H(X,Z|Y), \\
H(M_{12}) &\geq RI(X;Z) + RI(Y;Z) + H(X,Y|Z).
\end{align*}
\end{thm}
\begin{proof}
From \Lemmaref{cutset_samp} we have $H(X|M_{12},M_{31})=0$, $H(Y|M_{12},M_{23})=0$, and $H(Z|M_{23},M_{31})=0$. 
Note that we can apply \Lemmaref{infoineq} for secure sampling of {\em dependent} $X$, $Y$, and $Z$, 
because, in the beginning users only have independent randomness, but no inputs. 
In the end, they output from a joint distribution $p_{XYZ}$, where $X$, $Y$ and $Z$ may be dependent, 
but this does not affect the requirements of \Lemmaref{infoineq} in any way. The proof for $H(M_{23})$ is given below; the other two bounds follows similarly.
\begin{align*}
&H(M_{31}) = I(M_{12};M_{31}) + H(M_{31}|M_{12}) \\
&= I(M_{12};M_{31}) + I(M_{31};M_{23}|M_{12}) + H(M_{31}|M_{12},M_{23})\\ 
&\stackrel{\text{(a)}}{\ge} I(M_{12};M_{31}|M_{23}) + I(M_{31};M_{23}|M_{12}) \\
&\qquad \qquad + H(M_{31}|M_{12},M_{23})\\ 
&\stackrel{\text{(b)}}{\ge} RI(Y;Z) + RI(X;Y) + H(X,Z | Y),
\end{align*}
where (a) used $I(M_{12};M_{31}) \ge I(M_{12};M_{31}|M_{23})$, which follows from \Lemmaref{infoineq}; (b) used $I(M_{12};M_{31}|M_{23}) \geq RI(Y;Z)$, $I(M_{31};M_{23}|M_{12}) \geq RI(X;Y)$, and $H(M_{31}|M_{12},M_{23}) \geq H(X,Z | Y)$, all of which we have shown in the proof of \Theoremref{prelim_lbs}.
\end{proof}

We remark that if the marginal distributions satisfy $p_{XY} = p_X p_Y$ (i.e., $X$ and $Y$ are independent), then a secure computation protocol for $p_{Z|XY}$ can be turned into a secure sampling protocol (with the same communication costs), by having Alice and Bob locally sample inputs $X$ and $Y$ according to $p_X$ and $p_Y$ and then run the computation protocol. So, whenever $X$ and $Y$ are independent, the lower bounds on communication for secure sampling imply lower bounds for secure computation.

\paragraph{Correlated Multi-Secret Sharing Schemes.}
We define a notion of secret-sharing, called Correlated Multi-Secret Sharing (CMSS) that is closely related to secure sampling/computation problem. We will show that lower bounds on the entropy of shares of such secret-sharing schemes will also be lower bounds on entropy of transcripts for the corresponding secure computation protocols. However, we shall show a separation between the efficiency of secret-sharing (where there is an omniscient dealer) and a protocol, using the stronger lower bounds we have established in \Subsectionref{improved_lbs}.
\begin{defn}
Given a graph $G=(V,E)$, an adversary structure $\mathcal{A} \subseteq 2^V$, and a joint distribution $p_{(X_v)_{v\in V}}$ over random variables $X_v$ indexed by $v\in V$, a correlated multiple secret sharing scheme for $(G,p_{(X_v)_{v\in V}})$ defines a distribution $p_{(M_e)_{e\in E}|(X_v)_{v\in V}}$ of shares $M_e$ for each edge $e\in E$, such that the following hold. Below, for $S\subseteq E$, $M_S$ stands for the collection of all $M_e$ for $e\in
S$; similarly $X_T$ is defined for $T\subseteq V$; $E_v\subseteq E$ denotes the set of edges incident on a vertex $V$.
\begin{itemize}
\item Correctness: For all $v\in V$, $H(X_v|M_{E_v})=0$.
\item Privacy: For every set $T\in \mathcal{A}$, let $E_T=\cup_{v\in T} E_v$; then, $I(X_{\overline{T}};M_{E_T}|X_T)= 0$.
\end{itemize}
\end{defn}
Below we give a specialised version of the above general definition which is suitable to our setting, where $G$ is the clique over the vertex set $V=\{1,2,3\}$, and $\mathcal{A}=\{ \{1\}, \{2\}, \{3\} \}$ (corresponding to 1-privacy).

We define $\Sigma$ to be a {\em correlated multi-secret sharing scheme for a joint distribution $p_{XYZ}$} (with respect to our fixed adversary structures) if it probabilistically maps secrets $(X,Y,Z)$ to shares $M_{12},M_{23},M_{31}$ such that the following conditions hold:
\begin{itemize}
\item Correctness: $H(X|M_{12},M_{31}) = H(Y|M_{12},M_{23}) = H(Z|M_{23},M_{31}) = 0$.
\item Privacy: 
\begin{align*}
I(M_{12},M_{31};Y,Z|X) &=0 \quad \text{(privacy against Alice)}, \\
I(M_{12},M_{23};X,Z|Y) &=0 \quad \text{(privacy against Bob)}, \\
I(M_{23},M_{31};X,Y|Z) &=0 \quad \text{(privacy against Charlie)}.
\end{align*}

\end{itemize}
We point out that while the correctness condition relates only to the supports of $X$, $Y$, and $Z$ individually, the privacy condition is crucially influenced by the joint distribution.

\begin{thm}\label{thm:CMSS_prelim_lbs}
Any CMSS scheme for any joint distribution $p_{XYZ}$ satisfies
\begin{align*}
H(M_{12}) &\geq \max\{RI(X;Z), RI(Y;Z)\} + H(X,Y|Z), \\
H(M_{23}) &\geq \max\{RI(X;Z), RI(X;Y)\} + H(Y,Z|X), \\
H(M_{31}) &\geq \max\{RI(Y;Z), RI(X;Y)\} + H(X,Z|Y).
\end{align*}
\end{thm}
\begin{proof}
We proceed along the lines of the proof of \Theoremref{prelim_lbs}, except that here we do not need \Lemmaref{general_cutset} to argue that $H(X|M_{12},M_{31})$ = $H(Y|M_{12},M_{23})$ = $H(Z|M_{23},M_{31}) = 0$, instead, these follow from the correctness of CMSS.
\end{proof}

If $p_{XYZ}=p_{XY}p_{Z|XY}$, where $p_{XY}$ has full support and  $p_{Z|XY}$ is in normal form, using \Lemmaref{general_cutset}, the bounds in \Theoremref{CMSS_prelim_lbs} imply bounds in \Theoremref{prelim_lbs}. If $p_{XYZ}$ has full support, then we can further strengthen the bounds in \Theoremref{CMSS_prelim_lbs} by applying distribution switching.

\begin{thm}\label{thm:CMSS_lbs}
Consider any CMSS scheme for a joint distribution $p_{XYZ}$, where $p_{XYZ}$ has full support.
\begin{align*}
H(M_{12}) \geq \max\left\{\begin{array}{l}\displaystyle \max_{p_{X'Y'Z'}}RI(X';Z')+H(X',Y'|Z'), \\
\displaystyle \max_{p_{X'Y'Z'}}RI(Y';Z') + H(X',Y'|Z')\end{array}\right\},
\end{align*}
where $p_{X'Y'Z'}$ is any distribution for which the characteristic bipartite graph of $p_{X'Y'}$ is connected.
\begin{align*}
H(M_{23}) \geq \max\left\{\begin{array}{l}\displaystyle \max_{p_{X'Y'Z'}}RI(X';Z')+H(Y',Z'|X'), \\
\displaystyle \max_{p_{X'Y'Z'}}RI(X';Y') + H(Y',Z'|X')\end{array}\right\},
\end{align*}
where $p_{X'Y'Z'}$ is any distribution for which the characteristic bipartite graph of $p_{Y'Z'}$ is connected.
\begin{align*}
H(M_{31}) \geq \max\left\{\begin{array}{l}\displaystyle \max_{p_{X'Y'Z'}}RI(Y';Z')+H(X',Z'|Y'), \\
\displaystyle \max_{p_{X'Y'Z'}}RI(X';Y') + H(X',Z'|Y')\end{array}\right\},
\end{align*}
where $p_{X'Y'Z'}$ is any distribution for which the characteristic bipartite graph of $p_{X'Z'}$ is connected.

\end{thm}
\begin{proof}
First we observe that we can apply distribution switching to CMSS schemes also, i.e., if we have a CMSS $\Sigma(p_{XYZ})$, where $p_{XYZ}$ has full support, it will remain a CMSS if we change the distribution to a different one $p_{X'Y'Z'}$. This follows from the correctness and privacy conditions of a CMSS. Proceeding as in the proof of \Lemmaref{XYZ_inde_M123}, we can show that for any CMSS $\Sigma(p_{XYZ})$, connectedness of the characteristic bipartite graph of $p_{XY}$ implies $I(X,Y,Z;M_{12})=0$. The other two, i.e., connectedness of the characteristic bipartite graph of $p_{XZ}$ implies $I(X,Y,Z;M_{31})=0$,
 and connectedness of the characteristic bipartite graph of $p_{YZ}$ implies $I(X,Y,Z;M_{23})=0$, follow similarly. Now, we can apply the distribution switching to the bounds in \Theoremref{CMSS_prelim_lbs}.
\end{proof}

It is easy to see that any secure sampling protocol $\Pi(p_{XYZ})$, where $p_{XYZ}$ is in normal form, yields a CMSS scheme for
the same joint distribution $p_{XYZ}$: An omniscient dealer can always
produce the shares $M_{12},M_{23},M_{31}$ which are precisely the
transcripts produced by the secure sampling protocol. Now, correctness for
this CMSS follows from \Lemmaref{cutset_samp}, and
privacy of CMSS scheme follows from the privacy of the secure sampling
protocol. Thus the lower bounds on the transcripts produced by a CMSS
scheme for a given $p_{XYZ}$ in normal form, gives lower bounds on the corresponding links for any secure sampling protocol for this $p_{XYZ}$. Furthermore, if $p_{XYZ}=p_{XY}p_{Y|XY}$, where $p_{XY}$ has full support and $p_{Z|XY}$ is in normal form, then lower bounds for CMSS schemes provide lower bounds for secure
computation problems. As we discuss in \Pageref{thm:and}, these lower
bounds are not tight in general for secure computation, i.e., there is a function (in fact the {\sc
and} function) for which there is a CMSS scheme which requires less
communication than what our lower bounds for secure computation for that
function provide. Towards this, here we give upper bounds on the share
sizes of a 3-user CMSS for {\sc and}, which is defined as $X$ and $Y$
independent and uniformly distributed bits, and $Z=X\land Y$.

\begin{thm}\label{thm:gap_CMSS}
For $p_{XYZ}$ such that $X$ and $Y$ independent and uniformly distributed
bits, and $Z=X\land Y$, there is a CMSS $\Sigma(p_{XYZ})$ which has
$H(M_{12})=H(M_{23})=H(M_{31})=\log(3)$.
\end{thm}
\begin{proof}
Consider a CMSS scheme $\Sigma$ defined as follows.
Let $(\alpha,\beta,\gamma)$ be a random permutation of the set
$\{0,1,2\}$. Let $M_{12}=\alpha$ and
\begin{align*}
M_{31}=\begin{cases}
\alpha \qquad \text{if }X=1, \\
\beta  \qquad \text{if }X=0,
\end{cases}
&
\qquad
M_{23}=\begin{cases}
\alpha \qquad \text{if }Y=1, \\
\gamma  \qquad \text{if }Y=0.
\end{cases}
\end{align*}
It can be seen that this scheme satisfies the correctness and privacy
requirements (in particular, 
$(M_{12},M_{31})$ is uniformly random, conditioned on
$M_{12}=M_{31}$ when $X=1$ and conditioned on
$M_{12}\not=M_{31}$  when $X=0$).
$H(M^\Sigma_{12})=H(M^\Sigma_{23})=H(M^\Sigma_{31})=\log 3 < 1.585$.

\noindent \Theoremref{CMSS_lbs} implies that this scheme is optimal.
\end{proof}

\section{Details Omitted from \Sectionref{asymp_lowerbounds}}\label{app:asymp_proofs}
\begin{proof}[Proof of \Lemmaref{two-user_lower-bound}] We define the following function:
\begin{align}
R_f(\delta,D) := \displaystyle \min_{\substack{p_{U|XY}: \\ I(U;Y|X)\leq \delta \\ \exists g: \mathbb{E}[d_H(f(X,Y),g(U,Y))] \leq D}} I(U;X|Y), \label{eq:approx_orlitsky-roche}
\end{align}
where $d_H$ is the Hamming distortion function.
Note that $R_f^{\WZ}(D)=R_f(0,D)$. Now define the rate-region tradeoff corresponding to \eqref{eq:approx_orlitsky-roche} as follows:
\begin{align}
\mathcal{T}_{R_f}(X,Y) &= \{(R_1,R_2,D): \exists p_{U|XY} \text{ and a function } g, \nonumber \\
& \text{ for which } I(U;X|Y)\leq R_1, I(U;Y|X) \leq R_2, \nonumber \\
& \text{ and } \mathbb{E}[d_H(f(X,Y),g(U,Y))]\leq D \}. \label{eq:region_orlitsky-roche}
\end{align}
\begin{lem}\label{lem:properties_OR-region}
$\mathcal{T}_{R_f}(X,Y)$, as defined in \eqref{eq:region_orlitsky-roche}, is a closed and convex set. Hence, $R_f(\delta,D)$ is convex in $(\delta,D)$.
\end{lem}
\begin{proof}
\noindent {\bf Closedness:} Let $\mathcal{P}_{XY}$ denote the set of all conditional p.m.f.'s $p_{U|XY}$. Since $\mathcal{X}$ and $\mathcal{Y}$ are finite alphabets, it follows from the Fenchel-Eggleston's strengthening of Carath$\acute{\text{e}}$odory's theorem \cite[pg. 310]{CsiszarKorner81}, that we can restrict the alphabet size of $U$ s.t. $|\mathcal{U}| \leq |\mathcal{X}|\cdot|\mathcal{Y}|+2$. This implies that $\mathcal{P}_{XY}$ is a compact set (since it is closed and bounded).
For a fixed $p_{XY}$, consider the following function:
\begin{align}
m(p_{U|XY}) &= (I(U;X|Y),I(U;Y|X),\nonumber \\
&\qquad\min_g \mathbb{E}[d_H(f(X,Y),g(U,Y)]) \label{eq:continuity_func}.
\end{align}
Note that $\mathcal{T}_{R_f}(X,Y)$ is the increasing hull of $\image(m)$ -- image of the function $m$ -- where increasing hull of a set $S\subseteq \mathbb{R}^3$ is defined as $\{(a,b,c) \in \mathbb{R}^3: \exists (a',b',c') \in S \text{ s.t. } a' \leq a, b' \leq b, \text{ and } c' \leq c\}$. Since the increasing hull of a closed set is always closed, it is enough to show that $\image(m)$ is closed. We show, below, that $m$ is continuous in $p_{U|XY}\in\mathcal{P}_{XY}$; this proves closedness of the set $\image(m)$, since image of a compact set under a continuous function is always compact -- and therefore closed.

In order to show that $m(p_{U|XY})$ is continuous in $p_{U|XY}$, we need to show that $I(U;X|Y)$, $I(U;Y|X)$, and $\min_g \mathbb{E}[d_H(f(X,Y),g(U,Y)]$ are continuous in $p_{U|XY}\in\mathcal{P}_{XY}$. It is well known that conditional mutual information is a continuous function of the distribution. To show that $\min_g \mathbb{E}[d_H(f(X,Y),g(U,Y)]$ is continuous in $p_{U|XY}$, it is sufficient to show that $\mathbb{E}[d_H(f(X,Y),g(U,Y)]$ is continuous for every choice of $g$. This is because there are only finitely many functions $g:\mathcal{U}\times\mathcal{Y}\to\mathcal{Z}$, $\min$ is a continuous function, and composition of two continuous functions is continuous. 

Consider $p_{U|XY},p_{U'|XY}\in\mathcal{P}_{XY}$ such that $\sum_{x,y,u} p_{XY}(x,y)|p_{U|XY}(u|x,y)-p_{U'|XY}(u|x,y)| \leq \gamma$. Since the value of Hamming distortion function $d_H$ is at most 1, it can be easily seen that for every function $g$,  $|\mathbb{E}[d_H(f(X,Y),g(U,Y)]-\mathbb{E}[d_H(f(X,Y),g(U',Y)]| \leq \gamma$. Hence $\mathbb{E}[d_H(f(X,Y),g(U,Y)]$ is continuous in $p_{U|XY}$ for every $g$.

This proves the closedness of $\mathcal{T}_{R_f}(X,Y)$, which justifies taking $\min$, instead of $\inf$, in the definition of $R_f(\delta,D)$ in \eqref{eq:approx_orlitsky-roche}.

\noindent {\bf Convexity:} Let $(R_1^{(0)},R_2^{(0)},D^{(0)})$, $(R_1^{(1)},R_2^{(1)},D^{(1)}) \in \mathcal{T}_{R_f}(X,Y)$, and let $(p_{U_0|XY},g_0)$ and $(p_{U_1|XY},g_1)$ be such that $I(U_i;X|Y)\leq R_1^{(i)}, I(U_i;Y|X) \leq R_2^{(i)}, \text{ and } \mathbb{E}[d_H(f(X,Y),g_i(U_i,Y))]\leq D^{(i)}$, for $i=0,1$. In order to prove the convexity of $\mathcal{T}_{R_f}(X,Y)$, we need to show that $\forall \alpha \in [0,1]$, $(R_1,R_2,D):=(\alpha R_1^{(0)} + (1-\alpha)R_1^{(1)},\alpha R_2^{(0)} + (1-\alpha)R_2^{(1)}, \alpha D^{(0)} + (1-\alpha)D^{(1)})\in\mathcal{T}_{R_f}(X,Y)$. For a given $\alpha\in[0,1]$, we show, below, that the distribution defined by random variable $U=(\Phi,U_{\Phi})$, where $\Phi\sim$ Bern($\alpha$) and independent of $(X,Y)$, and function $g$ defined by $g((\phi,u_{\phi}),y)=g_{\phi}(u_{\phi},y)$ imply $(R_1,R_2,D)\in\mathcal{T}_{R_f}(X,Y)$.
\begin{align*}
I(U;X|Y) &= I(\Phi,U_{\Phi};X|Y) = I(U_{\Phi};X|Y,\Phi) \\
&= \alpha I(U_1;X|Y) + (1-\alpha) I(U_2;X|Y) \\
&\leq \alpha R_1^{(0)} + (1-\alpha)R_1^{(1)}.
\end{align*}
Similarly we can show $I(U;Y|X) \leq \alpha R_2^{(0)} + (1-\alpha)R_2^{(1)}$. For the third quantity:
\begin{align*}
\mathbb{E}[d_H(f(X,Y),g(U,Y)] & = \alpha \mathbb{E}[d_H(f(X,Y),g_0(U_0,Y)]\\
&\; + (1-\alpha) \mathbb{E}[d_H(f(X,Y),g_1(U_1,Y)] \\
&\leq \alpha D^{(0)} + (1-\alpha) D^{(1)}.
\end{align*}
\end{proof}

\begin{lem}\label{lem:continuity_orlitsky-roche}
For a fixed pair $(f,p_{XY})$, the function $R_f(\delta,D)$, defined in \eqref{eq:approx_orlitsky-roche}, is right continuous at $(\delta,D)=(0,0)$.
\end{lem}
\begin{proof}
This is proved using the property that the region $\mathcal{T}_{R_f}(X,Y)$ is closed. From the definition of $R_f(\delta,D)$ in \eqref{eq:approx_orlitsky-roche}, it is easy to see that it is a non-increasing function of $(\delta,D)$. Now suppose, to the contrary, that $R_f(\delta,D)$ is not right continuous at $(\delta,D)=(0,0)$. This implies that there exists a monotone decreasing sequence $(\delta_m,D_m)\downarrow 0$ (i.e., $\delta_m \geq \delta_{m+1}$ and $D_m \geq D_{m+1}$, $\forall m\in \mathbb{N}$, and $\delta_m \downarrow 0$, $D_m \downarrow 0$) and $\gamma >0$ s.t. $R_f(\delta_m,D_m) \leq R_f(0,0)-\gamma$ for all $m\in\mathbb{N}$. As observed earlier, $R_f(\delta_m,D_m)$ is a monotone non-decreasing sequence that is bounded above by $R_f(0,0)$, which implies that it is convergent (since every monotone non-decreasing sequence that is bounded above is convergent). Let $L=\lim_{m\to\infty}R_f(\delta_m,D_m)$ be the limit of this sequence. We have $L\leq R_f(0,0)-\gamma < R_f(0,0)$. This contradicts the fact that $R_f(0,0)$ is the minimum value $r$ s.t. $(r,0,0)\in \mathcal{T}_{R_f}(X,Y)$, because $\mathcal{T}_{R_f}(X,Y)$ is closed (i.e., $\mathcal{T}_{R_f}(X,Y)$ contains all its limit points), implying that $(L,0,0)\in \mathcal{T}_{R_f}(X,Y)$.
\end{proof}

Now we are ready to lower-bound $I(X^n;M_{12},M_{31}|Y^n)$. For simplicity of notation, define $M:=(M_{12},M_{31})$. Note that $\hat{Z}^n$ depends on $M_{23}$ in the original problem of \Figureref{asymp_setup}. The transcript $M_{23}$ and $\hat{Z}^n$ can be sampled, conditioned on $(Y^n,M)$, by combined Bob-Charlie using additional private randomness $\Theta$, which is independent of $(X^n,Y^n,M)$. So, we can assume that the $i$-th symbol $\hat{Z}_i$ of the output is determined by a function $g_i(M,Y^n,\Theta)$. Let $U_i:=(M,Y^{i-1},Y_{i+1}^n,\Theta)$, then $\hat{Z}_i=g_i(U_i,Y_i)$.
\begin{align}
&I(X^n;M|Y^n) \nonumber \\
&= \sum_{i=1}^{n} H(X_i|Y_i) - H(X_i|Y^n,X^{i-1},M) \quad ((X_i,Y_i)\text{'s are i.i.d.}) \nonumber \\
&\stackrel{\text{(a)}}{\geq} \sum_{i=1}^{n} H(X_i|Y_i) - H(X_i|Y^n,M,\Theta) \nonumber \\
&= \sum_{i=1}^{n} I(U_i;X_i|Y_i) \quad (\text{where }U_i=(M,Y^{i-1},Y_{i+1}^n,\Theta))\\
&\stackrel{\text{(b)}}{\geq} \sum_{i=1}^{n} R_f\Big(I(U_i;Y_i|X_i), \mathbb{E}[d_H(f(X_i,Y_i),g_i(U_i,Y_i))]\Big) \nonumber \\
&\stackrel{\text{(c)}}{\geq} nR_f\Big(\frac{1}{n}\sum_{i=1}^{n}I(U_i;Y_i|X_i), \frac{1}{n}\sum_{i=1}^{n}\mathbb{E}[d_H(f(X_i,Y_i),g_i(U_i,Y_i))]\Big) \label{eq:last_app_cutset}
\end{align}
(a) follows from independence of $\Theta$ and $(M,X^n,Y^n)$, and the fact that conditioning reduces entropy; (b) follows by the definition of $R_f$ in \eqref{eq:approx_orlitsky-roche}; (c) follows from the convexity of $R_f(\delta,D)$, proved in \Lemmaref{properties_OR-region}. Now, we bound both the arguments of $R_f$ in \eqref{eq:last_app_cutset}; we use privacy against Alice \eqref{eq:privacy_alice} for the first argument and correctness condition \eqref{eq:correct_cond} for the second argument.
\begin{align}
&\sum_{i=1}^n I(U_i;Y_i|X_i) = \sum_{i=1}^n I(M,Y^{i-1},Y_{i+1}^n,\Theta;Y_i|X_i) \nonumber \\
&\stackrel{\text{(d)}}{=}\sum_{i=1}^n I(M,Y^{i-1},Y_{i+1}^n;Y_i|X_i) \nonumber \\
&= \sum_{i=1}^n \underbrace{H(Y_i|X_i)}_{=\ H(Y_i|X^n,Y^{i-1})} - H(Y_i|X_i,M,Y^{i-1},Y_{i+1}^n) \nonumber \\
&\leq \sum_{i=1}^n H(Y_i|X^n,Y^{i-1}) - H(Y_i|X^n,M,Y^{i-1},Y_{i+1}^n) \nonumber \\
&= \sum_{i=1}^n I(Y_i;M|X^n,Y^{i-1}) + \sum_{i=1}^n I(Y_i;Y_{i+1}^n|X^n,M,Y^{i-1}) \nonumber \\
&\leq \underbrace{I(Y^n;M|X^n)}_{\leq \ \epsilon, \text{ by }\eqref{eq:privacy_alice}} + \sum_{i=1}^n I(Y_i;M,Y^{i-1},Y_{i+1}^n|X^n) \nonumber \\
&\leq \sum_{i=1}^n \underbrace{I(Y_i;Y^{i-1},Y_{i+1}^n|X^n)}_{=\ 0} + I(Y_i;M|X^n,Y^{i-1},Y_{i+1}^n) + \epsilon \nonumber \\
&= \sum_{i=1}^n \underbrace{I(Y^n;M|X^n)}_{\leq \ \epsilon, \text{ by }\eqref{eq:privacy_alice}} - \underbrace{I(Y^{i-1},Y_{i+1}^n;M|X^n)}_{{\geq \ 0}} + \epsilon \nonumber \\
&\leq (n+1)\epsilon \nonumber \\
&\leq 2n\epsilon \label{eq:first_small_ubs},
\end{align}
where (d) follows from independence of $\Theta$ and $(M,X^n,Y^n)$. For the second argument of \eqref{eq:last_app_cutset}:
\begin{align}
&\sum_{i=1}^{n}\mathbb{E}[d_H(f(X_i,Y_i),g_i(U_i,Y_i))] \nonumber \\
&\qquad= \sum_{i=1}^{n}\mathbb{E}[d_H(f(X_i,Y_i),\hat{Z}_i)] \quad (\text{where } \hat{Z}_i=g_i(U_i,Y_i)) \nonumber \\
&\qquad= \sum_{i=1}^{n} \Pr[\hat{Z}_i \neq f(X_i,Y_i)] \nonumber \\
&\qquad \leq \sum_{i=1}^{n} \Pr[\hat{Z}^n \neq Z^n] \leq \sum_{i=1}^{n} \epsilon = n\epsilon \label{eq:second_small_ubs}.
\end{align}
\noindent Now we can complete the proof by using \eqref{eq:first_small_ubs}-\eqref{eq:second_small_ubs} in \eqref{eq:last_app_cutset}:
\begin{align*}
I(X^n;M|Y^n) &\stackrel{(e)}{\geq} nR_f(2\epsilon,\epsilon) \\
&\stackrel{\text{(f)}}{\geq} n(R_f(0,0) - \delta_{\epsilon}) \quad (\text{where }\delta_{\epsilon}\to 0 \text{ as } \epsilon\to0) \\
&= n(R_f^{\WZ}(0) - \delta_{\epsilon}),
\end{align*}
where (e) uses the fact that $(R_f(\delta,D)$ is non-increasing in $(\delta,D))$, and (f) follows from \Lemmaref{continuity_orlitsky-roche}.
\end{proof}

\begin{proof}[Proof of \Lemmaref{data-processing}]
We prove only the first inequality of \Lemmaref{data-processing}, and as stated there, we use only the weak privacy conditions -- where \eqref{eq:privacy_alice}-\eqref{eq:privacy_charlie} are upper-bounded by $n\epsilon$ -- to prove this. The other two inequalities can be proved similarly. Let $M_1:=(M_{12},M_{31})$ and $M_3:=(M_{23},M_{31})$. For $(X^n,Z^n)$, we define the function $\frac{1}{n}RI_{n\epsilon}(X^n;Z^n)$ as follows:
\begin{align} 
\frac{1}{n}RI_{n\epsilon}(X^n;Z^n) := \displaystyle \min_{\substack{p_{Q|X^nZ^n}: \\ \frac{1}{n}I(Q;Z^n|X^n)\leq \epsilon \\ \frac{1}{n}I(Q;X^n|Z^n)\leq \epsilon}} \frac{1}{n}I(X^n;Z^n|Q). \label{eq:approx_vector-ri}
\end{align}
For $(X,Z)$, we define the function $RI_{\epsilon}(X;Z)$ as follows:
\begin{align}
RI_{\epsilon}(X;Z) := \displaystyle \min_{\substack{p_{Q'|XZ}: \\ I(Q';Z|X)\leq \epsilon \\ I(Q';X|Z)\leq \epsilon}} I(X;Z|Q'). \label{eq:approx_ri}
\end{align}
Note that $RI_0(X;Z) = RI(X;Z)$. We prove the result by proving the following three inequalities:
\begin{align}
\frac{1}{n}RI(M_1,X^n;M_3,Z^n) &\geq \frac{1}{n}RI_{n\epsilon}(X^n;Z^n) \label{eq:first_interim_data-processing}\\
&\geq RI_{\epsilon}(X;Z) \label{eq:second_interim_data-processing}\\
&\geq RI(X;Z) - \epsilon_{31}. \label{eq:third_interim_data-processing}
\end{align}
For \eqref{eq:first_interim_data-processing}, we proceed as follows:
\begin{itemize}
\item $I(Q;M_3,Z^n|M_1,X^n)=0$, together with weak privacy against Alice ($I(M_1;Y^n,Z^n|X^n)\leq n\epsilon$) implies $I(Q;Z^n|X^n) \leq n\epsilon$.
{\allowdisplaybreaks
\begin{align*}
0 &= I(Q;M_3,Z^n|M_1,X^n) \\
&\geq I(Q;Z^n|M_1,X^n) \\
&= I(Q,M_1;Z^n|X^n) - \underbrace{I(M_1;Z^n|X^n)}_{\leq \ n\epsilon} \\
&\geq I(Q;Z^n|X^n) - n\epsilon
\end{align*}
}
\item Similarly, it can be shown that $I(Q;M_1,X^n|M_3,Z^n)=0$ and weak privacy against Charlie ($I(M_3;X^n,Y^n|Z^n)\leq n\epsilon$) imply $I(Q;X^n|Z^n) \leq n\epsilon$.
\end{itemize}
Now, $I((M_1,X^n);(M_3,Z^n)|Q) \geq I(X^n;Z^n|Q)$ (which is always true), together with the above two implications, implies \eqref{eq:first_interim_data-processing}.
For \eqref{eq:second_interim_data-processing} we proceed as follows:
{\allowdisplaybreaks
\begin{align*}
\epsilon &\geq \frac{1}{n}I(Q;Z^n|X^n) \\
&\stackrel{\text{(a)}}{=} \frac{1}{n}\sum_{i=1}^n H(Z_i|X_i) - H(Z_i|X^n,Z^{i-1},Q) \\
&\geq \frac{1}{n}\sum_{i=1}^n H(Z_i|X_i) - H(Z_i|X^{i-1},X_i,Z^{i-1},Q) \\
&= \frac{1}{n}\sum_{i=1}^n I(Q,X^{i-1},Z^{i-1};Z_i|X_i) \\
&\stackrel{\text{(b)}}{=} \sum_{i=1}^n p_T(i)I(Q,X^{i-1},Z^{i-1};Z_i|X_i,T=i) \\
&= I(Q,X^{T-1},Z^{T-1};Z_T|X_T,T) \\
&= I(Q,X^{T-1},Z^{T-1},T;Z_T|X_T) \\
&= I(Q_T,T;Z_T|X_T) \quad (\text{where }Q_T=(Q,X^{T-1},Z^{T-1}))
\end{align*}
(a) follows because $(X_i,Z_i)$'s are i.i.d; the random variable $T$ in (b) is distributed with $\text{Unif}\{1,2,\hdots,n\}$ and is independent of $(Q,X^n,Z^n)$. For the other constraint: $\frac{1}{n}I(Q;X^n|Z^n) \leq \epsilon \implies I(Q_T,T;X^n|Z^n) \leq \epsilon$, we can proceed similarly as above. For the objective function:
}
\begin{align*}
\frac{1}{n}I(X^n;Z^n|Q) &\geq \sum_{i=1}^n \frac{1}{n} I(X_i;Z_i|Q,X^{i-1},Z^{i-1}) \\
&= I(X_T;Z_T|Q_T,T).
\end{align*}
So, we get the following:
\begin{align} 
\frac{1}{n}RI_{n\epsilon}(X^n;Z^n) \geq \displaystyle \min_{\substack{p_{Q|X^nZ^n}: \\ I(Q_T,T;Z_T|X_T)\leq \epsilon \\ I(Q_T,T;X_T|Z_T)\leq \epsilon}} I(X_T;Z_T|Q_T,T),
\end{align}
where, on the RHS, $T\sim \text{Unif}\{1,2,\hdots,n\}$ and is independent of $(Q,X^n,Z^n)$, and $Q_T=(Q,X^{T-1},Z^{T-1})$. To get \eqref{eq:second_interim_data-processing}, we define $p_{Q'|XZ} := p_{Q_TT|X_TZ_T}$ to get
\begin{align*}
\displaystyle \min_{\substack{p_{Q|X^nZ^n}: \\ I(Q_T,T;Z_T|X_T)\leq \epsilon \\ I(Q_T,T;X_T|Z_T)\leq \epsilon}} I(X_T;Z_T|Q_T,T) \geq \displaystyle \min_{\substack{p_{Q'|XZ}: \\ I(Q';Z|X)\leq \epsilon \\ I(Q';X|Z)\leq \epsilon}} I(X;Z|Q'),
\end{align*}
which proves \eqref{eq:second_interim_data-processing}.

For \eqref{eq:third_interim_data-processing}, we prove that for fixed joint distribution $p_{XZ}$, $RI_{\epsilon}(X;Z)$ is right continuous at $\epsilon=0$. This is proved, below, using the property that the tension region $\mathfrak{T}(X;Z)$ is closed \cite{PrabhakaranPr14}. For simplicity of notation, we denote $RI_{\epsilon}(X;Z)$ by $RI_{\epsilon}$. Note that $RI_{0} = RI$.

From the definition of $RI_{\epsilon}$ in \eqref{eq:approx_ri}, it is easy to see that it is a non-increasing function of $\epsilon$, that is to say, if $\epsilon < \epsilon'$, then $RI_{\epsilon} \geq RI_{\epsilon'}$. Now suppose, to the contrary, that $RI_{\epsilon}$ is not right continuous at $\epsilon=0$. This implies that there exists a monotone decreasing sequence $\epsilon_m \downarrow 0$ and $\gamma > 0$ s.t. $RI_{\epsilon_m} \leq RI_0 - \gamma$ for all $m\in\mathbb{N}$. Note that $RI_{\epsilon_m}$ is a monotone non-decreasing sequence that is bounded above by $RI_0$, which implies that it is convergent (since every monotone non-decreasing sequence that is bounded above is convergent). Let $L=\lim_{m\to\infty} RI_{\epsilon_m}$ be the limit of this sequence. We have $L\leq RI_0-\gamma < RI_0$. This contradicts the fact that $RI_{0}$ is the minimum value $r$ s.t. $(0,0,r) \in \mathfrak{T}(X;Z)$, because $\mathfrak{T}(X;Z)$ is closed (i.e., $\mathfrak{T}(X;Z)$ contains all its limit points), implying that $(0,0,L) \in \mathfrak{T}(X;Z)$.
\end{proof}

\section*{Acknowledgements}
We would like to gratefully acknowledge several useful discussions with Bikash Dey and Manoj Mishra. The optimal protocol in \Figureref{addition}, which builds on K\"orner and Marton's~\cite{KornerMa79} scheme, that we used to show a gap between the rate regions of perfectly and asymptotically secure computations in \Subsectionref{asymp_examples}.1 is a joint work with them and has appeared in \cite{DataDeMiPr14}. We also gratefully acknowledge helpful comments from the anonymous referees of the CRYPTO 2014 conference.

\bibliographystyle{IEEEtran}
\bibliography{crypto,reference}

% Generated by IEEEtran.bst, version: 1.12 (2007/01/11)
\begin{thebibliography}{10}
\providecommand{\url}[1]{#1}
\csname url@samestyle\endcsname
\providecommand{\newblock}{\relax}
\providecommand{\bibinfo}[2]{#2}
\providecommand{\BIBentrySTDinterwordspacing}{\spaceskip=0pt\relax}
\providecommand{\BIBentryALTinterwordstretchfactor}{4}
\providecommand{\BIBentryALTinterwordspacing}{\spaceskip=\fontdimen2\font plus
\BIBentryALTinterwordstretchfactor\fontdimen3\font minus
  \fontdimen4\font\relax}
\providecommand{\BIBforeignlanguage}[2]{{%
\expandafter\ifx\csname l@#1\endcsname\relax
\typeout{** WARNING: IEEEtran.bst: No hyphenation pattern has been}%
\typeout{** loaded for the language `#1'. Using the pattern for}%
\typeout{** the default language instead.}%
\else
\language=\csname l@#1\endcsname
\fi
#2}}
\providecommand{\BIBdecl}{\relax}
\BIBdecl

\bibitem{FeigeKiNa94}
U.~Feige, J.~Kilian, and M.~Naor, ``A minimal model for secure computation
  (extended abstract),'' in \emph{STOC}.\hskip 1em plus 0.5em minus 0.4em\relax
  ACM, 1994, pp. 554--563.

\bibitem{CramerDaNi15}
R.~Cramer, I.~B. Damgard, and J.~B. Nielsen, \emph{Secure Multiparty
  Computation and Secret Sharing}, 1st~ed.\hskip 1em plus 0.5em minus
  0.4em\relax New York, NY, USA: Cambridge University Press, 2015.

\bibitem{ShamirRiAd81}
A.~Shamir, R.~L. Rivest, and L.~M. Adleman, ``Mental poker,'' in \emph{The
  Mathematical Gardner}, D.~Klarner, Ed.\hskip 1em plus 0.5em minus 0.4em\relax
  Belmont, California: Wadsworth, 1981, pp. 37--43, preliminary version as MIT
  TM-125, 1978.

\bibitem{Rabin81}
M.~Rabin, ``How to exchange secrets by oblivious transfer,'' Harvard Aiken
  Computation Laboratory, Tech. Rep. TR-81, 1981.

\bibitem{Blum81}
M.~Blum, ``Three applications of the oblivious transfer: Part {I}: Coin
  flipping by telephone; part {II}: How to exchange secrets; part {III}: How to
  send certified electronic mail,'' Technical report, University of California,
  Berkeley, 1981.

\bibitem{Yao82mpc}
A.~C. Yao, ``Protocols for secure computation,'' in \emph{Proc.\ $23$rd
  FOCS}.\hskip 1em plus 0.5em minus 0.4em\relax IEEE, 1982, pp. 160--164.

\bibitem{Yao86}
------, ``How to generate and exchange secrets,'' in \emph{Proc.\ $27$th
  FOCS}.\hskip 1em plus 0.5em minus 0.4em\relax IEEE, 1986, pp. 162--167.

\bibitem{GoldreichMiWi87}
O.~Goldreich, S.~Micali, and A.~Wigderson, ``How to play {ANY} mental game,''
  in \emph{Proc.\ $19$th STOC}, {ACM}, Ed.\hskip 1em plus 0.5em minus
  0.4em\relax ACM, 1987, pp. 218--229, see \cite[Chap.~7]{Goldreich04book} for
  more details.

\bibitem{BenorGoWi88}
M.~B{en-Or}, S.~Goldwasser, and A.~Wigderson, ``Completeness theorems for
  non-cryptographic fault-tolerant distributed computation,'' in \emph{Proc.\
  $20$th STOC}.\hskip 1em plus 0.5em minus 0.4em\relax ACM, 1988, pp. 1--10.

\bibitem{ChaumCrDa88}
D.~Chaum, C.~Cr{\'e}peau, and I.~Damg{\aa}rd, ``Multiparty unconditionally
  secure protocols,'' in \emph{Proc.\ $20$th STOC}.\hskip 1em plus 0.5em minus
  0.4em\relax ACM, 1988, pp. 11--19.

\bibitem{CrepeauKi88}
C.~Cr{\'e}peau and J.~Kilian, ``Achieving oblivious transfer using weakened
  security assumptions (extended abstract),'' in \emph{FOCS}.\hskip 1em plus
  0.5em minus 0.4em\relax IEEE, 1988, pp. 42--52.

\bibitem{Yao79}
A.~C.-C. Yao, ``Some complexity questions related to distributive computing
  (preliminary report),'' in \emph{STOC}.\hskip 1em plus 0.5em minus
  0.4em\relax ACM, 1979, pp. 209--213.

\bibitem{KushilevitzNi97book}
E.~Kushilevitz and N.~Nisan, \emph{Communication complexity}.\hskip 1em plus
  0.5em minus 0.4em\relax New York: Cambridge University Press, 1997.

\bibitem{ChakrabartiShWiYa01}
A.~Chakrabarti, Y.~Shi, A.~Wirth, and A.~C.-C. Yao, ``Informational complexity
  and the direct sum problem for simultaneous message complexity,'' in
  \emph{FOCS}.\hskip 1em plus 0.5em minus 0.4em\relax IEEE, 2001, pp. 270--278.

\bibitem{OrlitskyRoche}
A.~Orlitsky and J.~R. Roche, ``Coding for computing,'' \emph{Information
  Theory, IEEE Transactions on}, vol.~47, no.~3, pp. 903--917, 2001.

\bibitem{MaIs13}
N.~Ma and P.~Ishwar, ``The infinite-message limit of two-terminal interactive
  source coding,'' \emph{Information Theory, IEEE Transactions on}, vol.~59,
  no.~7, pp. 4071--4094, 2013.

\bibitem{KushilevitzRo98}
E.~Kushilevitz and A.~Ros{\'{e}}n, ``A randomness-rounds tradeoff in private
  computation,'' \emph{{SIAM} J. Discrete Math.}, vol.~11, no.~1, pp. 61--80,
  1998.

\bibitem{Kushilevitz92}
E.~Kushilevitz, ``Privacy and communication complexity,'' \emph{SIAM J.
  Discrete Math.}, vol.~5, no.~2, pp. 273--284, 1992.

\bibitem{FranklinYu92}
M.~K. Franklin and M.~Yung, ``Communication complexity of secure computation
  (extended abstract),'' in \emph{STOC}.\hskip 1em plus 0.5em minus 0.4em\relax
  ACM, 1992, pp. 699--710.

\bibitem{ChorKu93}
B.~Chor and E.~Kushilevitz, ``A communication-privacy tradeoff for modular
  addition.'' \emph{Inf. Process. Lett.}, vol.~45, no.~4, pp. 205--210, 1993.

\bibitem{KushilevitzMa97}
E.~Kushilevitz and Y.~Mansour, ``Randomness in private computations,''
  \emph{SIAM J. Discrete Math.}, vol.~10, no.~4, pp. 647--661, 1997.

\bibitem{BlundoSaPeVa99}
C.~Blundo, A.~D. Santis, G.~Persiano, and U.~Vaccaro, ``Randomness complexity
  of private computation,'' \emph{Computational Complexity}, vol.~8, no.~2, pp.
  145--168, 1999.

\bibitem{GalRo05}
A.~G{\'a}l and A.~Ros{\'e}n, ``Omega(log n) lower bounds on the amount of
  randomness in 2-private computation,'' \emph{SIAM J. Comput.}, vol.~34,
  no.~4, pp. 946--959, 2005.

\bibitem{IshaiKu04}
Y.~Ishai and E.~Kushilevitz, ``On the hardness of information-theoretic
  multiparty computation,'' in \emph{EUROCRYPT}, 2004, pp. 439--455.

\bibitem{DamgardIs06}
I.~Damg{\aa}rd and Y.~Ishai, ``Scalable secure multiparty computation,'' in
  \emph{CRYPTO}, ser. Lecture Notes in Computer Science, C.~Dwork, Ed., vol.
  4117.\hskip 1em plus 0.5em minus 0.4em\relax Springer, 2006, pp. 501--520.

\bibitem{KushilevitzOsRo96}
E.~Kushilevitz, R.~Ostrovsky, and A.~Ros{\'{e}}n, ``Characterizing linear size
  circuits in terms of privacy,'' in \emph{STOC}, 1996, pp. 541--550.

\bibitem{MaIs11}
N.~Ma and P.~Ishwar, ``Some results on distributed source coding for
  interactive function computation,'' \emph{IEEE Trans. Inform. Theory},
  vol.~57, no.~9, pp. 6180--6195, 2011.

\bibitem{MaIsGu12}
N.~Ma, P.~Ishwar, and P.~Gupta, ``Interactive source coding for function
  computation in collocated networks,'' \emph{IEEE Trans. Inform. Theory},
  vol.~58, no.~7, pp. 4289--4305, 2012.

\bibitem{LeeAb14}
E.~J. Lee and E.~Abbe, ``Two {S}hannon-type problems on secure multi-party
  computations,'' in \emph{Proc. 52nd Annual Allerton Conference on
  Communication, Control, and Computing}, 2014.

\bibitem{KerenidisLLRX12}
I.~Kerenidis, S.~Laplante, V.~Lerays, J.~Roland, and D.~Xiao, ``Lower bounds on
  information complexity via zero-communication protocols and applications,''
  in \emph{FOCS}, 2012, pp. 500--509.

\bibitem{MaurerWo03}
U.~M. Maurer and S.~Wolf, ``Secret-key agreement over unauthenticated public
  channels iii: Privacy amplification,'' \emph{IEEE Transactions on Information
  Theory}, vol.~49, no.~4, pp. 839--851, 2003.

\bibitem{DodisMi99}
Y.~Dodis and S.~Micali, ``Lower bounds for oblivious transfer reductions,'' in
  \emph{EUROCRYPT}, ser. Lecture Notes in Computer Science, J.~Stern, Ed., vol.
  1592.\hskip 1em plus 0.5em minus 0.4em\relax Springer, 1999, pp. 42--55.

\bibitem{BeimelOr11}
A.~Beimel and I.~Orlov, ``Secret sharing and non-{S}hannon information
  inequalities,'' \emph{IEEE Transactions on Information Theory}, vol.~57,
  no.~9, pp. 5634--5649, 2011.

\bibitem{BlundoSaCrGaVa94}
C.~Blundo, A.~D. Santis, G.~D. Crescenzo, A.~G. Gaggia, and U.~Vaccaro,
  ``Multi-secret sharing schemes,'' in \emph{CRYPTO}, ser. Lecture Notes in
  Computer Science, Y.~Desmedt, Ed., vol. 839.\hskip 1em plus 0.5em minus
  0.4em\relax Springer, 1994, pp. 150--163.

\bibitem{WolfWu08}
S.~Wolf and J.~Wullschleger, ``New monotones and lower bounds in unconditional
  two-party computation,'' \emph{IEEE Transactions on Information Theory},
  vol.~54, no.~6, pp. 2792--2797, 2008.

\bibitem{PrabhakaranPr14}
V.~M. Prabhakaran and M.~M. Prabhakaran, ``Assisted common information with an
  application to secure two-party sampling,'' \emph{IEEE Trans. Inform.
  Theory}, vol.~60, no.~6, pp. 3413 -- 3434, 2014.

\bibitem{CoverThomas06}
T.~M. Cover and J.~Thomas, \emph{Elements of Information Theory}.\hskip 1em
  plus 0.5em minus 0.4em\relax New York: Wiley, 2006.

\bibitem{Witsenhausen76}
H.~S. Witsenhausen, ``The zero-error side information problem and chromatic
  numbers (corresp.),'' \emph{{IEEE} Transactions on Information Theory},
  vol.~22, no.~5, pp. 592--593, 1976.

\bibitem{GacsKorner}
P.~G{\'a}cs and J.~K{\"o}rner, ``Common information is far less than mutual
  information,'' \emph{Problems of Control and Information Theory}, vol.~2,
  no.~2, pp. 149--162, 1973.

\bibitem{DodisMi00}
Y.~Dodis and S.~Micali, ``Parallel reducibility for information-theoretically
  secure computation.'' in \emph{CRYPTO}, ser. Lecture Notes in Computer
  Science, M.~Bellare, Ed., vol. 1880.\hskip 1em plus 0.5em minus 0.4em\relax
  Springer, 2000, pp. 74--92.

\bibitem{WinklerWu10}
S.~Winkler and J.~Wullschleger, ``On the efficiency of classical and quantum
  oblivious transfer reductions,'' in \emph{CRYPTO}, 2010, pp. 707--723, full
  version available at \url{http://arxiv.org/abs/1205.5136}.

\bibitem{AlonOr96}
N.~Alon and A.~Orlitsky, ``Source coding and graph entropies,'' \emph{{IEEE}
  Transactions on Information Theory}, vol.~42, no.~5, pp. 1329--1339, 1996.

\bibitem{Wiesner83}
S.~Wiesner, ``Conjugate coding,'' \emph{ACM Sigact News}, vol.~15, no.~1, pp.
  78--88, 1983.

\bibitem{EvenGL85}
S.~Even, O.~Goldreich, and A.~Lempel, ``A randomized protocol for signing
  contracts,'' \emph{Communications of the ACM}, vol.~28, no.~6, pp. 637--647,
  1985.

\bibitem{CsiszarKo78}
I.~Csisz{\'a}r and J.~K{\"o}rner, ``Broadcast channels with confidential
  messages,'' \emph{IEEE Transactions on Information Theory}, vol.~24, no.~3,
  pp. 339--348, 1978.

\bibitem{yamamoto82}
H.~Yamamoto, ``Wyner-{Z}iv theory for a general function of the correlated
  sources,'' \emph{{IEEE} Transactions on Information Theory}, vol.~28, no.~5,
  pp. 803--807, 1982.

\bibitem{WynerZi76}
A.~D. Wyner and J.~Ziv, ``The rate-distortion function for source coding with
  side information at the decoder,'' \emph{IEEE Trans. Inform. Theory},
  vol.~22, no.~1, pp. 1--10, 1976.

\bibitem{ElgamalKim11}
A.~E. Gamal and Y.-H. Kim, \emph{Network Information Theory}.\hskip 1em plus
  0.5em minus 0.4em\relax Cambridge, U.K.: Cambridge University Press, 2011.

\bibitem{Elias55}
P.~Elias, ``Coding for noisy channels,'' \emph{IRE Convention Record, {Part
  4}}, pp. 37--46, 1955.

\bibitem{KornerMa79}
J.~K{\"{o}}rner and K.~Marton, ``How to encode the modulo-two sum of binary
  sources (corresp.),'' \emph{{IEEE} Trans. Inform. Theory}, vol.~25, no.~2,
  pp. 219--221, 1979.

\bibitem{AroraBarak09}
S.~Arora and B.~Barak, \emph{Computational complexity: a modern
  approach}.\hskip 1em plus 0.5em minus 0.4em\relax Cambridge University Press,
  2009.

\bibitem{BeimelIsKuKu14}
A.~Beimel, Y.~Ishai, R.~Kumaresan, and E.~Kushilevitz, ``On the cryptographic
  complexity of the worst functions,'' in \emph{Theory of Cryptography}.\hskip
  1em plus 0.5em minus 0.4em\relax Springer Berlin Heidelberg, 2014, vol. 8349,
  pp. 317--342.

\bibitem{CsiszarKorner81}
I.~Csisz{\'a}r and J.~K{\"o}rner, \emph{Information Theory: Coding Theorems for
  Discrete Memoryless Systems, {1st ed.}}\hskip 1em plus 0.5em minus
  0.4em\relax Budapest, Hungary: Akad{\'e}miai Kiad{\'o}, 1981.

\bibitem{DataDeMiPr14}
D.~Data, B.~K. Dey, M.~Mishra, and V.~M. Prabhakaran, ``How to securely compute
  the modulo-two sum of binary sources,'' in \emph{{IEEE ITW}}, 2014, pp.
  496--500.

\bibitem{Goldreich04book}
O.~Goldreich, \emph{Foundations of Cryptography: Basic Applications}.\hskip 1em
  plus 0.5em minus 0.4em\relax Cambridge University Press, 2004.

\bibitem{stoc88}
\emph{Proc.\ $20$th STOC}.\hskip 1em plus 0.5em minus 0.4em\relax ACM, 1988.

\end{thebibliography}

\end{document}